\documentclass[12pt]{amsart}

\usepackage[stable]{footmisc}
\usepackage[latin1]{inputenc}
\usepackage{graphicx}
\usepackage{epstopdf}
\usepackage{enumitem}
\usepackage{cite}
\usepackage{amssymb,amsmath}
\usepackage{amsthm}
\usepackage{multirow}
\usepackage{hyperref}
\usepackage[margin=3cm]{geometry}
\usepackage{verbatim}
\usepackage[noend]{algpseudocode}
\usepackage{chngcntr}
\usepackage{booktabs}
\usepackage{multirow}
\usepackage{mathtools}
\usepackage{algorithm}
\usepackage{etex}
\usepackage{subfig}
\usepackage{bold-extra}
\usepackage{cleveref}

\usepackage{tikz}
\usetikzlibrary{arrows,backgrounds,fit,positioning,calc}
\usepackage[all]{xy}

\newcommand{\K}{{\mathbb K}}
\newcommand{\F}{{\mathbb F}}

\newcommand{\N}{{\mathbb N}}

\newcommand{\elim}[2]{{\mathrm {elim}_{#1} (#2)}}

\newcommand{\ideal}[1]{{\langle #1\rangle }}
\newcommand{\I}{{\mathbf I}}
\newcommand{\V}{{\mathbf V}}
\newcommand{\W}{{\mathbf W}}
\newcommand{\Z}{{\mathbf Z}}
\newcommand{\T}{{\mathcal T}}
\DeclareMathOperator{\mvar}{mvar}
\DeclareMathOperator{\mdeg}{mdeg}
\DeclareMathOperator{\init}{init}
\DeclareMathOperator{\rank}{rk}
\DeclareMathOperator{\sat}{sat}
\DeclareMathOperator{\prem}{prem}

\renewcommand{\subset}{\subseteq}
\renewcommand{\supset}{\supseteq}

\newcommand{\myfont}[1]{\textsc{#1}}

\newtheorem{theorem}{Theorem}
\newtheorem{proposition}[theorem]{Proposition}

\newtheorem{lemma}[theorem]{Lemma}
\newtheorem*{problem}{Problem}

\newtheorem*{question}{Question}
\newtheorem*{notation}{Notation}
\theoremstyle{definition}
\newtheorem{definition}{Definition}[section]
\newtheorem{example}{Example}[section]
\theoremstyle{remark}
\newtheorem{remark}{Remark}[section]

\let\oldtabular\tabular
\renewcommand{\tabular}{\footnotesize\oldtabular}

\title
[Chordal networks of polynomial ideals]
{ Chordal networks of polynomial ideals}

\date{\today}

\author{Diego Cifuentes} 
\address{
Laboratory for Information and Decision Systems (LIDS), 
Massachusetts Institute of Technology, Cambridge MA 02139, USA}
\email{diegcif@mit.edu}

\author{Pablo A. Parrilo}
\address{
Laboratory for Information and Decision Systems (LIDS), 
Massachusetts Institute of Technology, Cambridge MA 02139, USA}
\email{parrilo@mit.edu}


\begin{document}

\begin{abstract}
We introduce a novel representation of structured polynomial ideals,
which we refer to as \emph{chordal networks}.  The sparsity structure
of a polynomial system is often described by a graph that captures the
interactions among the variables.  Chordal networks provide a
computationally convenient decomposition into
simpler (triangular) polynomial sets, while preserving the underlying
graphical structure. We show that many interesting families of
polynomial ideals admit compact chordal network representations (of
size linear in the number of variables), even though the number of
components is exponentially large. Chordal networks can be
computed for arbitrary polynomial systems using a refinement of the
chordal elimination algorithm from~\cite{Cifuentes2014}. Furthermore,
they can be effectively used to obtain several properties of the
variety, such as its dimension, cardinality, and equidimensional
components, as well as an efficient probabilistic test for radical
ideal membership.  We apply our methods to examples from algebraic
statistics and vector addition systems; for these instances,
algorithms based on chordal networks outperform existing techniques by
orders of magnitude.
\end{abstract}

\maketitle

\section{Introduction}

Systems of polynomial equations can be used to model a large variety of applications, and in most cases the resulting systems have a particular sparsity structure.
We describe this sparsity structure using a graph.
A natural question that arises is whether this graphical structure can be effectively used to solve the system.
In~\cite{Cifuentes2014} we introduced the \emph{chordal elimination} algorithm, an elimination method that always preserves the graphical structure of the system.
In this paper we refine this algorithm to compute a new representation of the polynomial system that we call a \emph{chordal network}.

Chordal networks attempt to fix an intrinsic issue of Gr\"obner bases: they destroy the graphical structure of the system~\cite[Ex~1.2]{Cifuentes2014}.
As a consequence, polynomial systems with simple structure may have overly complicated Gr\"obner bases (see \Cref{exmp:triangcycle}).
In contrast, chordal networks will always preserve the underlying chordal graph.
We remark that chordal graphs have been successfully used in several other areas, such as numerical linear algebra~\cite{rose1976algorithmic}, discrete and continuous optimization~\cite{bodlaender2008combinatorial,Vandenberghe2014}, graphical models~\cite{Lauritzen1988} and constraint satisfaction~\cite{Dechter2003}.

Chordal networks describe a decomposition of the polynomial ideal into simpler (triangular) polynomial sets.
This decomposition gives quite a rich description of the underlying variety.
In particular, chordal networks can be efficiently used to compute dimension, cardinality, equidimensional components and also to test radical ideal membership.
Remarkably, several families of polynomial ideals (with exponentially large Gr\"obner bases) admit a compact chordal network representation, of size proportional to the number of variables.
We will shortly present some motivational examples after setting up the main terminology.

Throughout this document we work in the polynomial ring $\K[X]=\K[x_0,x_1,\ldots,x_{n-1}]$ over some field $\K$.
We fix once and for all the ordering of the variables $x_0>x_1>\cdots>x_{n-1}$~\footnote{Observe that smaller indices correspond to larger variables.}. 
We consider a system of polynomials $F = \{f_1,f_2,\ldots,f_m\}$.
There is a natural graph $\mathcal{G}(F)$, with vertex set $X=\{x_0,\ldots,x_{n-1}\}$, that abstracts the \emph{sparsity structure} of $F$.
The graph is given by cliques: for each $f_i$ we add a clique in all its variables.
Equivalently, there is an edge between $x_i$ and $x_j$ if and only if there is some polynomial in $F$ that contains both variables. 
We will consider throughout the paper a chordal completion $G$ of the graph $\mathcal{G}(F)$, and we will assume that $x_0>\cdots>x_{n-1}$ is a perfect elimination ordering of~$G$ (see \Cref{defn:perfectelimination}).
%We will say that the polynomial system $F$ is \emph{supported} on the chordal graph~$G$.

\subsection*{Some motivating examples}

The notions of chordality and treewidth are ubiquitous in applied
mathematics and computer science.  In particular, several hard
combinatorial problems can be solved efficiently on graphs of small
treewidth by using some type of recursion (or dynamic
program)~\cite{bodlaender2008combinatorial}.  We will see that this
recursive nature is also present in several polynomial systems of
small treewidth.  
We now illustrate this with three simple examples.

\begin{example}[Coloring a cycle graph]\label{exmp:triangcycle}
Graph coloring is a classical NP-complete problem that can be solved efficiently on graphs of small treewidth.
We consider the cycle graph $C_n$ with vertices $0,1,\ldots,n-1$, whose treewidth is two.
Coloring $C_n$ is particularly simple by proceeding in a recursive manner: color vertex $n-1$ arbitrarily and then subsequently color vertex $i$ avoiding the color of $i+1$ and possibly $n-1$.

The $q$-coloring problem for a graph $\mathcal{G}=(V,E)$ can be encoded in a system of polynomial equations (see e.g.,~\cite{DeLoera2008}):
\begin{subequations}\label{eq:colorings}
\begin{align}
  x_i^q -1 &= 0 &i\in V\\
  x_i^{q-1}+x_i^{q-2}x_j+\cdots+x_ix_j^{q-2}+x_j^{q-1} &= 0 &ij\in E
\end{align}
\end{subequations}
Let $F_{n,q}$ denote such system of polynomials for the cycle graph $C_n$.
Given that coloring the cycle graph is so easy, it should be possible to solve these equations efficiently.
However, if we compute a Gr\"obner basis the result is not so simple.
In particular, for the case of $F_{9,3}$ one of these polynomials has $81$ terms (with both lex and grevlex order).
This is a consequence of the fact that Gr\"obner bases destroy the graphical structure of the equations.

\begin{figure}[htb]
    \centering
    \null\hfill
    \raisebox{.5\height}{
    \includegraphics[scale=0.32]{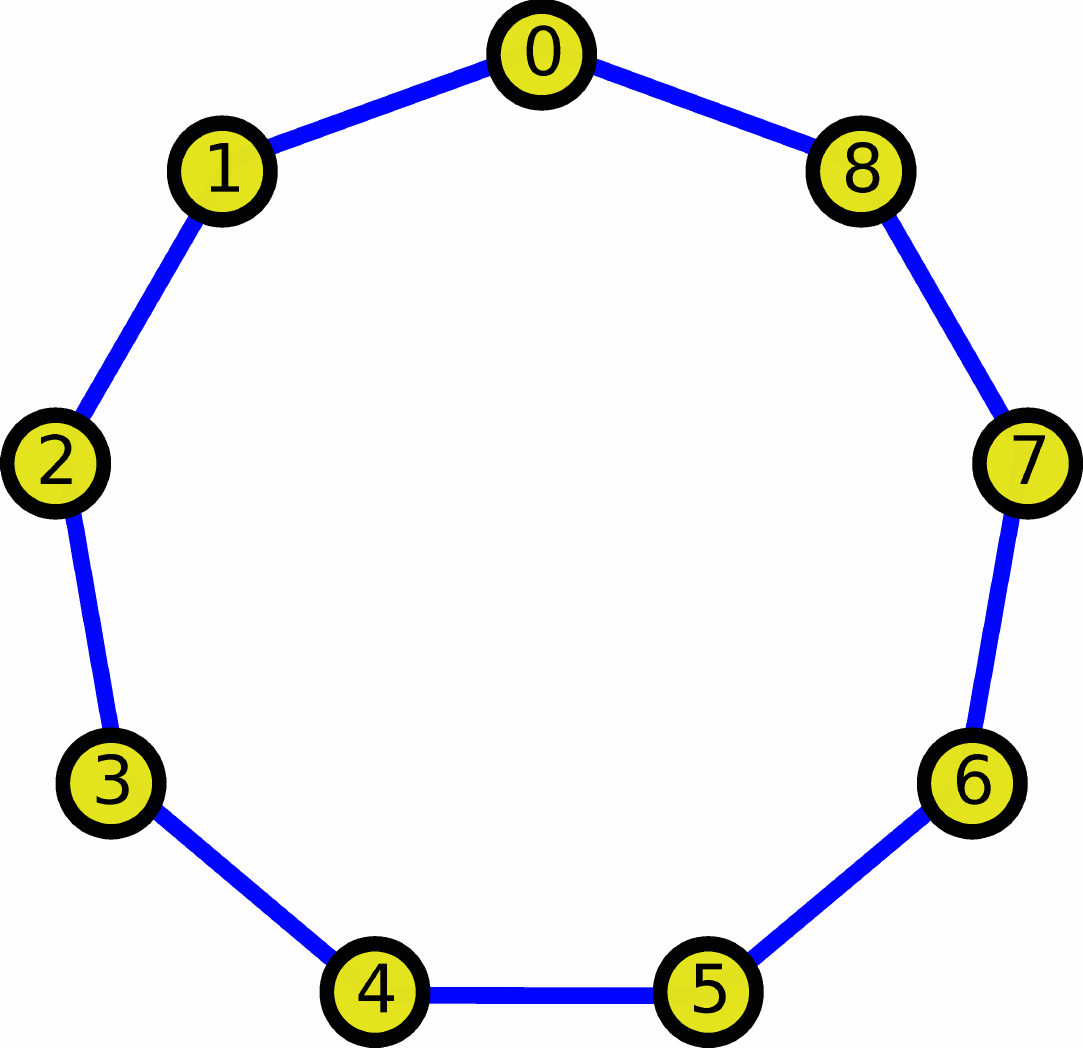}  
    }
    \hfill
    \begin{tikzpicture}[scale=1]
      \def\x{2.2}
      \def\y{.8}
      \tikzstyle{mybox} = [scale=0.65,draw=purple,fill=orange!8,thick,rectangle,rounded corners];
      \tikzstyle{myarrow} = [scale=0.65,purple,thick,->];
      \tikzstyle{mygroup} = [draw=blue,fill = blue!15, thick, dotted, minimum width = 220, minimum height = 16];
      \tikzstyle{mybox0} = [scale=0.65,fill=yellow!80,draw=blue,thick,densely dotted,circle];
      \tikzstyle{myarrow0} = [blue,thick,dotted,->];

      \node [mybox] (zero0) at (1*\x,0*\y) {$x_0^2 + x_0x_8 + x_8^2$};
      \node [mybox] (zero1) at (2*\x,0*\y) {$x_0 + x_1 + x_8$};

      \node [mybox] (one0) at (1*\x,-1*\y) {$x_1 - x_8$};
      \node [mybox] (one1) at (3*\x,-1*\y) {$x_1^2 + x_1x_8 + x_8^2$};
      \node [mybox] (one2) at (2*\x,-1*\y) {$x_1 + x_2 + x_8$};

      \node [mybox] (two0) at (1*\x,-2*\y) {$x_2^2 + x_2x_8 + x_8^2$};
      \node [mybox] (two1) at (2*\x,-2*\y) {$x_2 + x_3 + x_8$};
      \node [mybox] (two2) at (3*\x,-2*\y) {$x_2 - x_8$};

      \node [mybox] (three0) at (1*\x,-3*\y) {$x_3 - x_8$};
      \node [mybox] (three1) at (3*\x,-3*\y) {$x_3^2 + x_3x_8 + x_8^2$};
      \node [mybox] (three2) at (2*\x,-3*\y) {$x_3 + x_4 + x_8$};

      \node [mybox] (four0) at (1*\x,-4*\y) {$x_4^2 + x_4x_8 + x_8^2$};
      \node [mybox] (four1) at (2*\x,-4*\y) {$x_4 + x_5 + x_8$};
      \node [mybox] (four2) at (3*\x,-4*\y) {$x_4 - x_8$};

      \node [mybox] (five0) at (1*\x,-5*\y) {$x_5 - x_8$};
      \node [mybox] (five1) at (3*\x,-5*\y) {$x_5^2 + x_5x_8 + x_8^2$};
      \node [mybox] (five2) at (2*\x,-5*\y) {$x_5 + x_6 + x_8$};

      %\node [mybox] (six0) at (1*\x,-6*\y) {$x_6^2 + x_6x_8 + x_8^2$};
      \node [mybox] (six1) at (2*\x,-6*\y) {$x_6 + x_7 + x_8$};
      \node [mybox] (six2) at (3*\x,-6*\y) {$x_6 - x_8$};

      %\node [mybox] (seven0) at (1*\x,-7*\y) {$x_7 - x_8$};
      \node [mybox] (seven1) at (3*\x,-7*\y) {$x_7^2 + x_7x_8 + x_8^2$};
      %\node [mybox] (seven2) at (2*\x,-7*\y) {$x_7 + x_8 + x_9$};

      \node [mybox] (eight2) at (3*\x,-8*\y) {$x_8^3 - 1$};

      %\node [mybox] (eight0) at (1*\x,-8*\y) {$x_8^2 + x_8x_9 + x_9^2$};

      %\node [mybox] (nine0) at (1*\x,-9*\y) {$x_9^3 - 1$};

      \draw [myarrow] (zero0) to (one0);
      \draw [myarrow] (zero1) to (one1);
      \draw [myarrow] (zero1) to (one2);

      \draw [myarrow] (one0) to (two0);
      \draw [myarrow] (one0) to (two1);
      \draw [myarrow] (one1) to (two2);
      \draw [myarrow] (one2) to (two0);
      \draw [myarrow] (one2) to (two1);

      \draw [myarrow] (two0) to (three0);
      \draw [myarrow] (two1) to (three1);
      \draw [myarrow] (two1) to (three2);
      \draw [myarrow] (two2) to (three1);
      \draw [myarrow] (two2) to (three2);

      \draw [myarrow] (three0) to (four0);
      \draw [myarrow] (three0) to (four1);
      \draw [myarrow] (three1) to (four2);
      \draw [myarrow] (three2) to (four0);
      \draw [myarrow] (three2) to (four1);

      \draw [myarrow] (four0) to (five0);
      \draw [myarrow] (four1) to (five1);
      \draw [myarrow] (four1) to (five2);
      \draw [myarrow] (four2) to (five1);
      \draw [myarrow] (four2) to (five2);

      %\draw [myarrow] (five0) to (six0);
      \draw [myarrow] (five0) to (six1);
      \draw [myarrow] (five1) to (six2);
      %\draw [myarrow] (five2) to (six0);
      \draw [myarrow] (five2) to (six1);

      %\draw [myarrow] (six0) to (seven0);
      \draw [myarrow] (six1) to (seven1);
      \draw [myarrow] (six2) to (seven1);

      %\draw [myarrow] (seven0) to (eight0);
      \draw [myarrow] (seven1) to (eight2);

      %\draw [myarrow] (eight0) to (nine0);

      \begin{scope}[on background layer]
      \node [mygroup] (zero) at (2*\x,0*\y){}; 
      \node [mygroup] (one) at (2*\x,-1*\y){}; 
      \node [mygroup] (two) at (2*\x,-2*\y){}; 
      \node [mygroup] (three) at (2*\x,-3*\y){}; 
      \node [mygroup] (four) at (2*\x,-4*\y){}; 
      \node [mygroup] (five) at (2*\x,-5*\y){}; 
      \node [mygroup] (six) at (2*\x,-6*\y){}; 
      \node [mygroup] (seven) at (2*\x,-7*\y){}; 
      \node [mygroup] (eight) at (2*\x,-8*\y){}; 
      %\node [mygroup] (nine) at (2*\x,-9*\y){}; 
      \end{scope}

      \node [mybox0,left=.2 of zero] (zeroN) {$0$};
      \node [mybox0,left=.2 of one] (oneN) {$1$};
      \node [mybox0,left=.2 of two] (twoN) {$2$};
      \node [mybox0,left=.2 of three] (threeN) {$3$};
      \node [mybox0,left=.2 of four] (fourN) {$4$};
      \node [mybox0,left=.2 of five] (fiveN) {$5$};
      \node [mybox0,left=.2 of six] (sixN) {$6$};
      \node [mybox0,left=.2 of seven] (sevenN) {$7$};
      \node [mybox0,left=.2 of eight] (eightN) {$8$};
      %\node [mybox0,left=.2 of nine] (nineN) {$9$};
      \draw [myarrow0,bend right] (zeroN) to (oneN);
      \draw [myarrow0,bend right] (oneN) to (twoN);
      \draw [myarrow0,bend right] (twoN) to (threeN);
      \draw [myarrow0,bend right] (threeN) to (fourN);
      \draw [myarrow0,bend right] (fourN) to (fiveN);
      \draw [myarrow0,bend right] (fiveN) to (sixN);
      \draw [myarrow0,bend right] (sixN) to (sevenN);
      \draw [myarrow0,bend right] (sevenN) to (eightN);
      %\draw [myarrow0,bend right] (eightN) to (nineN);
    \end{tikzpicture}
    \hfill\null
    \caption{Chordal network for the $3$-chromatic ideal of a cycle}
    \label{fig:triangcycle}
\end{figure}
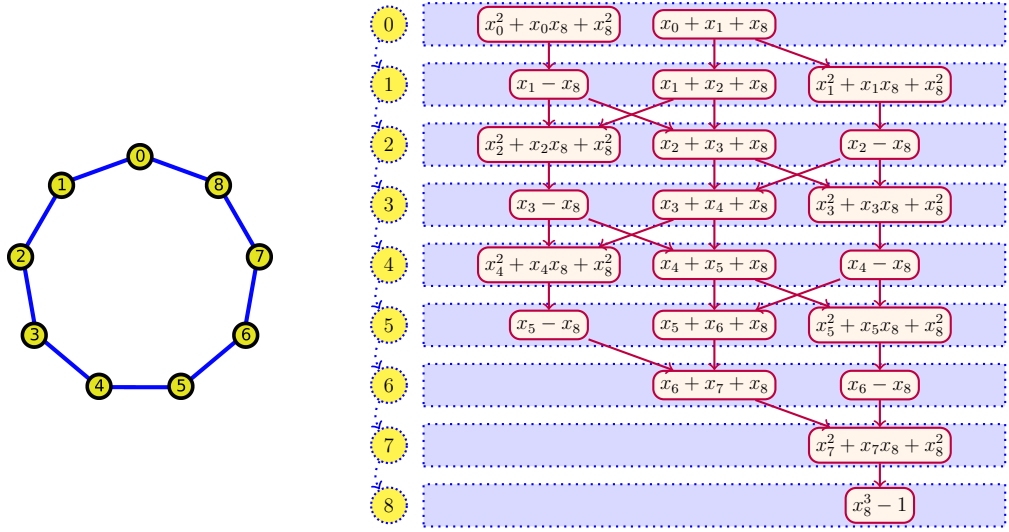

Nonetheless, one may hope to give a simple representation of the above polynomials that takes into account their recursive nature.
Indeed, a triangular decomposition of these equations is presented in \Cref{fig:triangcycle} for the case of $F_{9,3}$, and the pattern is very similar for arbitrary values of $n,q$.
The decomposition represented is:
\begin{align*}
  \V(F_{9,3}) = \bigcup_{T} \V(T)
\end{align*}
where the union is over all maximal directed paths in the diagram of \Cref{fig:triangcycle}. 
One path is
\begin{gather*}
T = \{
x_0+x_1+x_8,\,
x_1^2+x_1x_8+x_8^2,\,
x_2-x_8,\,
x_3^2+x_3x_8+x_8^2,\,
x_4-x_8,\\
x_5^2+x_5x_8+x_8^2,\,
x_6-x_8,\,
x_7^2+x_7x_8+x_8^2,\,
x_8^3-1
\}.
\end{gather*}
Recall that a set of polynomials is triangular if the largest variables of these polynomials are all distinct, and observe that all maximal paths $T$ are triangular.
Note that the total number of triangular sets is $21$, and in general we get the $(n-1)$-th Fibonacci number.
Even though the size of the triangular decomposition grows rapidly, it admits a very compact representation (linear in $n$) and the reason is precisely the recursive nature of the equations.
Indeed, the diagram of \Cref{fig:triangcycle} is constructed in a very similar way as we construct colorings: choose $x_8$ arbitrarily, then for each $x_i$ choose it based on the values of $x_{i+1}$ and $x_8$.
\end{example}

\begin{example}[Vertex covers of a tree]\label{exmp:triangnested}
We now consider the problem of finding minimum vertex coverings of a graph.
Recall that a subset~$S$ of vertices is a cover if any edge is incident to at least one element in~$S$.
%Computing a minimum vertex cover (or equivalently, a maximum independent set) is NP-complete.
Since the complement of a vertex cover is an independent set, computing a minimum vertex cover is NP-complete.
Nevertheless, when the graph is a tree the minimal vertex covers have a very special structure.
Indeed, we can construct such a cover recursively, starting from the root, as follows.
For the root node, we can decide whether to include it in the cover or not.
If we include it, we can delete the root and then recurse on each of its children.
Otherwise, we need to include in the cover all of its children, so we can delete them all, and then recurse.

  \begin{figure}[htb]
    \centering
    \null\hfill
    \raisebox{.2\height}{
    \includegraphics[scale=0.3]{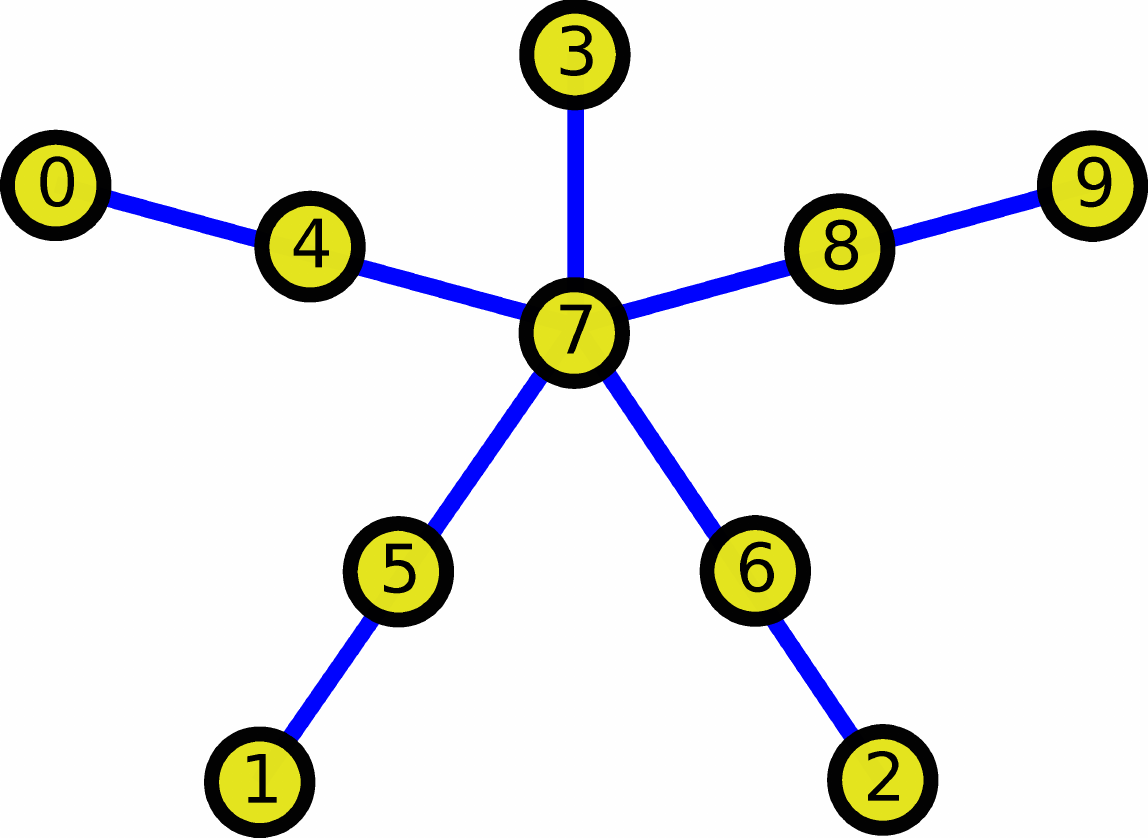}  
    }
    \hfill
    \begin{tikzpicture}[scale=1]
      \def\x{1.8}
      \def\y{.7}
      \def\xdelta{.15*\x}
      \def\xtree{-2.55*\x}
      \def\yzero{0*\y}
      \def\yone{\yzero-.25*\y}
      \def\ytwo{\yone-.25*\y}
      \def\ythree{\ytwo-.25*\y}
      \def\yfour{\ythree-.25*\y}
      \def\yfive{\yfour-.25*\y}
      \def\ysix{\yfive-.25*\y}
      \def\yseven{\ysix-1.4*\y}
      \def\yeight{\yseven-1*\y}
      \def\ynine{\yeight-1*\y}
      \tikzstyle{mybox} = [scale=0.7,draw=purple,fill=orange!8,thick,rectangle,rounded corners];
      \tikzstyle{myarrow} = [scale=0.65,purple,thick,->];
      \tikzstyle{mygroup} = [draw=blue,fill = blue!15, thick, dotted, minimum width = 40, minimum height = 15];
      \tikzstyle{mybox0} = [scale=0.6,fill=yellow!80,draw=blue,thick,densely dotted,circle];
      \tikzstyle{myarrow0} = [blue,thick,dotted,->];

      \node [mybox] (zero0) at (-.5*\x+\xdelta,\yzero) {$0$};
      \node [mybox] (zero1) at (-.5*\x-\xdelta,0*\yzero) {$x_0$};

      \node [mybox] (one0) at (.5*\x+\xdelta,-1*\yone) {$0$};
      \node [mybox] (one1) at (.5*\x-\xdelta,-1*\yone) {$x_1$};

      \node [mybox] (two0) at (1.5*\x+\xdelta,-1*\ytwo) {$0$};
      \node [mybox] (two1) at (1.5*\x-\xdelta,-1*\ytwo) {$x_2$};

      \node [mybox] (three0) at (-1.5*\x-\xdelta,-2*\ythree) {$0$};
      \node [mybox] (three1) at (-1.5*\x+\xdelta,-2*\ythree) {$x_3$};

      \node [mybox] (four0) at (-.5*\x-\xdelta,-3*\yfour) {$0$};
      \node [mybox] (four1) at (-.5*\x+\xdelta,-3*\yfour) {$x_4$};

      \node [mybox] (five0) at (.5*\x-\xdelta,-4*\yfive) {$0$};
      \node [mybox] (five1) at (.5*\x+\xdelta,-4*\yfive) {$x_5$};

      \node [mybox] (six0) at (1.5*\x-\xdelta,-4*\ysix) {$0$};
      \node [mybox] (six1) at (1.5*\x+\xdelta,-4*\ysix) {$x_6$};

      \node [mybox] (seven0) at (0*\x+1.5*\xdelta,-5*\yseven) {$0$};
      \node [mybox] (seven1) at (0*\x-1.5*\xdelta,-5*\yseven) {$x_7$};

      \node [mybox] (eight0) at (0*\x-\xdelta,-6*\yeight) {$0$};
      \node [mybox] (eight1) at (0*\x+\xdelta,-6*\yeight) {$x_8$};

      \node [mybox] (nine0) at (0*\x+\xdelta,-7*\ynine) {$0$};
      \node [mybox] (nine1) at (0*\x-\xdelta,-7*\ynine) {$x_9$};

      \draw [myarrow] (zero0) to (four1);
      \draw [myarrow] (zero1) to (four0);

      \draw [myarrow] (one0) to (five1);
      \draw [myarrow] (one1) to (five0);

      \draw [myarrow] (two0) to (six1);
      \draw [myarrow] (two1) to (six0);

      \draw [myarrow,out=-60,in=160] (three0) to (seven1);
      \draw [myarrow,out=-45,in=120] (three1) to (seven0);

      \draw [myarrow,out=-80,in=150] (four0) to (seven1);
      \draw [myarrow,out=-30,in=120] (four1) to (seven0);
      \draw [myarrow,out=-120,in=140] (four1) to (seven1);

      \draw [myarrow,out=-150,in=60] (five0) to (seven1);
      \draw [myarrow,out=-80,in=20] (five1) to (seven0);
      \draw [myarrow,out=-150,in=60] (five1) to (seven1);

      \draw [myarrow,out=-150,in=60] (six0) to (seven1);
      \draw [myarrow,out=-120,in=10] (six1) to (seven0);
      \draw [myarrow,out=-140,in=60] (six1) to (seven1);

      \draw [myarrow] (seven0) to (eight1);
      \draw [myarrow] (seven1) to (eight0);
      \draw [myarrow] (seven1) to (eight1);

      \draw [myarrow] (eight0) to (nine1);
      \draw [myarrow] (eight1) to (nine0);

      \begin{scope}[on background layer]
      \node [mygroup] (zero) at (-.5*\x,\yzero){}; 
      \node [mygroup] (one) at (.5*\x,\yone){}; 
      \node [mygroup] (two) at (1.5*\x,\ytwo){}; 
      \node [mygroup] (three) at (-1.5*\x,\ythree){}; 
      \node [mygroup] (four) at (-.5*\x,\yfour){}; 
      \node [mygroup] (five) at (.5*\x,\yfive){}; 
      \node [mygroup] (six) at (1.5*\x,\ysix){}; 
      \node [mygroup,minimum width=50] (seven) at (0*\x,\yseven){}; 
      \node [mygroup] (eight) at (0*\x,\yeight){}; 
      \node [mygroup] (nine) at (0*\x,\ynine){}; 
      \end{scope}

      \node [mybox0] (zeroN)  at (\xtree-.12*\x,\yzero)  {$0$};
      \node [mybox0] (oneN)   at (\xtree+.12*\x,\yone)   {$1$};
      \node [mybox0] (twoN)   at (\xtree+.36*\x,\ytwo)   {$2$};
      \node [mybox0] (threeN) at (\xtree-.36*\x,\ythree)   {$3$};
      \node [mybox0] (fourN)  at (\xtree-.12*\x,\yfour) {$4$};
      \node [mybox0] (fiveN)  at (\xtree+.12*\x,\yfive)  {$5$};
      \node [mybox0] (sixN)  at (\xtree+.36*\x,\ysix)  {$6$};
      \node [mybox0] (sevenN) at (\xtree,\yseven)  {$7$};
      \node [mybox0] (eightN) at (\xtree,\yeight)   {$8$};
      \node [mybox0] (nineN)  at (\xtree,\ynine) {$9$};
      \draw [myarrow0,out=-100,in=100] (zeroN) to (fourN);
      \draw [myarrow0,out=-70,in=70] (oneN) to (fiveN);
      \draw [myarrow0,out=-70,in=70] (twoN) to (sixN);
      \draw [myarrow0,bend right] (threeN) to (sevenN);
      \draw [myarrow0,out=-100,in=110] (fourN) to (sevenN);
      \draw [myarrow0,out=-80,in=70] (fiveN) to (sevenN);
      \draw [myarrow0,bend left] (sixN) to (sevenN);
      \draw [myarrow0,bend right] (sevenN) to (eightN);
      \draw [myarrow0,bend right] (eightN) to (nineN);
    \end{tikzpicture}
    \hfill\null
    \caption{Chordal network for the edge ideal of a tree.}
    \label{fig:triangnested}
  \end{figure}
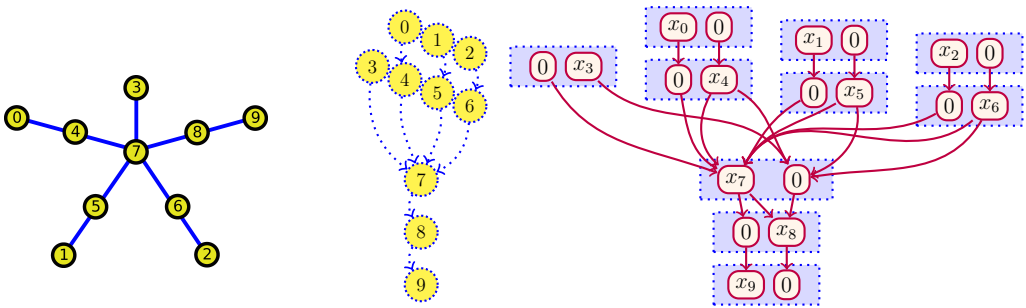

The minimal vertex covers of a graph $\mathcal{G}=(V,E)$ are in correspondence with the irreducible components of its \emph{edge ideal} $ I(\mathcal{G}):= \ideal{ x_ix_j: ij\in E}$ (see e.g., \cite[Prop~7.2.3]{Villarreal2015monomial}).
Therefore, the irreducible components of the edge ideal of a tree always have a very simple structure (although there might be exponentially many).
For instance, the diagram in \Cref{fig:triangnested} represents the components for the case of a simple $10$-vertex tree.
Here the components are given by the possible choices of one node from each of the (purple) boxes so that these nodes are connected
(e.g., $T=\{0,0,0,x_3,x_4,x_5,x_6,0,x_8,0\}$).
Note that there are $2^4+1=17$ components.

\end{example}

\begin{example}[Adjacent minors]\label{exmp:adjminors}
Let $X$ be a $2\times n$ matrix of indeterminates, and consider the polynomial set $F_n$ given by its adjacent minors, i.e.,
\begin{align*}
  X:=\bigl(\begin{smallmatrix}
    x_0 & x_2 & \cdots & x_{2n-2} \\ 
    x_1 & x_3 & \cdots & x_{2n-1}
  \end{smallmatrix}\bigr), \qquad
  F_n := \{x_{2i}x_{2i+3}-x_{2i+1}x_{2i+2}: 0\leq i < n-1\}.
\end{align*}
The corresponding ideal has been studied in e.g.,~\cite{Herzog2010,Diaconis1998}.
\Cref{fig:triangadjminors} shows the graph associated to this system.
We are interested in describing the irreducible components of $\V(F_n)$.

  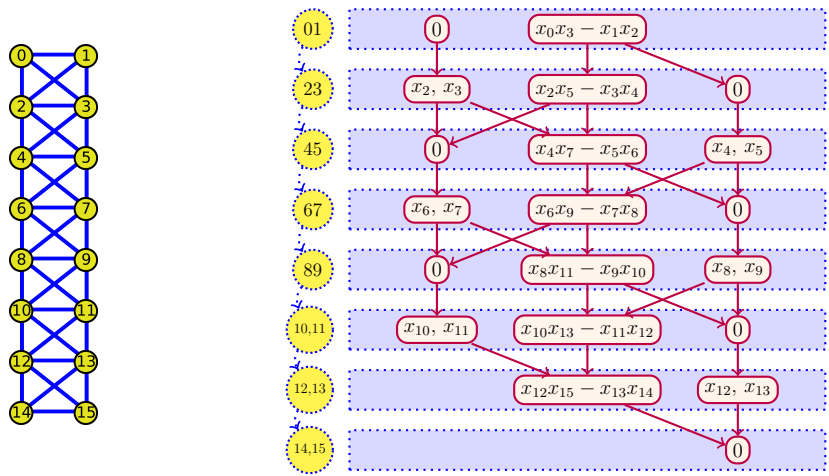
\begin{figure}[htb]
    \centering
    \null\hfill
    \raisebox{.13\height}{
    \includegraphics[scale=0.3]{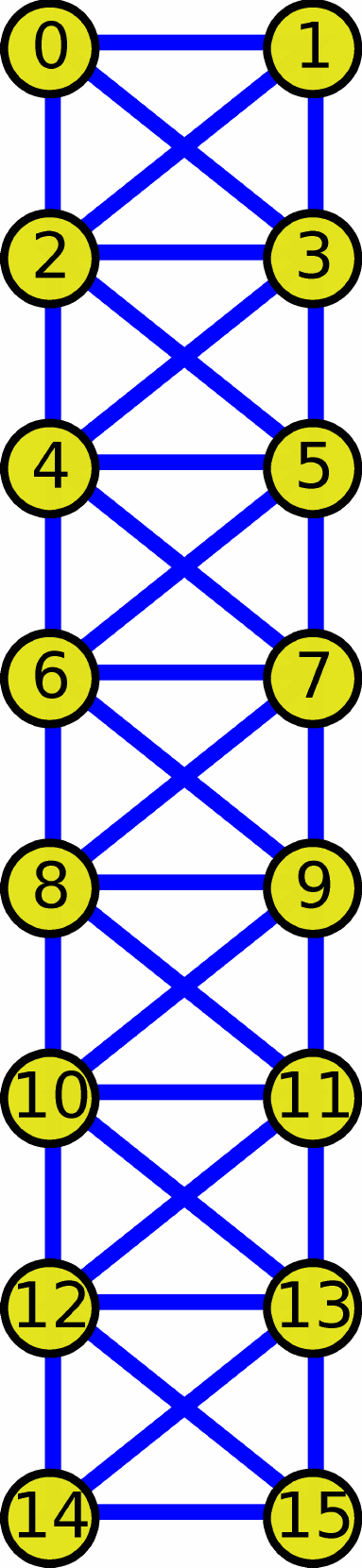}  
    }
    \hfill
    \begin{tikzpicture}[scale=1]
      \def\x{2.0}
      \def\y{.8}
      \tikzstyle{mybox} = [scale=0.65,draw=purple,fill=orange!8,thick,rectangle,rounded corners];
      \tikzstyle{myarrow} = [scale=0.65,purple,thick,->];
      \tikzstyle{mygroup} = [draw=blue,fill = blue!15, thick, dotted, minimum width = 180, minimum height = 15];
      \tikzstyle{mybox0} = [fill=yellow!80,draw=blue,thick,densely dotted,circle];
      \tikzstyle{myarrow0} = [blue,thick,dotted,->];
      %\def\x{1.8}
      %\def\y{1.1}
      %\tikzstyle{mybox} = [scale=0.65,draw=purple,thick,rectangle,rounded corners];
      %\tikzstyle{myarrow} = [purple,thick,->];
      %\tikzstyle{mygroup} = [draw=blue, thick, dotted, minimum width = 250, minimum height = 22];
      %\tikzstyle{mybox0} = [scale=0.65,text=blue,draw=blue,thick,densely dotted,circle];
      %\tikzstyle{myarrow0} = [blue,thick,dotted,->,bend right];

      \node [mybox] (zero1) at (1*\x,0*\y) {$x_0x_3-x_1x_2$};
      \node [mybox] (zero2) at (0*\x,0*\y) {$0$};
      %\node [mybox] (zero4) at (4*\x,0*\y) {$x_1$};

      \node [mybox] (two1) at (1*\x,-1*\y) {$x_2x_5-x_3x_4$};
      \node [mybox] (two2) at (2*\x,-1*\y) {$0$};
      \node [mybox] (two3) at (0*\x,-1*\y) {$x_2,\,x_3$};
      %\node [mybox] (two4) at (4*\x,-1*\y) {$x_3$};

      \node [mybox] (four1) at (1*\x,-2*\y) {$x_4x_7-x_5x_6$};
      \node [mybox] (four2) at (0*\x,-2*\y) {$0$};
      \node [mybox] (four3) at (2*\x,-2*\y) {$x_4,\,x_5$};
      %\node [mybox] (four4) at (4*\x,-2*\y) {$x_5$};

      \node [mybox] (six1) at (1*\x,-3*\y) {$x_6x_9-x_7x_8$};
      \node [mybox] (six2) at (2*\x,-3*\y) {$0$};
      \node [mybox] (six3) at (0*\x,-3*\y) {$x_6,\,x_7$};
      %\node [mybox] (six4) at (4*\x,-3*\y) {$x_7$};

      \node [mybox] (eight1) at (1*\x,-4*\y) {$x_8x_{11}-x_9x_{10}$};
      \node [mybox] (eight2) at (0*\x,-4*\y) {$0$};
      \node [mybox] (eight3) at (2*\x,-4*\y) {$x_8,\,x_9$};
      %\node [mybox] (eight4) at (4*\x,-4*\y) {$x_9$};

      \node [mybox] (ten1) at (1*\x,-5*\y) {$x_{10}x_{13}-x_{11}x_{12}$};
      \node [mybox] (ten2) at (2*\x,-5*\y) {$0$};
      \node [mybox] (ten3) at (0*\x,-5*\y) {$x_{10},\,x_{11}$};
      %\node [mybox] (ten4) at (4*\x,-5*\y) {$x_{11},\,x_{13}$};

      \node [mybox] (twelve1) at (1*\x,-6*\y) {$x_{12}x_{15}-x_{13}x_{14}$};
      %\node [mybox] (twelve2) at (0*\x,-6*\y) {$x_{14},\,x_{15}$};
      \node [mybox] (twelve3) at (2*\x,-6*\y) {$x_{12},\,x_{13}$};

      \node [mybox] (fourteen0) at (2*\x,-7*\y) {$0$};

      \draw [myarrow] (zero1) to (two1);
      \draw [myarrow] (zero1) to (two2);
      \draw [myarrow] (zero2) to (two3);
      %\draw [myarrow] (zero4) to (two3);
      %\draw [myarrow] (zero4) to (two4);

      \draw [myarrow] (two1) to (four1);
      \draw [myarrow] (two1) to (four2);
      \draw [myarrow] (two2) to (four3);
      \draw [myarrow] (two3) to (four1);
      \draw [myarrow] (two3) to (four2);
      %\draw [myarrow] (two3) to (four4);
      %\draw [myarrow] (two4) to (four3);
      %\draw [myarrow] (two4) to (four4);

      \draw [myarrow] (four1) to (six1);
      \draw [myarrow] (four1) to (six2);
      \draw [myarrow] (four2) to (six3);
      \draw [myarrow] (four3) to (six1);
      \draw [myarrow] (four3) to (six2);
      %\draw [myarrow] (four3) to (six4);
      %\draw [myarrow] (four4) to (six3);
      %\draw [myarrow] (four4) to (six4);

      \draw [myarrow] (six1) to (eight1);
      \draw [myarrow] (six1) to (eight2);
      \draw [myarrow] (six2) to (eight3);
      \draw [myarrow] (six3) to (eight1);
      \draw [myarrow] (six3) to (eight2);
      %\draw [myarrow] (six3) to (eight4);
      %\draw [myarrow] (six4) to (eight3);
      %\draw [myarrow] (six4) to (eight4);

      \draw [myarrow] (eight1) to (ten1);
      \draw [myarrow] (eight1) to (ten2);
      \draw [myarrow] (eight2) to (ten3);
      \draw [myarrow] (eight3) to (ten1);
      \draw [myarrow] (eight3) to (ten2);
      %\draw [myarrow] (eight3) to (ten4);
      %\draw [myarrow] (eight4) to (ten3);
      %\draw [myarrow] (eight4) to (ten4);

      \draw [myarrow] (ten1) to (twelve1);
      %\draw [myarrow] (ten1) to (twelve2);
      \draw [myarrow] (ten2) to (twelve3);
      \draw [myarrow] (ten3) to (twelve1);
      %\draw [myarrow] (ten3) to (twelve2);

      \draw [myarrow] (twelve1) to (fourteen0);
      \draw [myarrow] (twelve3) to (fourteen0);

      \begin{scope}[on background layer]
      \node [mygroup] (zero) at (1*\x,0*\y){}; 
      \node [mygroup] (two) at (1*\x,-1*\y){}; 
      \node [mygroup] (four) at (1*\x,-2*\y){}; 
      \node [mygroup] (six) at (1*\x,-3*\y){}; 
      \node [mygroup] (eight) at (1*\x,-4*\y){}; 
      \node [mygroup] (ten) at (1*\x,-5*\y){}; 
      \node [mygroup] (twelve) at (1*\x,-6*\y){}; 
      \node [mygroup] (fourteen) at (1*\x,-7*\y){}; 
      \end{scope}

      \node [mybox0,scale=0.60,left=.2 of zero] (zero0) {01};
      \node [mybox0,scale=0.60,left=.2 of two] (two0) {23};
      \node [mybox0,scale=0.60,left=.2 of four] (four0) {45};
      \node [mybox0,scale=0.60,left=.2 of six] (six0) {67};
      \node [mybox0,scale=0.60,left=.2 of eight] (eight0) {89};
      \node [mybox0,scale=0.45,left=.2 of ten] (ten0) {10,11};
      \node [mybox0,scale=0.45,left=.2 of twelve] (twelve0) {12,13};
      \node [mybox0,scale=0.45,left=.2 of fourteen] (fourteen0) {14,15};
      \draw [myarrow0,bend right] (zero0) to (two0);
      \draw [myarrow0,bend right] (two0) to (four0);
      \draw [myarrow0,bend right] (four0) to (six0);
      \draw [myarrow0,bend right] (six0) to (eight0);
      \draw [myarrow0,bend right] (eight0) to (ten0);
      \draw [myarrow0,bend right] (ten0) to (twelve0);
      \draw [myarrow0,bend right] (twelve0) to (fourteen0);
    \end{tikzpicture}
    \hfill\null
    \caption{Chordal network for the ideal of adjacent minors}
    \label{fig:triangadjminors}
  \end{figure}

  As in \Cref{exmp:triangcycle}, there is a simple recursive procedure to produce points on $\V(F_n)$:
we choose the values of the last column of the matrix arbitrarily, 
and then for column $i$ we either choose it arbitrarily, in case that column $i+1$ is zero, or we scale column $i+1$ if it is nonzero.
This procedure is actually describing the irreducible components of the variety.
In this way, the irreducible components admit a compact description, which is shown in \Cref{fig:triangadjminors}.
Again, the components are given by the maximal directed paths (e.g., $T = \{0, x_2, x_3, 0, x_6, x_7, x_{8}x_{11}-x_{9}x_{10}, 0, x_{12}, x_{13}, 0\}$) and its cardinality is the $n$-th Fibonacci number.
\end{example}

\subsection*{Contributions}

The examples from above show how certain polynomial systems with
tree-like structure admit a compact chordal network representation.
The aim of this paper is to develop a general framework to
systematically understand and compute chordal networks.  We also study
how to effectively use chordal networks to solve different problems
from computational algebraic geometry.  A major difficulty 
is that exponentially many triangular sets may appear (e.g., 
the Fibonacci number in \Cref{exmp:triangcycle}).  
%We remark that a data structure (BDDs, binary decision diagrams) with some similarities to chordal networks has been studied in the computer science literature for the case of binary functions; see \Cref{s:bdd}.

This paper presents the following contributions:
\begin{itemize}[leftmargin=.5in]
  \item We introduce the notion of chordal networks, a novel representation of polynomial ideals aimed toward exploiting structured sparsity.
  \item 
    We develop the \emph{chordal triangularization} method (\Cref{alg:chordtriang}) to compute such chordal network representation.
    Its correctness is established in
    \Cref{thm:chordtriang,thm:chordtriangpos}.
  \item We show that several families of polynomial systems admit a ``small'' chordal network representation, of size $O(n)$.
    This is true for certain zero-dimensional ideals (\Cref{thm:linearnetwork}), all monomial ideals (\Cref{thm:monlinear}) and certain binomial/determinantal ideals (\Cref{s:determinantal}), although in general this cannot be guaranteed (\Cref{rem:beyondzerodim}).
  \item We show how to effectively use chordal networks to compute several properties of the underlying variety.
    In particular, the cardinality (\Cref{s:countsolutions}), dimension and top-dimensional component (\Cref{s:usingtriangularmon}) can be computed in linear time. 
    In some interesting cases we can also describe the irreducible components.
  \item We present a Monte~Carlo algorithm to test radical ideal membership (\Cref{alg:member}).
    We show in \Cref{thm:memberlinear} that the complexity is linear when the given polynomial preserves some of the graphical structure of the system.
\end{itemize}
We point out that we have a preliminary implementation of a Macaulay2 package with all the methods from this paper, and it is available in \url{www.mit.edu/~diegcif}.

\subsection*{Structure of the paper}

The organization of this paper is as follows.
In \Cref{s:chordal} we review the concept of chordal graph and then formalize the notion of chordal network.
We then proceed to explain our methods, initially only for a restricted class of zero-dimensional problems (\Cref{s:zerodim,s:membership}), then for the case of monomial ideals (\Cref{s:monomial}), and finally considering the fully general case (\Cref{s:positivedim}).
We conclude the paper in \Cref{s:applications} with numerical examples of our methods.

The reason for presenting our results in this stepwise manner, is that the general case requires highly technical concepts from the theory of triangular sets.
Indeed, we encourage the reader unfamiliar with triangular sets to omit \Cref{s:positivedim} in the first read.
On the other hand, by first specializing our methods to the zero-dimensional and monomial cases we can introduce them all in a transparent manner.
Importantly, the basic structure of the chordal triangularization algorithm, presented in \Cref{s:zerodim}, remains the same for the general case.
Similarly, our algorithms that use chordal networks to compute properties of the variety (e.g., cardinality, dimension), introduced in~\Cref{s:membership,s:monomial}, also extend in a natural way.

\subsection*{Related work}
The development of chordal networks can be seen as a continuation of
our earlier work~\cite{Cifuentes2014}, and we refer the reader to that
paper for a detailed survey of the relevant literature on graphical
structure in computational algebraic geometry.  
For this reason, below we only discuss related work in the context of triangular sets, and point out the main differences between this paper
and~\cite{Cifuentes2014}.

This paper improves upon~\cite{Cifuentes2014} in two
main areas.  Firstly, chordal networks provide a much richer
description of the variety than the elimination ideals obtained by
chordal elimination.  For instance, the elimination ideals of the
equations from \Cref{exmp:adjminors} are trivial, but its
chordal network representation reveals its irreducible components.  In
addition, neither the dimension, cardinality nor radical ideal
membership can be directly computed from the elimination ideals (we
need a Gr\"obner basis).  Secondly, we show how to compute chordal
network representations for arbitrary polynomial systems (in
characteristic zero).
In contrast, chordal elimination only computes the elimination ideals under certain assumptions.

There is a broad literature studying triangular decompositions of
ideals~\cite{Aubry1999,Lazard1992,Kalkbrener1993,Maza2000,Wang2001}.
However, past work has not considered the case of sparse polynomial
systems.  Among the many existing triangular decomposition algorithms,
Wang's elimination methods are particularly relevant to
us~\cite{Wang2001,epsilon}.  Although seemingly unnoticed by Wang,
most of his algorithms preserve the chordal structure of the system.
As a consequence, we have experimentally seen that his methods are
more efficient than those based on regular
chains~\cite{regularchains,Maza2000} for the examples considered in
this paper.

%As opposed to all algorithms based on triangular sets, we emphasize chordal networks as the central object of this paper, 
As opposed to previous work, we emphasize chordal networks as our central object of study, 
rather than the explicit triangular decomposition obtained.
This is a key distinction since for several families of ideals the size of the chordal network is linear even though the corresponding triangular decomposition has exponential size (see the examples from above).
In addition, our methods deliberately treat triangular decompositions as a black box algorithm, allowing us to use either Lazard's LexTriangular algorithm~\cite{Lazard1992} for the zero-dimensional case or Wang's RegSer algorithm~\cite{Wang2000} for the positive-dimensional case.

\section{Chordal networks}\label{s:chordal}

\subsection{Chordal graphs}

Chordal graphs have many equivalent characterizations.
A good presentation is found in~\cite{Blair1993}.
For our purposes, we use the following definition.

\begin{definition}\label{defn:perfectelimination}
  Let $G$ be a graph with vertices $x_0,\ldots,x_{n-1}$.
  An ordering of its vertices $x_0>x_1>\dots>x_{n-1}$ is a \emph{perfect elimination ordering} if for each $x_l$ the set
 \begin{align}\label{eq:cliqueXl}
 X_l:=\{x_l\}\cup \{x_m: x_m\mbox{ is adjacent to }x_l,\; x_m<x_l\}
 \end{align}
 is such that the restriction $G|_{X_l}$ is a clique.
 A graph $G$ is \emph{chordal} if it has a  perfect elimination ordering.
\end{definition}
\begin{remark}
  Observe that lower indices correspond to larger vertices.
\end{remark}

Chordal graphs have many interesting properties.
For instance, they have at most $n$ maximal cliques, given that any clique is contained in some~$X_l$.
Note that trees are chordal graphs, since
by successively pruning a leaf from the tree we get a perfect elimination ordering.
We can always find a perfect elimination ordering of a chordal graph in linear time~\cite{rose1976algorithmic}.

\begin{definition}\label{defn:chordalcompletion}
  Let $\mathcal{G}$ be an arbitrary graph.
  We say that $G$ is a \emph{chordal completion} of $\mathcal{G}$ if it is chordal and $\mathcal{G}$ is a subgraph of ${G}$.
  The \emph{clique number} of ${G}$, denoted as $\kappa$, is the size of its largest clique.
  The \emph{treewidth} of $\mathcal{G}$ is the minimum clique number of ${G}$ (minus one) among all possible chordal completions.
\end{definition}

Observe that given any ordering $x_0>\cdots>x_{n-1}$ of the vertices of $\mathcal{G}$, there is a natural chordal completion ${G}$, i.e. we add edges to $\mathcal{G}$ in such a way that each ${G}|_{X_l}$ is a clique.
In general, we want to find a chordal completion with a small clique number.
However, there are $n!$ possible orderings of the vertices and thus finding the best chordal completion is not simple.
Indeed, this problem is NP-hard~\cite{Arnborg1987}, but there are good heuristics and approximation algorithms~\cite{bodlaender2008combinatorial,Vandenberghe2014}.
See~\cite{Bodlaender2010} for a comparison of some of these heuristics.

\begin{example}
\begin{figure}[htb]
  \centering
  \null\hfill
  \subfloat[][Chordal completion ]{\includegraphics[scale=0.4]{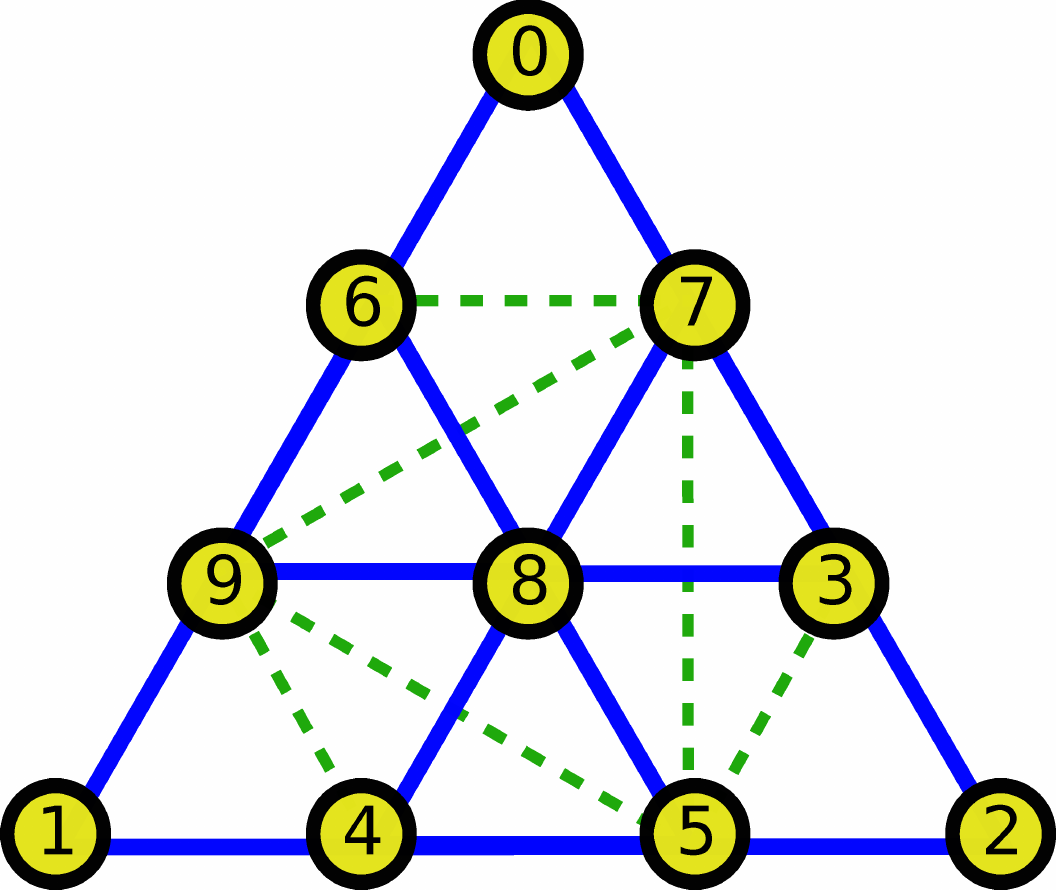}\label{fig:graph10notchordal}}
  \hfill
  \subfloat[][Elimination tree]{\includegraphics[scale=0.4]{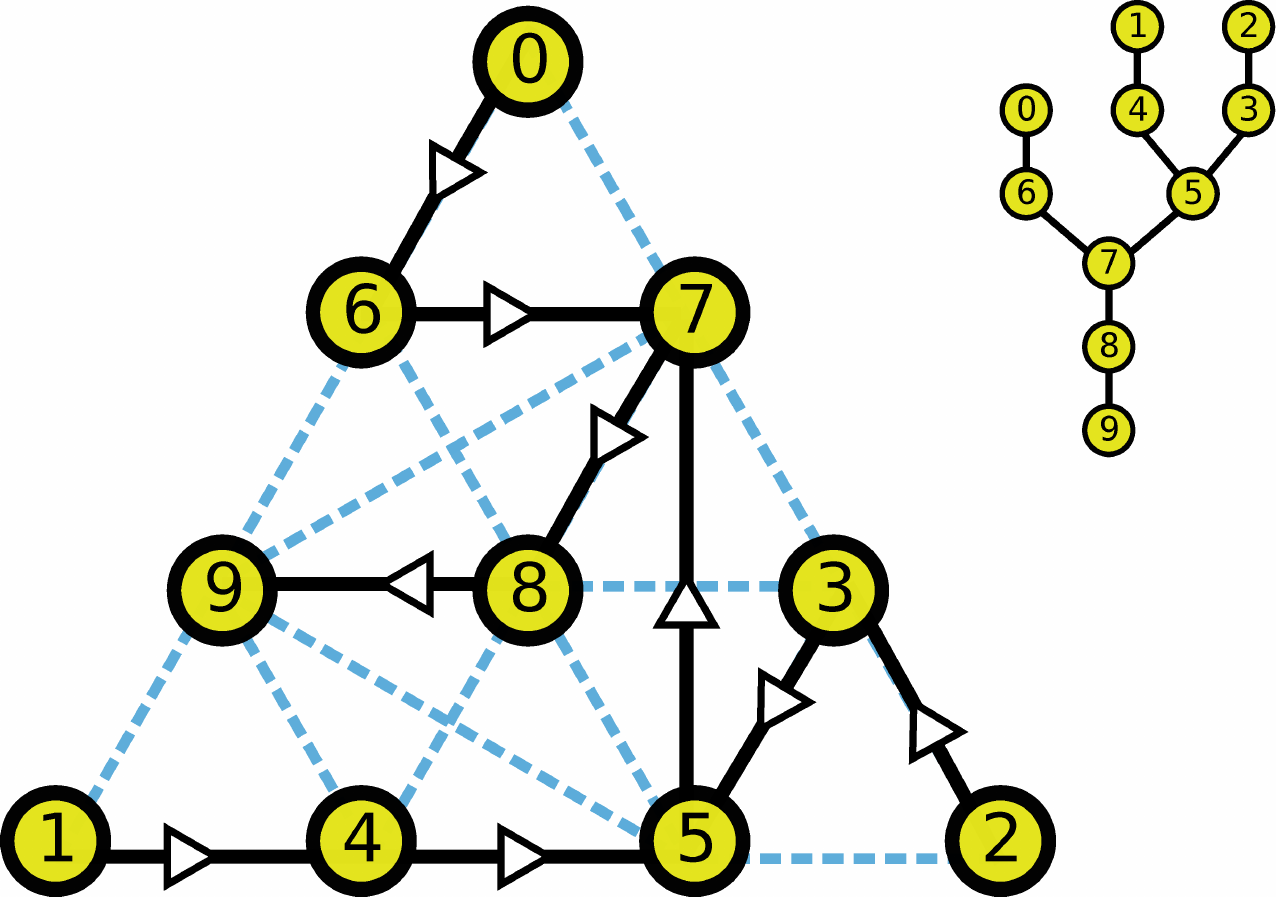}\label{fig:elimtree}}
  \hfill\null
  \caption
  {Left: 10-vertex graph (blue/solid) and a chordal completion (green/dashed).
  Right: Elimination tree of the chordal completion.}
\end{figure}
  Let $\mathcal{G}$ be the blue/solid graph in \Cref{fig:graph10notchordal}.
  This graph is not chordal but if we add the six green/dashed edges shown in the figure we obtain a chordal completion ${G}$.
  In fact, the ordering $x_0>\cdots>x_9$ is a perfect elimination ordering of the chordal completion.
  The clique number of ${G}$ is four and the treewidth of $\mathcal{G}$ is three.
\end{example}

As mentioned earlier, we will assume throughout this document that the polynomial system~$F$ is \emph{supported} on a given chordal graph~$G$,
where by supported we mean that $G$ is a chordal completion of~$\mathcal{G}(F)$.
Moreover, we assume that the ordering of the vertices (inherited from the polynomial ring) is a perfect elimination ordering of~$G$.

Given a chordal graph $G$ with some perfect elimination ordering, there is an associated tree that will be very helpful in our discussion.

\begin{definition}\label{defn:eliminationtree}
  Let $G$ be an ordered graph with vertices $x_0>\cdots>x_{n-1}$.
  The \emph{elimination tree} of~$G$ is the following directed spanning tree:
  for each $l$ there is an arc from $x_l$ towards the largest  $x_p$ that is adjacent to $x_l$ and $x_p<x_l$.
 We will say that $x_p$ is \emph{the parent} of $x_l$ and $x_l$ is \emph{a child} of $x_p$.
 Note that the elimination tree is rooted at~$x_{n-1}$.
\end{definition}

\Cref{fig:elimtree} shows an example of the elimination tree.
We now present a simple property of the elimination tree of a chordal graph.

\begin{lemma}\label{thm:cliquecontainment}
  Let $G$ be a chordal graph, let $x_l$ be some vertex and let $x_p$ be its parent in the elimination tree $T$. Then
    $X_l \setminus \{x_l\}\subset X_p,$
  where $X_i$ is as in~\eqref{eq:cliqueXl}.
\end{lemma}
\begin{proof}
  Let $Y :=  X_l\setminus \{x_l\}$. 
  Note that $Y$ is a clique, whose largest variable is $x_p$.
  Since $X_p$ is the unique largest clique satisfying such property, we must have $Y\subset X_p$. 
\end{proof}

\subsection{Chordal networks}\label{s:chordalnetwork}

We proceed to formalize the concept of chordal networks.

\begin{definition}
  Let $G$ be a chordal graph with vertex set $X$, and let $X_l$ be as in~\eqref{eq:cliqueXl}.
  A $G$-\emph{chordal network} is a directed graph $\mathcal{N}$, whose nodes are polynomial sets in $\K[X]$, such that:
  \begin{itemize}[leftmargin=.5in]
    \item \emph{(nodes supported on cliques)} each node $F_l$ of $\mathcal{N}$ is given a rank $l= \rank(F_l)$, with $0\leq l< n$, such that $F_l\subset \K[X_l]$.
    \item \emph{(arcs follow elimination tree)} if $(F_l,F_p)$ is an arc of $\mathcal{N}$ then $(l,p)$ is an arc of the elimination tree of $G$, where $l=\rank(F_l),p=\rank(F_p)$.
  \end{itemize}
%We are particularly interested in a special type of chordal networks.
A chordal network is \emph{triangular} if each node consists of a single polynomial $f$, and either $f=0$ or its largest variable is $x_{\rank(f)}$.
\end{definition}

There is one parameter of a chordal network that will determine the complexity of some of our methods.
The \emph{width} of a chordal network, denoted as~$W$, is the largest number of nodes of any given rank.
Note that the number of nodes in the network is at most~$nW$, and the number of arcs is at most~$(n-1)W^2$.

We can represent chordal networks using the diagrams we have shown throughout the paper.
Since the structure of a chordal network resembles the elimination tree (second item in the definition), we usually show the elimination tree to the left of the network.

\begin{figure}[htb]
  \centering
    \begin{tikzpicture}[scale=.99]
      \def\x{2.50}
      \def\y{1.1}
      \def\yzero{0*\y}
      \def\yone{\yzero-.3*\y}
      \def\ytwo{\yone-.3*\y}
      \def\ythree{\ytwo-1*\y}
      \def\yfour{\ythree-.3*\y}
      \def\yfive{\yfour-1.2*\y}
      \def\ysix{\yfive-.3*\y}
      \def\yseven{\ysix-1*\y}
      \def\yeight{\yseven-.8*\y}
      \def\ynine{\yeight-.8*\y}
      \tikzstyle{myboxA} = [scale=0.63,draw=purple,fill=orange!8,thick,rectangle,rounded corners];
      \tikzstyle{myboxB} = [scale=0.63,draw=black!20!orange,fill=orange!8,thick,rectangle,rounded corners];
      \tikzstyle{myarrowA} = [purple,thick,->];
      \tikzstyle{myarrowB} = [black!20!orange,thick,->];
      \tikzstyle{mygroup} = [draw=blue,fill = blue!15, thick, dotted, minimum height = 15.5*\y];
      \tikzstyle{myboxA0} = [scale=0.65,fill=yellow!80,draw=blue,thick,densely dotted,circle];
      \tikzstyle{myarrowA0} = [blue,thick,dotted,->];
      \node [scale=0.65] (legend) at (1.9*\x,\yeight-.4*\y) {$g(a,b,c):=a^2+b^2+c^2+ab+bc+ca$};
      \node [myboxA] (zero0) at (-2.15*\x,\yzero) {$x_0^3 + x_0^2x_7 + x_0x_7^2 + x_7^3$};
      \node [myboxA] (zero1) at (-1.3*\x,\yzero) {$g(x_0,x_6,x_7)$};
      \node [myboxB](one0) at (-.32*\x,\yone) {$x_1^3 + x_1^2x_9 + x_1x_9^2 + x_9^3$};
      \node [myboxB](one1) at (.55*\x,\yone) {$g(x_1,x_4,x_9)$};
      \node [myboxA](two0) at (1.6*\x,\ytwo) {$x_2^3 + x_2^2x_5 + x_2x_5^2 + x_5^3$};
      \node [myboxA](two1) at (2.5*\x,\ytwo) {$g(x_2,x_3,x_5)$};
      \node [myboxB](three0) at (1.19*\x,\ythree) {$x_3 - x_5$};
      \node [myboxB](three1) at (2.55*\x,\ythree) {$g(x_3,x_7,x_8)$};
      \node [myboxB](three2) at (1.82*\x,\ythree) {$x_3 + x_5 + x_7 + x_8$};
      \node [myboxA](four0) at (-.83*\x,\yfour) {$x_4 - x_9$};
      \node [myboxA](four1) at (0.55*\x,\yfour) {$g(x_4,x_8,x_9)$};
      \node [myboxA](four2) at (-.19*\x,\yfour) {$x_4 + x_5 + x_8 + x_9$};
      \node [myboxB](five0) at (-.3*\x,\yfive) {$g(x_5,x_8,x_9)$};
      \node [myboxB](five1) at (.55*\x,\yfive) {$x_5 + x_7 + x_8 + x_9$};
      \node [myboxB](five2) at (1.8*\x,\yfive) {$x_5 - x_9$};
      \node [myboxB](five3) at (1.3*\x,\yfive) {$x_5 - x_7$};
      \node [myboxB](five4) at (2.3*\x,\yfive) {$x_5 - x_9$};
      \node [myboxA](six0) at (-2.55*\x,\ysix) {$x_6 - x_7$};
      \node [myboxA](six1) at (-2.0*\x,\ysix) {$g(x_6,x_8,x_9)$};
      \node [myboxA](six2) at (-1.25*\x,\ysix) {$x_6 + x_7 + x_8 + x_9$};
      \node [myboxB](seven0) at (-0.5*\x,\yseven) {$x_7 - x_9$};
      \node [myboxB](seven1) at (0.45*\x,\yseven) {$g(x_7,x_8,x_9)$};
      \node [myboxA](eight0) at (0*\x,\yeight) {$x_8^3 + x_8^2x_9 + x_8x_9^2 + x_9^3$};
      \node [myboxB](nine0) at (0*\x,\ynine) {$x_9^4 - 1$};

      \draw [myarrowA,out=-100,in=100] (zero0) to (six0);
      \draw [myarrowA,out=-120,in=100] (zero1) to (six1);
      \draw [myarrowA,out=-100,in=120] (zero1) to (six2);
      \draw [myarrowB,out=-120,in=70] (one0) to (four0);
      \draw [myarrowB] (one1) to (four1);
      \draw [myarrowB,out=-130,in=60] (one1) to (four2);
      \draw [myarrowA] (two0) to (three0);
      \draw [myarrowA] (two1) to (three1);
      \draw [myarrowA] (two1) to (three2);
      \draw [myarrowB,out=-120,in=20] (three0) to (five0);
      \draw [myarrowB,out=-70,in=18] (three0) to (five1);
      \draw [myarrowB,out=-30,in=125] (three0) to (five2);
      \draw [myarrowB,out=-130,in=50] (three1) to (five3);
      \draw [myarrowB] (three1) to (five4);
      \draw [myarrowB,out=-135,in=20] (three2) to (five0);
      \draw [myarrowB,out=-120,in=18] (three2) to (five1);
      \draw [myarrowB] (three2) to (five2);
      \draw [myarrowA] (four0) to (five0);
      \draw [myarrowA,out=-40,in=160] (four0) to (five1);
      \draw [myarrowA,out=-30,in=155] (four0) to (five3);
      \draw [myarrowA] (four1) to (five2);
      \draw [myarrowA] (four1.east) to (five4);
      \draw [myarrowA] (four2) to (five0);
      \draw [myarrowA] (four2) to (five1);
      \draw [myarrowA,out=-30,in=155] (four2) to (five3);
      \draw [myarrowB] (five0) to (seven0);
      \draw [myarrowB] (five1) to (seven1);
      \draw [myarrowB] (five2) to (seven1);
      \draw [myarrowB] (five3) to (seven1);
      \draw [myarrowB,out=-150,in=20] (five4) to (seven0);
      \draw [myarrowA,out=-15,in=160] (six0.south) to (seven1.west);
      \draw [myarrowA,out=-40,in=175] (six1) to (seven0.west);
      \draw [myarrowA,out=-20,in=160] (six2) to (seven1);
      \draw [myarrowB] (seven0) to (eight0);
      \draw [myarrowB] (seven1) to (eight0);
      \draw [myarrowA] (eight0) to (nine0);
      \begin{scope}[on background layer]
      \node [mygroup, minimum width=50*\x] (zero) at (-1.84*\x,\yzero){}; 
      \node [mygroup, minimum width=50*\x] (one) at (.02*\x,\yone){}; 
      \node [mygroup, minimum width=51*\x] (two) at (1.94*\x,\ytwo){}; 
      \node [mygroup, minimum width=54*\x] (three) at (1.91*\x,\ythree){}; 
      \node [mygroup, minimum width=55*\x] (four) at (-.10*\x,\yfour){}; 
      \node [mygroup, minimum width=91*\x] (five) at (.95*\x,\yfive){}; 
      \node [mygroup, minimum width=57*\x] (six) at (-1.8*\x,\ysix){}; 
      \node [mygroup, minimum width=45*\x] (seven) at (0*\x,\yseven){}; 
      \node [mygroup, minimum width=35*\x] (eight) at (0*\x,\yeight){}; 
      \node [mygroup, minimum width=20*\x] (nine) at (0*\x,\ynine){}; 
      \end{scope}
      \node [myboxA0] (zeroN)  at (-3.3*\x,\yzero) {$0$};
      \node [myboxA0] (oneN)   at (-3.1*\x,\yone) {$1$};
      \node [myboxA0] (twoN)   at (-2.9*\x,\ytwo) {$2$};
      \node [myboxA0] (threeN) at (-2.9*\x,\ythree) {$3$};
      \node [myboxA0] (fourN)  at (-3.1*\x,\yfour) {$4$};
      \node [myboxA0] (fiveN)  at (-3.1*\x,\yfive) {$5$};
      \node [myboxA0] (sixN)   at (-3.3*\x,\ysix) {$6$};
      \node [myboxA0] (sevenN) at (-3.1*\x,\yseven) {$7$};
      \node [myboxA0] (eightN) at (-3.1*\x,\yeight) {$8$};
      \node [myboxA0] (nineN)  at (-3.1*\x,\ynine) {$9$};
      \draw [myarrowA0,out=-100,in=100] (zeroN) to (sixN);
      \draw [myarrowA0,out=-105,in=105] (oneN) to (fourN);
      \draw [myarrowA0,out=-75,in=75] (twoN) to (threeN);
      \draw [myarrowA0,out=-80,in=60] (threeN) to (fiveN);
      \draw [myarrowA0,out=-105,in=105] (fourN) to (fiveN);
      \draw [myarrowA0,out=-80,in=80] (fiveN) to (sevenN);
      \draw [myarrowA0,bend right] (sixN) to (sevenN);
      \draw [myarrowA0,bend right] (sevenN) to (eightN);
      \draw [myarrowA0,bend right] (eightN) to (nineN);
    \end{tikzpicture}
  \caption{${G}$-chordal network, where ${G}$ is the chordal graph from \Cref{fig:graph10notchordal}. The elimination tree of ${G}$ is shown on the left.}
  \label{fig:triang10vars}
\end{figure}
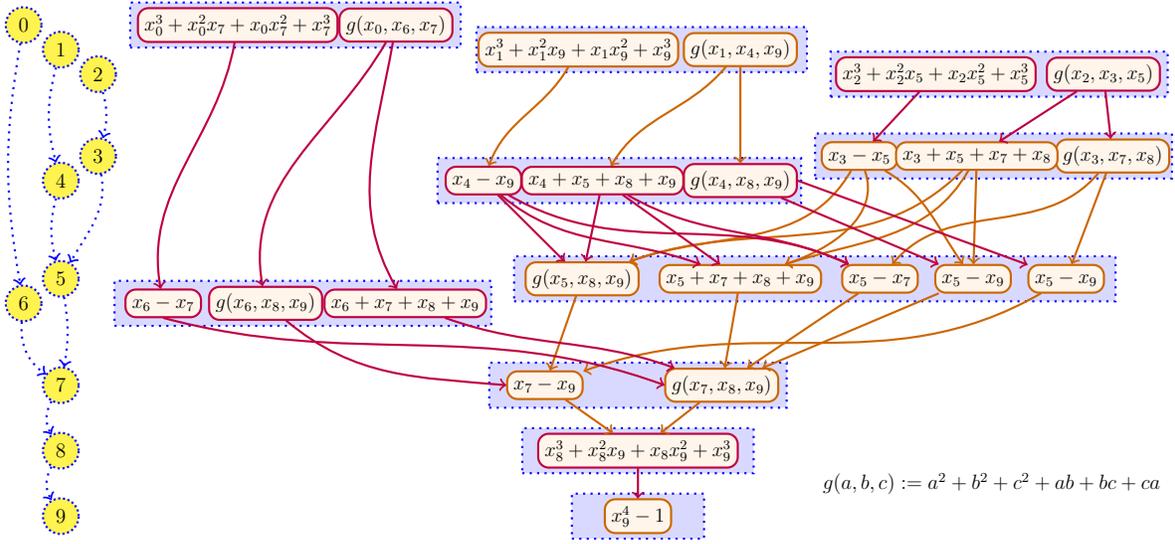

\begin{example}
  Let $\mathcal{G}$ be the blue/solid graph from \Cref{fig:graph10notchordal}, and let ${G}$ be the green/dashed chordal completion.
  \Cref{fig:triang10vars} shows a ${G}$-chordal network of width~$5$, that represents the $4$\nobreakdash-colorings of graph~$\mathcal{G}$ (\Cref{eq:colorings}).
  The elimination tree of ${G}$ is shown to the left of the diagram.
  Note that this network is triangular, and thus all its nodes consist of a single polynomial.
  For instance, two of its nodes are $f_5 =x_5+x_7+x_8+x_9$ and $f_6=x_6-x_7$.
  Nodes are grouped in blue rectangular boxes according to their rank.
  In particular, $f_5$ has rank~$5$ and $f_6$ rank~$6$, and indeed $f_5\in \K[X_5]=\K[x_5,x_7,x_8,x_9]$ and $f_6\in \K[X_6]=\K[x_6,x_7,x_8,x_9]$.
\end{example}

\begin{example}
  Let $\mathcal{G}$ be the $9$-cycle with vertices $x_0,\ldots,x_8$.
  Let ${G}$ be the chordal completion obtained by connecting vertex $x_8$ to all the others.
  \Cref{fig:triangcycle} shows a triangular ${G}$-chordal network.
  The elimination tree, shown to the left of the network, is the path $x_0\to  \cdots\to x_8$.
\end{example}

\begin{remark}
Sometimes we collapse certain ranks to make the diagram visually simpler.
In particular, in \Cref{fig:triangadjminors} we collapse the ranks $2i,2i+1$ into a single group.
\end{remark}

%As suggested by the examples in the introduction, we will use chordal networks to decompose polynomial ideals.
%In particular, triangular chordal networks yield decompositions into triangular sets.
As suggested by the examples in the introduction, a triangular chordal network gives a decomposition of the polynomial ideal into triangular sets.
Each such triangular set corresponds to a \emph{chain} of the network, as defined next.

\begin{definition}
  Let $\mathcal{N}$ be a $G$-chordal network.
  A \emph{chain} of $\mathcal{N}$ is a tuple of nodes $C=(F_0,F_1,\ldots,F_{n-1})$ such that:
  \begin{itemize}[leftmargin=.5in]
    \item $\rank(F_l)=l$ for each $l$.
    \item if $x_p$ is the parent of $x_l$, then $(F_l,F_p)$ is an arc of $\mathcal{N}$.
  \end{itemize}
\end{definition}

\begin{example}\label{exmp:triang10vars}
  The chordal network from \Cref{fig:triang10vars} has~$21$ chains, one of which is:
  {\small
  \begin{gather*}
    C = (
      x_0^2 + x_0x_6 + x_0x_7 + x_6^2 + x_6x_7 + x_7^2,\;
      x_1^3 + x_1^2x_9 + x_1x_9^2 + x_9^3,\;
      x_2^3 + x_2^2x_5 + x_2x_5^2 + x_5^3,\;\\
      x_3 - x_5,\;
      x_4 - x_9,\;
      x_5^2 + x_5x_8 + x_5x_9 + x_8^2 + x_8x_9 + x_9^2,\;\\
      x_6^2 + x_6x_8 + x_6x_9 + x_8^2 + x_8x_9 + x_9^2,\;
      x_7 - x_9,\;
      x_8^3 + x_8^2x_9 + x_8x_9^2 + x_9^3,\;
      x_9^4 - 1
      ).
  \end{gather*}
}
\end{example}

\subsection{Binary Decision Diagrams}
\label{s:bdd}

Although motivated from a different perspective and with quite
distinct goals, throughout the development of this paper we realized
the intriguing similarities between chordal networks and a class of
data structures used in computer science known as ordered binary
decision diagrams
(OBDD)~\cite{Akers1978,Bryant1992,Knuth2011,Wegener2000}.

A binary decision diagram (BDD) is a data structure that can be used
to represent Boolean (binary) functions in terms of a directed acyclic
graph. They can be interpreted as a binary analogue of a straight-line
program, where the nodes are associated with variables and the
outgoing edges of a node correspond to the possible values of that
variable.  A particularly important subclass are the \emph{ordered}
BDDs (or OBDDs), where the branching occurs according to a specific
fixed variable ordering. Under a mild condition (reducibility) this
representation can be made unique, and thus every Boolean function has
a canonical OBDD representation. OBDDs can be effectively used for
further manipulation (e.g., decide satisfiability, count satisfying
assignments, compute logical operations).  Interestingly, several
important functions have a compact OBDD representation.  A further
variation, \emph{zero-suppressed} BDDs (ZBDDs), can be used to
efficiently represent subsets of the hypercube $\{0,1\}^n$ and to manipulate them (e.g.,
intersections, sampling, linear optimization).

Chordal networks can be thought of as a wide generalization of
OBDDs/ZBDDs to arbitrary algebraic varieties over general fields (instead of finite sets in $(\F_2)^{n}$).  
Like chordal networks, an OBDD corresponds to a
certain directed graph, but where the nodes are variables ($x_0,x_1,\ldots$) instead of polynomial sets.  We will see in
\Cref{s:monomial} that for the specific case of monomial
ideals, the associated chordal networks also have 
this form. Since one of our main goals is to preserve graphical
structure for efficient computation, in this paper we define chordal
networks only for systems that are structured according to some
chordal graph.  In addition,
for computational purposes we do not insist on uniqueness of the
representation (although it might be possible to make them canonical
after further processing).

The practical impact of data structures like BDDs and OBDDs over the
past three decades has been very significant, as they have enabled
breakthrough results in many areas of computer science including model
checking, formal verification and logic synthesis. We hope that
chordal networks will make possible similar advances in computational
algebraic geometry. The connections between BDDs and chordal
networks run much deeper, and we plan to further explore them in the
future.

\section{The chordally zero-dimensional case}\label{s:zerodim}

In this section we present our main methods to compute triangular chordal networks, although focused on a restricted type of zero-dimensional problems.
This restriction is for simplicity only; we will see that our methods naturally extend to arbitrary ideals.
Concretely, we consider the following family of polynomial sets.

\begin{definition}[Chordally zero-dimensional]\label{defn:chordallyzerodim}
  \hspace{-5pt}Let $F\subset \K[X]$ be supported on a chordal graph~$G$.
  We say that $F$ is \emph{chordally} zero-dimensional, if for each maximal clique $X_l$ of graph~$G$ the ideal $\ideal{F\cap \K[X_l]}$ is zero-dimensional.
\end{definition}

Note that the $q$-coloring equations in~\eqref{eq:colorings} are chordally zero-dimensional.
As in \Cref{exmp:triangcycle}, chordally zero-dimensional problems always have simple chordal network representations.

\begin{remark}[The geometric picture]
  There is a nice geometric interpretation behind the chordally zero-dimensional condition.
  Denoting $V_l$ the variety of $F\cap\K[X_l]$, the condition is that each $V_l$ is finite.
  Note now that $\pi_{X_l}(\V(F))\subset V_l$, where $\pi_{X_l}$ denotes the projection onto the coordinates of~$X_l$.
  Thus, independent of the size of $\V(F)$, the chordally zero-dimensional condition allows us to bound the size of its projections onto each~$X_l$.
  More generally, although not elaborated in this paper, our methods are expected to perform well on any $F$ (possibly positive-dimensional) for which the projections $\pi_{X_l}(\V(F))$ are well-behaved.
\end{remark}

\subsection{Triangular sets}\label{s:triangularsets}

We now recall the basic concepts of triangular sets for the case of zero-dimensional ideals, following~\cite{Lazard1992}.
We delay the exposition of the positive-dimensional case to \Cref{s:positivedim}.

\begin{definition}\label{defn:triangularzero}
  Let $f\in \K[X]\setminus \K$ be a non-constant polynomial.
  The \emph{main variable} of~$f$, denoted $\mvar(f)$, is the greatest variable appearing in $f$.
  The \emph{initial} of~$f$, denoted $\init(f)$, is the leading coefficient of~$f$ when viewed as a univariate polynomial in $\mvar(f)$.
  A \emph{zero-dimensional triangular set} is a collection of non-constant polynomials $T = \{t_0,\ldots,t_{n-1}\}$ such that
  $\mvar(t_i)=x_i$ and 
  $\init(t_i)=1$ for each $i$.
\end{definition}

Most of the analysis done in this paper will work over an arbitrary field $\K$.
For some results we require the field to contain sufficiently many elements, so we might need to consider a field extension.
We denote by $\overline{\K}$ the algebraic closure of~$\K$.
For a polynomial set $F$, we let $\V(F)\subset \overline{\K}^n$ be its \emph{variety}.
Note that for a zero-dimensional triangular set $T$, we always have 
\begin{align} \label{eq:degT}
  |\V(T)|\leq \deg(T):=\prod_{t\in T}\mdeg(t),
\end{align}
where $\mdeg(t):=\deg(t,\mvar(t))$ denotes the degree on the main variable.
Furthermore, the above is an equality if we count multiplicities.

For a triangular set $T$, let $\ideal{T}$ denote the generated ideal.
It is easy to see that a zero-dimensional triangular set is a lexicographic Gr\"obner basis of $\ideal{T}$.
In particular, we can test ideal membership by taking normal form.
We also denote as $\elim{p}{T}:= T\cap \K[x_p,x_{p+1},\ldots]$ the subset of $T$ restricted to variables less or equal than $x_p$.
Note that $\elim{p}{T}$ generates the elimination ideal of $\ideal{T}$ because of the elimination property of lexicographic Gr\"obner bases.

\begin{notation}
  We let $S=S_1\sqcup S_2$ denote a disjoint union, i.e., $S=S_1\cup S_2$ and $S_1\cap S_2=\emptyset$.
\end{notation}

\begin{definition}
  Let $I\subset\K[X]$ be a zero-dimensional ideal.
  A \emph{triangular decomposition} of $I$ is a collection $\T$ of triangular sets, such that
    $\V(I) = \bigsqcup_{\,T\in \T} \V(T).$
  We say that $\T$ is \emph{squarefree} if each $T\in \T$ generates a radical ideal.
  We say that $\T$ is \emph{irreducible} if each $T\in \T$ generates a prime  ideal (or equivalently, a maximal ideal).
\end{definition}

Lazard proposed algorithms to compute a triangular decomposition from a Gr\"obner basis~\cite{Lazard1992}.
He also showed how to post-process it to make it squarefree/irreducible.

\begin{remark}
  As explained in~\cite{Lazard1992}, there might be several distinct triangular decompositions of an ideal, but there are simple ways to pass from one description to another.
\end{remark}

\subsection{Chordal triangularization}

We proceed to explain how to compute a triangular chordal network representation of a polynomial set~$F$.
We will start with a particular (induced) chordal network that is modified step after step to make it triangular.

\begin{definition}
  Let $F\subset\K[X]$ be supported on a chordal graph $G$.
  The \emph{induced} $G$-chordal network has a unique node of rank $k$, namely $F_k:=F\cap \K[X_k]$, and its arcs are the same as in the elimination tree, i.e., $(F_l,F_p)$ is an arc if $x_p$ is the parent of $x_l$.
\end{definition}

We will sequentially perform two types of operations to the induced chordal network.
\begin{description}
  \item[\myfont{Triangulate}($F_l$)] Let $\mathcal{T}$ be a triangular decomposition of a node  $F_l$ of the network.
    Replace node $F_l$ with one node for each triangular set in $\mathcal{T}$.
    Any node which was previously connected to $F_l$ is then connected to each of the new nodes.
  \item[\myfont{Eliminate}($T$)] Let $T$ be a rank~$l$ node and let $x_p$ be the parent of $x_l$.
    Let $T_{p}:= \elim{p}{T}$ and $T_{l}:= T\setminus T_{p}$.
    For each arc $(T,F_p)$ we create a new rank~$p$ node $F_p':=F_p\cup T_{p}$, and we substitute arc $(T,F_p)$ with $(T,F_p')$.
    Then, we copy all arcs coming out of $F_p$ to $F_p'$ (while keeping the old arcs).
    %In case there is no other node $F_l'$ (with variable $x_l$) connected to $F_p$, we delete $F_p$.
    Next, we replace the content of node $T$ with the polynomial set $T_{l}$. 
\end{description}
The operations are performed in rounds: in the $l$-th round we triangulate/eliminate all rank~$l$ nodes.
After each round, we may reduce the network with the following additional operations.
\begin{description}
  \item[\myfont{MergeIn}($l$)] Merge any two rank~$l$ nodes $F_l,F_l'$ if they define the same ideal, and they have the same sets of incoming arcs.
  \item[\myfont{MergeOut}($l$)] Merge any two rank~$l$ nodes $F_l,F_l'$ if they define the same ideal, and they have the same sets of outgoing arcs.
\end{description}

\begin{example}\label{exmp:triang4vars}
  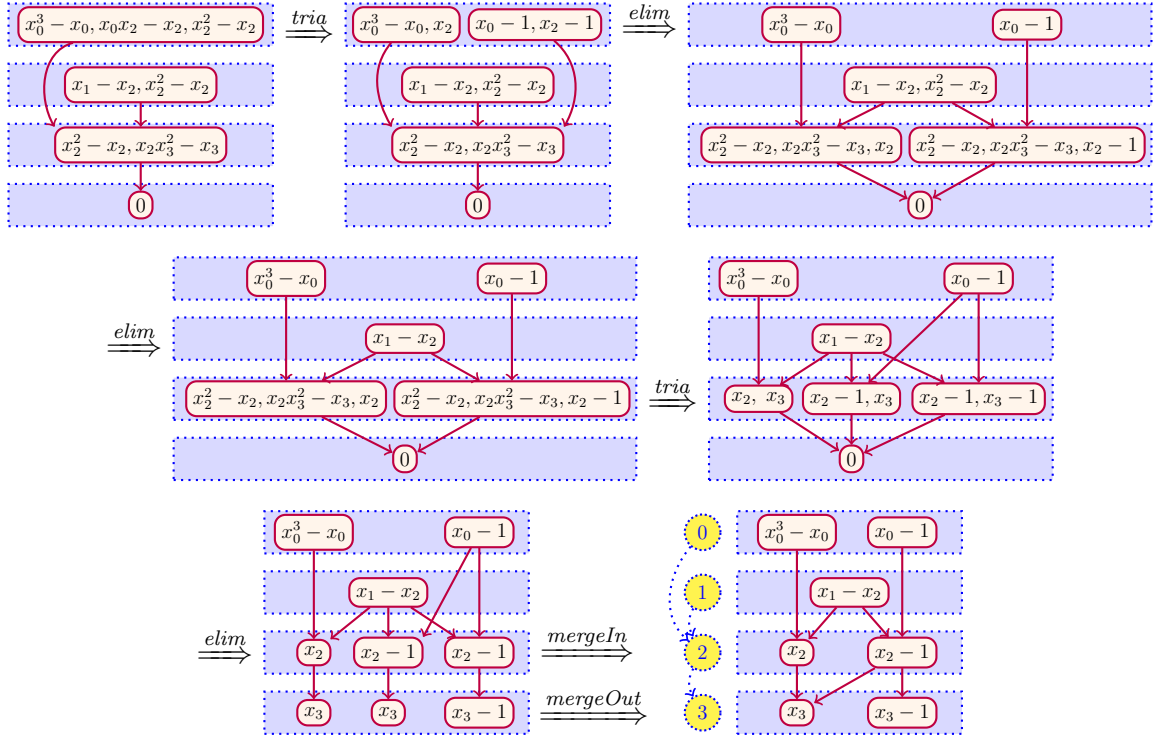
\begin{figure}[htb]
    \centering

    \def\y{.8}
    \tikzstyle{myboxA} = [scale=0.65,draw=purple,fill=orange!8,thick,rectangle,rounded corners]
    \tikzstyle{myarrowA} = [purple,thick,->]
    \tikzstyle{myboxA0} = [scale=0.65,fill=yellow!80,text=blue,draw=blue,thick,densely dotted,circle]
    \tikzstyle{myarrowA0} = [blue,thick,dotted,->]
    \tikzstyle{mygroup0} = [draw=blue,fill=blue!15, thick, dotted, minimum height=16]
    \begin{tikzpicture}[scale=1]
      \def\x{.8}
      \tikzstyle{mygroup} = [mygroup0, minimum width = 100]
      \node [myboxA] (zero1) at (0*\x,0*\y) {$x_0^3 - x_0, x_0x_2-x_2,x_2^2-x_2$};
      \node [myboxA](one1) at (0*\x,-1*\y) {$x_1 - x_2, x_2^2-x_2$};
      \node [myboxA](two1) at (0*\x,-2*\y) {$x_2^2-x_2, x_2x_3^2-x_3$};
      \node [myboxA](three1) at (0*\x,-3*\y) {$0$};
      \draw [myarrowA,out=-165,in=120] (zero1) to (two1.north west);
      \draw [myarrowA] (one1) to (two1);
      \draw [myarrowA] (two1) to (three1);
      \begin{scope}[on background layer]
      \node [mygroup] (zero) at (0*\x,0*\y){}; 
      \node [mygroup] (one) at (0*\x,-1*\y){}; 
      \node [mygroup] (two) at (0*\x,-2*\y){}; 
      \node [mygroup] (three) at (0*\x,-3*\y){}; 
      \end{scope}
    \end{tikzpicture}
    \raisebox{5.3\height}{$\xRightarrow{\mathit{tria}}$}
    \begin{tikzpicture}[scale=1]
      \def\x{.95}
      \tikzstyle{mygroup} = [mygroup0, minimum width = 100]
      \node [myboxA] (zero1) at (-1*\x,0*\y) {$x_0^3 - x_0, x_2$};
      \node [myboxA] (zero2) at (.8*\x,0*\y) {$x_0-1, x_2-1$};
      \node [myboxA](one1) at (0*\x,-1*\y) {$x_1 - x_2, x_2^2-x_2$};
      \node [myboxA](two1) at (0*\x,-2*\y) {$x_2^2-x_2, x_2x_3^2-x_3$};
      \node [myboxA](three1) at (0*\x,-3*\y) {$0$};
      \draw [myarrowA,out=-130,in=120] (zero1) to (two1.north west);
      \draw [myarrowA,out=-40,in=60] (zero2) to (two1.north east);
      \draw [myarrowA] (two1) to (three1);
      \draw [myarrowA] (one1) to (two1);
      \begin{scope}[on background layer]
      \node [mygroup] (zero) at (0*\x,0*\y){}; 
      \node [mygroup] (one) at (0*\x,-1*\y){}; 
      \node [mygroup] (two) at (0*\x,-2*\y){}; 
      \node [mygroup] (three) at (0*\x,-3*\y){}; 
      \end{scope}
    \end{tikzpicture}
    \raisebox{5.3\height}{$\xRightarrow{\mathit{elim}}$}
    \begin{tikzpicture}[scale=1]
      \def\x{1.58}
      \tikzstyle{mygroup} = [mygroup0, minimum width = 175]
      \node [myboxA] (zero1) at (-1*\x,0*\y) {$x_0^3 - x_0$};
      \node [myboxA] (zero2) at (.9*\x,0*\y) {$x_0-1$};
      \node [myboxA](one1) at (0*\x,-1*\y) {$x_1 - x_2, x_2^2-x_2$};
      \node [myboxA](two1) at (-1*\x,-2*\y) {$x_2^2-x_2, x_2x_3^2-x_3,x_2$};
      \node [myboxA](two2) at (.9*\x,-2*\y) {$x_2^2-x_2, x_2x_3^2-x_3,x_2-1$};
      \node [myboxA](three1) at (0*\x,-3*\y) {$0$};
      \draw [myarrowA] (zero1) to (two1);
      \draw [myarrowA] (zero2) to (two2);
      \draw [myarrowA] (one1) to (two1);
      \draw [myarrowA] (one1) to (two2);
      \draw [myarrowA] (two1) to (three1);
      \draw [myarrowA] (two2) to (three1);
      \begin{scope}[on background layer]
      \node [mygroup] (zero) at (0*\x,0*\y){}; 
      \node [mygroup] (one) at (0*\x,-1*\y){}; 
      \node [mygroup] (two) at (0*\x,-2*\y){}; 
      \node [mygroup] (three) at (0*\x,-3*\y){}; 
      \end{scope}
    \end{tikzpicture}
    \\\vspace{10pt}
    \raisebox{3.5\height}{$\xRightarrow{\mathit{elim}}$}
    \begin{tikzpicture}[scale=1]
      \def\x{1.58}
      \tikzstyle{mygroup} = [mygroup0, minimum width = 175]
      \node [myboxA] (zero1) at (-1*\x,0*\y) {$x_0^3 - x_0$};
      \node [myboxA] (zero2) at (.9*\x,0*\y) {$x_0-1$};
      \node [myboxA](one1) at (0*\x,-1*\y) {$x_1 - x_2$};
      \node [myboxA](two1) at (-1*\x,-2*\y) {$x_2^2-x_2, x_2x_3^2-x_3,x_2$};
      \node [myboxA](two2) at (.9*\x,-2*\y) {$x_2^2-x_2, x_2x_3^2-x_3,x_2-1$};
      \node [myboxA](three1) at (0*\x,-3*\y) {$0$};
      \draw [myarrowA] (zero1) to (two1);
      \draw [myarrowA] (zero2) to (two2);
      \draw [myarrowA] (one1) to (two1);
      \draw [myarrowA] (one1) to (two2);
      \draw [myarrowA] (two1) to (three1);
      \draw [myarrowA] (two2) to (three1);
      \begin{scope}[on background layer]
      \node [mygroup] (zero) at (0*\x,0*\y){}; 
      \node [mygroup] (one) at (0*\x,-1*\y){}; 
      \node [mygroup] (two) at (0*\x,-2*\y){}; 
      \node [mygroup] (three) at (0*\x,-3*\y){}; 
      \end{scope}
    \end{tikzpicture}
    \raisebox{2\height}{$\xRightarrow{\mathit{tria}}$}
    \begin{tikzpicture}[scale=1]
      \def\x{1.3}
      \tikzstyle{mygroup} = [mygroup0, minimum width = 130]
      \node [myboxA] (zero1) at (-1.25*\x,0*\y) {$x_0^3 - x_0$};
      \node [myboxA](zero2) at (1*\x,0*\y) {$x_0 - 1$};
      \node [myboxA](one1) at (-.3*\x,-1*\y) {$x_1 - x_2$};
      \node [myboxA](two1) at (-1.25*\x,-2*\y) {$x_2,\;x_3$};
      \node [myboxA](two2) at (-.3*\x,-2*\y) {$x_2 - 1,x_3$};
      \node [myboxA](two3) at (1*\x,-2*\y) {$x_2 - 1,x_3-1$};
      \node [myboxA](three1) at (-.3*\x,-3*\y) {$0$};
      \draw [myarrowA] (zero1) to (two1);
      \draw [myarrowA] (zero2) to (two2);
      \draw [myarrowA] (zero2) to (two3);
      \draw [myarrowA] (one1) to (two1);
      \draw [myarrowA] (one1) to (two2);
      \draw [myarrowA] (one1) to (two3);
      \draw [myarrowA] (two1) to (three1);
      \draw [myarrowA] (two2) to (three1);
      \draw [myarrowA] (two3) to (three1);
      \begin{scope}[on background layer]
      \node [mygroup] (zero) at (0*\x,0*\y){}; 
      \node [mygroup] (one) at (0*\x,-1*\y){}; 
      \node [mygroup] (two) at (0*\x,-2*\y){}; 
      \node [mygroup] (three) at (0*\x,-3*\y){}; 
      \end{scope}
    \end{tikzpicture}
    \\\vspace{10pt}
    \raisebox{2\height}{$\xRightarrow{\mathit{elim}}$}
    \begin{tikzpicture}[scale=1]
      \def\x{1.1}
      \tikzstyle{mygroup} = [mygroup0, minimum width = 100]
      \node [myboxA] (zero1) at (-1*\x,0*\y) {$x_0^3 - x_0$};
      \node [myboxA](zero2) at (1*\x,0*\y) {$x_0 - 1$};
      \node [myboxA](one1) at (-.1*\x,-1*\y) {$x_1 - x_2$};
      \node [myboxA](two1) at (-1*\x,-2*\y) {$x_2$};
      \node [myboxA](two2) at (-.1*\x,-2*\y) {$x_2 - 1$};
      \node [myboxA](two3) at (1*\x,-2*\y) {$x_2 - 1$};
      \node [myboxA](three1) at (-1*\x,-3*\y) {$x_3$};
      \node [myboxA](three2) at (-.1*\x,-3*\y) {$x_3$};
      \node [myboxA](three3) at (1*\x,-3*\y) {$x_3 - 1$};
      \draw [myarrowA] (zero1) to (two1);
      \draw [myarrowA] (zero2) to (two2.north east);
      \draw [myarrowA] (zero2) to (two3);
      \draw [myarrowA] (one1) to (two1);
      \draw [myarrowA] (one1) to (two2);
      \draw [myarrowA] (one1) to (two3);
      \draw [myarrowA] (two1) to (three1);
      \draw [myarrowA] (two2) to (three2);
      \draw [myarrowA] (two3) to (three3);
      \begin{scope}[on background layer]
      \node [mygroup] (zero) at (0*\x,0*\y){}; 
      \node [mygroup] (one) at (0*\x,-1*\y){}; 
      \node [mygroup] (two) at (0*\x,-2*\y){}; 
      \node [mygroup] (three) at (0*\x,-3*\y){}; 
      \end{scope}
    \end{tikzpicture}
    \raisebox{2\height}{$\xRightarrow{\mathit{mergeIn}}$}
    \hspace{-40pt}\raisebox{0.3\height}{$\xRightarrow{\mathit{mergeOut}}$}
    %\raisebox{2\height}{
    %  \begin{align*}
    %  \xRightarrow{\mathit{mergeIn}}\\
    %  \xRightarrow{\mathit{mergeOut}}
    %  \end{align*}
    %}
    \begin{tikzpicture}[scale=1]
      \def\x{.7}
      \tikzstyle{mygroup} = [mygroup0, minimum width = 85]
      \node [myboxA] (zero1) at (-1*\x,0*\y) {$x_0^3 - x_0$};
      \node [myboxA](zero2) at (1*\x,0*\y) {$x_0 - 1$};
      \node [myboxA](one1) at (0*\x,-1*\y) {$x_1 - x_2$};
      \node [myboxA](two1) at (-1*\x,-2*\y) {$x_2$};
      \node [myboxA](two2) at (1*\x,-2*\y) {$x_2 - 1$};
      \node [myboxA](three1) at (-1*\x,-3*\y) {$x_3$};
      \node [myboxA](three2) at (1*\x,-3*\y) {$x_3 - 1$};
      \draw [myarrowA] (zero1) to (two1);
      \draw [myarrowA] (zero2) to (two2);
      \draw [myarrowA] (one1) to (two1);
      \draw [myarrowA] (one1) to (two2);
      \draw [myarrowA] (two1) to (three1);
      \draw [myarrowA] (two2) to (three1);
      \draw [myarrowA] (two2) to (three2);
      \begin{scope}[on background layer]
      \node [mygroup] (zero) at (0*\x,0*\y){}; 
      \node [mygroup] (one) at (0*\x,-1*\y){}; 
      \node [mygroup] (two) at (0*\x,-2*\y){}; 
      \node [mygroup] (three) at (0*\x,-3*\y){}; 
      \end{scope}
      \node [myboxA0,left=.2 of zero] (zero0) {$0$};
      \node [myboxA0,left=.2 of one] (one0) {$1$};
      \node [myboxA0,left=.2 of two] (two0) {$2$};
      \node [myboxA0,left=.2 of three] (three0) {$3$};
      \draw [myarrowA0,out=-130,in=140] (zero0) to (two0);
      \draw [myarrowA0,bend right] (one0) to (two0);
      \draw [myarrowA0,bend right] (two0) to (three0);
    \end{tikzpicture}
    \caption{Chordal triangularization from \Cref{exmp:triang4vars}.}
    \label{fig:triang4vars}
  \end{figure}
  Consider the polynomial set $F=\{x_0^3-x_0,x_0x_2-x_2,x_1-x_2,x_2^2-x_2,x_2x_3^2-x_3\}$, whose associated graph is the star graph ($x_2$ is connected to $x_0,x_1,x_3$).
  \Cref{fig:triang4vars} illustrates a sequence of operations (triangulation, elimination and merge) performed on its induced chordal network.
  The chordal network obtained has three chains:
  \begin{align*}
    (x_0^3-x_0,x_1-x_2,x_2,x_3), \;
    (x_0-1,x_1-x_2,x_2-1,x_3), \;
    (x_0-1,x_1-x_2,x_2-1,x_3-1).
  \end{align*}
  These chains give triangular decomposition of~$F$.
\end{example}

\Cref{alg:chordtriang} presents the chordal triangularization method.
The input consists of a polynomial set $F\subset \K[X]$ and a chordal graph~$G$.
As in the above example, the output of the algorithm is always a triangular chordal network, and it encodes a triangular decomposition of the given polynomial set~$F$.

\begin{algorithm}
  \caption{Chordal Triangularization}
  \label{alg:chordtriang}
  \begin{algorithmic}[1]
    \Require{Polynomial set $F\subset\K[X]$ supported on a chordal graph $G$}
    \Ensure{Triangular $G$-chordal network $\mathcal{N}$ such that $\V(\mathcal{N})=\V(F)$}
    \Procedure{ChordalNet}{$F,G$}
    \State $\mathcal{N}:=$ induced $G$-chordal network of $F$
    \For {$l = 0:n-1$}
    \For {$F_l$ node of $\mathcal{N}$ of rank $l$}
    \State \Call{Triangulate}{$F_l$}\label{line:triangulate}
    \EndFor
    \State \Call{MergeOut}{$l$}\label{line:mergeoutl}
    \If {$l<n-1$}
    \State $x_p :=$ parent of $x_l$
    \For {$T$ node of $\mathcal{N}$ of rank $l$}
    \State \Call{Eliminate}{$T$}
    \EndFor
    \State \Call{MergeOut}{$p$}\label{line:mergeoutp}
    \EndIf
    \State \Call{MergeIn}{$l$}\label{line:mergein}
    \EndFor
    \State \Return $\mathcal{N}$
    \EndProcedure
  \end{algorithmic}
\end{algorithm}

\subsection{Algorithm analysis}

The objective of this section is to prove that, when the input $F$ is chordally zero-dimensional (\Cref{defn:chordallyzerodim}), \Cref{alg:chordtriang} produces a $G$-chordal network, whose chains give a triangular decomposition of $F$.
As described below, the chordally zero-dimensional assumption is only needed in order for the algorithm to be \emph{well-defined} (recall that up to this point we have only defined triangular decompositions of zero-dimensional systems).
Later in the paper we will see how to extend \Cref{alg:chordtriang} to arbitrary ideals.

\begin{definition}
  Let $\mathcal{N}$ be a chordal network, and let $C=(F_0,\ldots,F_{n-1})$ be a chain.
  The \emph{variety} of the chain is $\V(C):=\V(F_0\cup\cdots\cup F_{n-1})$.
  The \emph{variety} $\V(\mathcal{N})$ of the chordal network is the union of $\V(C)$ over all chains $C$.
\end{definition}

\begin{theorem}\label{thm:chordtriang}
  Let~$F\subset\K[X]$, supported on chordal graph $G$, be chordally zero-dimensional.
  \Cref{alg:chordtriang} computes a $G$-chordal network whose chains give a triangular decomposition of~$F$.
\end{theorem}

We will split the proof of \Cref{thm:chordtriang} into several lemmas.
We first show that the algorithm is {well-defined}, i.e., we only perform triangulation operations (line~\ref{line:triangulate}) on nodes $F_l$ that define zero-dimensional ideals.

\begin{lemma}\label{thm:maximalclique}
  Let $F\subset\K[X]$ be chordally zero-dimensional.
  Then in \Cref{alg:chordtriang} every triangulation operation is performed on a zero-dimensional ideal.
\end{lemma}
\begin{proof}
  See \Cref{s:zerodimproofs}.
\end{proof}
%\begin{remark}
%  There are polynomial sets for which \Cref{alg:chordtriang} is well-defined (i.e., the conclusion of the above lemma holds) but are not chordally zero-dimensional (e.g., $F = \{x_0^2+1,x_0x_1+1,x_1^2+x_1x_2\}$).
%\end{remark}

We now show that the chordal structure is preserved during the algorithm.

\begin{lemma}\label{thm:preservechordal}
  Let $\mathcal{N}$ be a $G$-chordal network.
  Then the result of performing a triangulation or elimination operation is also a $G$-chordal network.
\end{lemma}
\begin{proof}
  Consider first a triangulation operation.
  Note that if $F_l\subset\K[X_l]$ then each $T$ in a triangular decomposition is also in $\K[X_l]$.
  Consider now an elimination operation.
  Let $T\subset \K[X_l]$ and $F_p\subset \K[X_p]$ be two adjacent nodes.
  Using \Cref{thm:cliquecontainment}, $\elim{p}{T}\subset \K[X_l\setminus \{x_l\}]\subset \K[X_p]$.
  Thus, the new node $F_p':= F_p\cup \elim{p}{T}\subset \K[X_p]$.
  It is clear that for both operations the layered structure of $\mathcal{N}$ is preserved (i.e., arcs follow the elimination tree).
\end{proof}

We next show that the chains of the output network are triangular sets.

\begin{lemma}\label{thm:outputtriangular}
  The output of \Cref{alg:chordtriang} is a triangular $G$-chordal network.
\end{lemma}
\begin{proof}
  Let $T\subset\K[X_l]$ be a rank $l$ node for which we will perform an elimination operation.
  Note that $T$ must be triangular as we previously performed a triangulation operation.
  Therefore, there is a unique polynomial $f\in T$ with $\mvar(f)=x_l$.
  When we perform the elimination operation this is the only polynomial of $T$ we keep, which concludes the proof.
\end{proof}

Finally, we show that the variety is preserved during the algorithm.

\begin{lemma}\label{thm:preservevariety}
  Let $\mathcal{N}$ be the output of \Cref{alg:chordtriang}.
  Then $\V(\mathcal{N})=\V(F)$, and moreover, any two chains of $\mathcal{N}$ have disjoint varieties.
\end{lemma}
\begin{proof}
  Let us show that the variety is preserved when we perform triangulation, elimination and merge operations.
  Firstly, note that a merge operation does not change set of chains of the network, so the variety is preserved.
  Consider now the case of a triangulation operation.
  Let $\mathcal{N}$ be a chordal network and let $F$ be one of its nodes.
  Let $\T$ be a triangular decomposition of $F$, and let $\mathcal{N}'$ be the chordal network obtained after replacing $F$ with~$\T$.
  Let $C$ be a chain of $\mathcal{N}$ containing $F$, and let $C'=C\setminus \{F\}$.
  Then
  \begin{align*}
    \V(C) = \V(C')\cap \V(F) 
    = \V(C') \cap (\bigsqcup_{T\in \T}\V(T))
    = \bigsqcup_{T\in \T} \V(C') \cap \V(T)
    = \bigsqcup_{T\in \T} \V(C'\cup \{T\}).
  \end{align*}
  Note that $C'\cup\{T\}$ is a chain of $\mathcal{N}'$.
  Moreover, all chains of $\mathcal{N}'$ that contain one of the nodes of $\T$ have this form.
  Thus, the triangulation step indeed preserves the variety.

  Finally, consider the case of an elimination operation.
  Let $T\subset \K[X_l]$ be a node, let $(T,F_p)$ be an arc and let $F_p'=F_p\cup \elim{p}{T}, T_l=T\setminus\elim{p}{T}$.
  Let $\mathcal{N}'$ be the network obtained after an elimination step on $T$.
  It is clear that
  \begin{align*}
    \V(T\cup F_p)
    = \V(T_l\cup \elim{p}{T}\cup F_p)= \V(T_l\cup F_p').
  \end{align*}
  Since a chain in $\mathcal{N}$ containing $T,F_p$ turns into a chain in $\mathcal{N}'$ containing $T_l,F_p'$, we conclude that the elimination step also preserves the variety.
\end{proof}

\begin{proof}[Proof of \Cref{thm:chordtriang}]
  We already proved the theorem, since we showed that:
  the algorithm is well-defined (\Cref{thm:maximalclique}), 
  chordal structure is preserved (\Cref{thm:preservechordal}) and 
  the chains in the output are triangular sets (\Cref{thm:outputtriangular}) that decompose the given variety (\Cref{thm:preservevariety}).
\end{proof}

\subsection{Radical and irreducible decompositions}\label{s:radicalirreducible}

We just showed that \Cref{alg:chordtriang} can compute chordal network representations of some zero-dimensional problems.
However, we sometimes require additional properties of the chordal network.
In particular, in \Cref{s:membership} we will need squarefree representations, i.e., such that any chain generates a radical ideal.
As shown next, we can obtain such representations by making one change in \Cref{alg:chordtriang}: 
whenever we perform a triangulation operation, we should produce a squarefree decomposition.

\begin{proposition}\label{thm:radicalnetwork}
  Assume that all triangular decompositions computed in \Cref{alg:chordtriang} are squarefree.
  Then any chain of the output network generates a radical ideal.
\end{proposition}
\begin{proof}
  See \Cref{s:zerodimproofs}
\end{proof}

Instead of radicality, we could further ask for an irreducible representation, i.e., such that any chain generates a prime ideal.
The obvious modification to make is to require all triangulation operations to produce irreducible decompositions.
Unfortunately, this does not always work.
Indeed, we can find irreducible univariate polynomials $f\subset \K[x_0]$, $g\subset\K[x_1]$ such that $\ideal{f,g}\subset \K[x_0,x_1]$ is not prime (e.g., $f=x_0^2+1, g =x_1^2+1$).

Nonetheless, there is an advantage of maintaining prime ideals through the algorithm: it gives a simple bound on the size of the triangular network computed, as shown next.
This bound will be used when analyzing the complexity of the algorithm.

\begin{lemma}\label{thm:boundW}
  Assume that all triangular decompositions computed in \Cref{alg:chordtriang} are irreducible.
  Then the number of rank~$l$ nodes in the output is at most $|\V(F\cap \K[X_l])|$.
\end{lemma}
\begin{proof}
  Let us see that there are at most $|\V(F\cap \K[X_l])|$ rank~$l$ nodes after the merge operation from line~\ref{line:mergeoutl}.
  First note that when we perform this operation any rank~$l$ node has an outgoing arc to all rank~$p$ nodes (where $x_p$ is the parent of $x_l$).
  Therefore, this operation merges any two rank~$l$ nodes that define the same ideal.
  Since these ideals are all maximal, then for any two distinct nodes $T_l,T_l'$ we must have $\V(T_l)\cap \V(T_l')=\emptyset$.
  Also note that both $\V(T_l), \V(T_l')$ are subsets of $\V(F\cap\K[X_l])$.
  The lemma follows.
\end{proof}

\begin{remark}
  There are other ways to achieve the above bound that do not require computing irreducible decompositions.
  For instance, we can force the varieties $\V(T_l),\V(T_l')$ to be disjoint by using ideal saturation.
\end{remark}

\subsection{Complexity}

We proceed to estimate the cost of \Cref{alg:chordtriang} in the chordally zero-dimensional case.
We will show that the complexity\footnote{
Here the complexity is measured in terms of the number of field operations.}
is $O(n\,q^{O(\kappa)})$, where $\kappa$ is the treewidth (or clique number) of the graph, and $q$ is a certain degree bound on the polynomials that we formalize below.
In particular, when the treewidth~$\kappa$ is bounded the complexity is linear in~$n$ and polynomial in the degree bound~$q$.

\begin{definition}[$q$-domination]
  We say that a polynomial set $F_l\subset \K[X_l]$ is \emph{$q$-dominated} if for each $x_i\in X_l$ there is some $f\in F_l$ such that $\mvar(f)=x_i$, $\init(f)=1$ and $\deg(f,x_i)\leq q$.
  Let $F\subset \K[X]$ be supported on a chordal graph~$G$.
  We say that $F$ is \emph{chordally} $q$-dominated if $F\cap \K[X_l]$ is $q$-dominated for each maximal clique $X_l$ of graph~$G$.
\end{definition}

\begin{example}
  The coloring equations in~\eqref{eq:colorings} are chordally $q$-dominated since the equations $x_i^q-1$ are present.
  Another important example is the case of finite fields $\F_q$, since if we include the equations $x_i^q-x_i$, as is often done, the problem becomes chordally $q$-dominated.
\end{example}

\begin{remark}
  Observe that if $F$ is chordally $q$-dominated then it is also chordally zero-dimensional.
  Conversely, if $F$ is chordally zero-dimensional then we can apply a simple transformation to make it chordally $q$-dominated (for some $q$).
  Concretely, for each maximal clique $X_l$ we can enlarge $F$ with a Gr\"obner basis of $F\cap \K[X_l]$.
\end{remark}

We note that we also used the $q$-dominated condition in~\cite{Cifuentes2014} to analyze the complexity of chordal elimination.
The importance of this condition is that it allows us to easily bound the complexity of computing Gr\"obner bases or triangular decompositions, as stated next.

\begin{proposition}\label{thm:qdominatedconstant}
  For any $q$-dominated polynomial set on $k$ variables, the complexity of computing Gr\"obner bases and (squarefree/irreducible) triangular decompositions is~$q^{O(k)}$.
\end{proposition}
\begin{proof}
  See \Cref{s:zerodimproofs}.
\end{proof}

The above proposition gives us the cost of the triangulation operations.
However, we need to ensure that these operations are indeed performed on a $q$-dominated ideal, as shown next.

\begin{lemma}\label{thm:maximalclique2}
  Let $F\subset \K[X]$ be chordally $q$-dominated.
  Then in \Cref{alg:chordtriang} any triangulation operation is performed on a $q$-dominated ideal.
\end{lemma}
\begin{proof}
  The proof is analogous to the one of \Cref{thm:maximalclique}.
\end{proof}

We are ready to estimate the complexity of chordal triangularization.
For the analysis we assume that the merge operation from line~\ref{line:mergeoutl} (resp.\ line~\ref{line:mergeoutp}) is performed simultaneously with the triangulation (resp.\ elimination) operations, i.e., as soon as we create a new node we compare it with the previous nodes of the same rank to check if it is repeated.

\begin{lemma}\label{thm:boundW2}
  Let $F\subset \K[X]$ be chordally $q$-dominated.
  Assume that all triangular decompositions computed in \Cref{alg:chordtriang} are irreducible.
  Then throughout the algorithm the width of the network is always bounded by $q^\kappa$, independent of the number of variables.
\end{lemma}
\begin{proof}
  This is a consequence of \Cref{thm:boundW}.
  See \Cref{s:zerodimproofs}.
\end{proof}
\begin{remark}[Chordal network of linear size]\label{thm:linearnetwork}
  It follows from the lemma that for fixed~$q, \kappa$, any chordally $q$-dominated $F\subset\K[X]$ of treewidth $\kappa$ has a chordal network representation with $O(n)$ nodes.
\end{remark}

\begin{theorem}\label{thm:chordtrianglinear}
  Let $F\subset \K[X]$ be chordally $q$-dominated.
  The complexity of chordal triangularization is ${O}(nWq^{O(\kappa)})$, where $W$ is a bound on the width of the network throughout the algorithm.
  If all triangulation operations are irreducible, the complexity is $O(n\,q^{O(\kappa)})$.
\end{theorem}
\begin{proof}
  From \Cref{thm:qdominatedconstant} and \Cref{thm:maximalclique2} we know that each triangulation operation takes $q^{O(\kappa)}$, and thus the cost of all triangulations is $O(nWq^{O(\kappa)})$.
  The cost of the elimination operations is negligible.
  As for the merging operations, we can efficiently verify if a new node is repeated by using a hash table.
  Thus, the cost of the merging operation is also negligible.
  Finally, if all triangulation operations are irreducible, then $W\leq q^k$ because of \Cref{thm:boundW2}.
\end{proof}

\begin{remark}[Beyond chordally zero-dimensional]\label{rem:beyondzerodim}
  We will later see that, after a suitable redefinition of the triangulation step, Algorithm~\ref{alg:chordtriang} can also be applied to arbitrary ideals.
  Nonetheless, the complexity bounds from above do depend on the special structure of the chordally zero-dimensional case.
  Indeed, solving polynomial equations of treewidth one is NP-hard~\cite[Ex~1.1]{Cifuentes2014}, and counting their number of solutions is $\sharp$P-hard (even in the \emph{generic} case for treewidth two~\cite[Prop~24]{Cifuentes2016}).
  As a consequence, chordal triangularization will not always run in polynomial time.
  When using \Cref{alg:chordtriang} in such hard instances we may end up with very high degree polynomials or with a very large number of nodes.
\end{remark}

\section{Computing with chordal networks}\label{s:membership}

Triangular decompositions are one of the most common tools in computational algebraic geometry.
The reason is that there are many good algorithms to compute them, and that they can be used to derive several properties of the underlying variety.
However, as seen in \Cref{exmp:triangcycle}, the size of the decomposition obtained might be extremely large (exponential) even for very simple cases.
Chordal networks can provide a compact representation for these large decompositions.
We will see how to effectively use the data structure of chordal networks to compute several properties of the variety.

Let $I=\ideal{F}$ be a zero-dimensional ideal.
We consider the following problems.
\begin{description}
  \item[Elimination] Describe the projection of $\V(I)$ onto the last $n-l$ coordinates.
  \item[Zero count] Determine the number of solutions, i.e., the cardinality of $\V(I)$.
  \item[Sampling] Sample random points from $\V(I)$ uniformly.
  \item[Radical membership] Determine if a polynomial $h\in\K[X]$ vanishes on $\V(I)$, or equivalently, determine if $h\in \sqrt{I}$.
\end{description}

In this section we will develop efficient algorithms for the above problems, given a squarefree chordal network $\mathcal{N}$ (with possibly exponentially many chains).
Recall that such network can be obtained as explained in \Cref{thm:radicalnetwork}.
We will see that the first three problems can be solved relatively easily.
The radical membership problem is more complicated, and most of this section will be dedicated to it.
We note that the algorithms for elimination and radical membership will naturally extend to the positive-dimensional case.

\subsection{Elimination}\label{s:elimination}
The elimination problem is particularly simple, thanks to the elimination property of lexicographic Gr\"obner bases.
For an arbitrary chordal network $\mathcal{N}$, let $\mathcal{N}_{\geq l}$ denote the subset of $\mathcal{N}$ consisting of nodes of rank~$k$ with $k\geq l$.
Then $\mathcal{N}_{\geq l}$ is a chordal network representation of the projection of $\V(I)$ onto the last $n-l$ coordinates.

\subsection{Counting solutions}\label{s:countsolutions}
We want to determine $|\V(\mathcal{N})|$ for a squarefree chordal network~$\mathcal{N}$.
Recall from \Cref{eq:degT} that $|\V(T)|=\deg(T)$ for a squarefree triangular set~$T$.
Therefore, we just need to compute the sum of $\deg(C)$ over all chains $C$ of the network.
We can do this efficiently via dynamic programming, as explained in the following example.

\begin{example}[Zero count]\label{exmp:countsolutions}
  Let us determine $|\V(\mathcal{N})|$ for the chordal network from \Cref{fig:triang10vars}, which corresponds to counting $4$-colorings for the blue/solid graph from \Cref{fig:graph10notchordal}.
  For a rank~$l$ node $f_l$ of the network, let its weight $w(f_l)$ be its degree in $x_l$.
  Then we just need to compute $\sum_{C}\prod_{f_l\in C}w(f_l)$ where the sum is over all chains of the network.
  We can do this efficiently by successively eliminating the nodes of the network.

  Let us first eliminate the nodes of rank $0$.
  Let $f_0^a,f_0^b$ be the two nodes of rank $0$, with weights $w(f_0^a)=3,w(f_0^b)=2$.
  Let $f_6^a,f_6^b,f_6^c$ be the nodes of rank $6$, with weights $w(f_6^a)=w(f_6^c)=1,w(f_6^b)=2$.
  Note that any chain containing $f_6^a$ must also contain $f_0^a$.
  Therefore, we can remove the arc $(f_0^a,f_6^a)$ and update the weight $w(f_6^a)=1\times 3$.
  Similarly, any chain containing $f_6^b$ (or $f_6^c$) must contain also $f_0^b$.
  So we may delete the arcs $(f_0^b,f_6^b)$ and $(f_0^b,f_6^c)$ and update the weights $w(f_6^b)=2\times 2$, $w(f_6^c)=1\times 2$.
  By doing this, we have disconnected, or eliminated, all nodes of rank $0$.
  Continuing this procedure, the final weights obtained for each rank are shown below. 
  The number of solutions is the last number computed: $10968$.
  \begin{gather*}
    \rank(0)\rightarrow [3, 2],\quad
    \rank(1)\rightarrow [3, 2],\quad
    \rank(2)\rightarrow [3, 2],\quad
    \rank(3)\rightarrow [3, 2, 4],\quad\\ %order 0,2,1
    \rank(4)\rightarrow [3, 2, 4],\quad %order 0,2,1
    \rank(5)\rightarrow [50, 25, 20, 20, 16],\quad %order 0,1,3,2,4
    \rank(6)\rightarrow [3, 4, 2],\quad\\
    \rank(7)\rightarrow [264, 650],\quad
    \rank(8)\rightarrow [2742],\quad
    \rank(9)\rightarrow [10968].
  \end{gather*}
\end{example}

\begin{algorithm}[htb]
  \caption{Count solutions}
  \label{alg:degree}
  \begin{algorithmic}[1]
    \Require{Chordal network $\mathcal{N}$ (triangular, squarefree)}
    \Ensure{Cardinality of $\V(\mathcal{N})$}
    \Procedure{ZeroCount}{$\mathcal{N}$}
    \For {$f$ node of $\mathcal{N}$}
    \State $w(f):= (\mdeg(f)$ \algorithmicif\ $x_{\rank(f)}$ is a leaf \algorithmicelse\ $0)$
    \EndFor
    \For {$l = 0:n-1$}
    \For {$(f_l,f_p)$ arc of $\mathcal{N}$ with $\rank(f_l)=l$}
    \State $w(f_p):=w(f_p)+w(f_l)\mdeg(f_p)$
    \EndFor
    \EndFor
    \State \Return sum of $w(f_{n-1})$ over all nodes of rank $n-1$
    \EndProcedure
  \end{algorithmic}
\end{algorithm}

\Cref{alg:degree} generalizes the above example to arbitrary chordal networks.
The complexity is $O(nW^2)$, since we perform one operation for each arc of the network.

\subsection{Sampling solutions}\label{s:sampling}
Uniformly sampling solutions can be done quite easily, by using the partial root counts computed in \Cref{alg:degree}.
Instead of giving a formal description we simply illustrate the procedure with an example.

\begin{example}[Sampling]\label{exmp:sampling}
  Consider again the chordal network of \Cref{fig:triang10vars}.
  We want to uniformly sample a point $(\hat{x}_0,\ldots,\hat{x}_9)$ from its variety, and we follow a bottom up strategy.
  Let us first choose the value $\hat{x}_9$.
  Since there is a unique rank~$9$ node $f_9=x_9^4-1$, then $\hat{x}_9$ must be one of its four roots.
  Note that each of those roots extend to $2742$ solutions (a fourth of the total number of solutions).
  Therefore, $\hat{x}_9$ should be equally likely to be any of these roots.
  Given the value of $\hat{x}_9$, we can now set $\hat{x}_8$ to be any of the three roots of $f_8=x_8^3+x_8^2x_9+x_8x_9^2+x_9^3$, each equally likely.
  Consider now the two rank~$7$ nodes $f_7^a,f_7^b$ of degrees $1$ and $2$.
  Note that $\hat{x}_7$ should be either a root of $f_7^a$ or a root of $f_7^b$ (for the given values of $\hat{x}_8, \hat{x}_9$).
  In order to sample uniformly, we need to know the number of solutions that each of those values extend to.
  From \Cref{exmp:countsolutions} we know that $f_7^a$ leads to $264$ points on the variety, and $f_7^b$ leads to $650$.
  Therefore, we can decide which of them to use based on those weights.
  Assuming we choose $f_7^b$, we can now set $\hat{x}_7$ to be any of its two roots, each equally likely.
  It is clear how to continue.
\end{example}

\subsection{Radical membership}\label{s:idealmembership}
In the radical ideal membership problem we want to check whether $h\in \K[X]$ vanishes on $\V(\mathcal{N})$.
This is equivalent to determining whether for each chain $C$ of $\mathcal{N}$ the normal form $h_C := h\bmod C$ is identically zero.
We will propose a Monte Carlo algorithm to efficiently test this property (without iterating over all chains) under certain structural assumptions on the polynomial~$h$.
Our main result is the following.

\begin{theorem}[Radical membership]\label{thm:memberlinear}
  Let~$F\subset\K[X]$ be chordally $q$-dominated.
  Let~$\mathcal{N}$ be a chordal network representation of~$F$ of width~$W$.
  Let~$h$ be a polynomial that decomposes as $h=\sum_l h_l$ with $h_l\subset \K[X_l]$.
  There is a Monte Carlo algorithm that determines whether $h$ vanishes on $\V({F})$ in $\widetilde{O}(n W q^{2\kappa} + n W^2 q^{\kappa})$.
  Here the notation~$\widetilde{O}$ ignores polynomial factors in the clique number~$\kappa$.
\end{theorem}

\begin{remark}
  The theorem is restricted to polynomials~$h$ that preserve some of the structure of the graph $G$, although they may involve all the variables in the ring $\K[X]$ (as opposed to the polynomials of the chordal network).
  The above mentioned Monte Carlo algorithm also works for other types of polynomials~$h$, but we do not prove complexity bounds for them.
\end{remark}

We point out that the above complexity result is far from trivial.
To justify this claim we can show that a simple variation of the radical membership problem is NP-hard under very mild assumptions.

\begin{example}[Zero divisor problem]
  Consider the zero divisor problem: determine if a polynomial $h\in\K[X]$ vanishes on at least one point of $\V(I)$.
  Also consider the NP-complete subset sum problem: decide if a set of integers $A=\{a_0,\ldots,a_{n-1}\}$ contains a subset whose sum is some given value~$S$.
  We can reduce it to the zero divisor problem by considering the ideal $I := \ideal{x_i(x_i-a_i): 0\leq i<n}$ and the polynomial $h := \sum_i x_i - S$.
  Note that the associated graph is the completely disconnected graph ($\kappa=1$) and thus its induced chordal network is already triangular ($W=1$, $q=2$).
\end{example}

We proceed to derive our radical ideal membership test.
We will initially assume that the variables of~$h$ are all contained in a path of the elimination tree.
Later, we will extend the algorithm to polynomials~$h$ that decompose into multiple paths of the elimination tree.
Finally, we will prove the complexity bound from \Cref{thm:memberlinear}.

\subsubsection*{Membership on a path}
Consider the case where the elimination tree of the graph~$G$ is a path (i.e., it has only one leaf).
Alternatively, we can assume that all the variables of~$h$ are contained in a path of the elimination tree.
As before, let $h_C:= h\bmod C$ denote the normal form with respect to chain~$C$.
Our radical ideal membership test is based on two simple ideas.
Firstly, we will check whether the polynomial $H(X):=\sum_C r_C \,h_C(X)$ is identically zero, for some random coefficients $r_C\in \K$.
Clearly, for sufficiently generic values of~$r_C$, the polynomial $H(X)$ will be zero if and only if each $h_C$ is zero.
The second idea is that we evaluate $H(X)$ in some random points $\hat{x}_i\in \K$.
Thus, we just need to check whether the scalar $H(\hat{X})\in \K$ is zero.
We illustrate how the algorithm works through the following example.

\begin{example}[Radical membership]\label{exmp:prem10vars}
  Consider again the chordal network of \Cref{fig:triang10vars}.
  Let us verify that the polynomial $h(x)$ from \Cref{fig:prem10vars} vanishes on its variety.
  %Let us verify that the following polynomial vanishes on its variety:
  %{\footnotesize
  %\begin{align*}
  %  h(x)=x_0^2x_6 - x_0^2x_7 - x_0x_6x_8 - x_0x_6x_9 - x_0x_7^2 - x_0x_8^2 - x_0x_8x_9 - x_0x_9^2 + x_6x_8x_9 - x_7^3 + x_8^2x_9 + x_8x_9^2.
  %\end{align*}
  %}
  We need to show that the reduction (normal form) of $h$ by each chain of the network is zero.
  As in the case of counting solutions, we will achieve this by successively eliminating nodes.
  Note that the variables of $h$ are $\{x_0,x_6,x_7,x_8,x_9\}$, which correspond to a path of the elimination tree.
  Thus, we restrict ourselves to the part of the network given by these variables, as shown in \Cref{fig:prem10vars}.

  \begin{figure}[!htb]
    \centering

    \begin{tikzpicture}[scale=1]
      \def\x{2.6}
      \def\y{1.1}
      \tikzstyle{myboxA} = [scale=0.65,draw=purple,fill=orange!8,thick,rectangle,rounded corners];
      %\tikzstyle{myboxB} = [scale=0.65,draw=black!20!orange,fill=orange!8,thick,rectangle,rounded corners];
      \tikzstyle{myarrowA} = [scale=0.65,purple,thick,->];
      %\tikzstyle{myarrowB} = [black!20!orange,thick,->];
      \tikzstyle{mygroup} = [draw=blue,fill = blue!15, thick, dotted, minimum width = 340, minimum height = 16];
      \tikzstyle{myboxA0} = [scale=0.65,fill=yellow!80,draw=blue,thick,densely dotted,circle];
      \tikzstyle{myarrowA0} = [blue,thick,dotted,->];
      \tikzstyle{mylabel} = [draw=none,pos=0.6];
      \node [scale=0.67,purple] (h) at (0*\x,-3.9*\y) {
      $h(x)=x_0^2x_6 - x_0^2x_7 - x_0x_6x_8 - x_0x_6x_9 - x_0x_7^2 - x_0x_8^2 - x_0x_8x_9 - x_0x_9^2 + x_6x_8x_9 - x_7^3 + x_8^2x_9 + x_8x_9^2$
      };
      \node [myboxA] (zero0) at (-1*\x,-5*\y) {$x_0^3 + x_0^2x_7 + x_0x_7^2 + x_7^3$};
      \node [myboxA] (zero1) at (1*\x,-5*\y) {$x_0^2 + x_0x_6 + x_0x_7 + x_6^2 + x_6x_7 + x_7^2$};
      \node [myboxA](six0) at (-1.5*\x,-6*\y) {$x_6 - x_7$};
      \node [myboxA](six1) at (0*\x,-6*\y) {$x_6^2 + x_6x_8 + x_6x_9 + x_8^2 + x_8x_9 + x_9^2$};
      \node [myboxA](six2) at (1.5*\x,-6*\y) {$x_6 + x_7 + x_8 + x_9$};
      \node [myboxA](seven0) at (-1*\x,-7*\y) {$x_7 - x_9$};
      \node [myboxA](seven1) at (1*\x,-7*\y) {$x_7^2 + x_7x_8 + x_7x_9 + x_8^2 + x_8x_9 + x_9^2$};
      \node [myboxA](eight0) at (0*\x,-8*\y) {$x_8^3 + x_8^2x_9 + x_8x_9^2 + x_9^3$};
      \node [myboxA](nine0) at (0*\x,-9*\y) {$x_9^4 - 1$};
      \node (aux0) at (-1*\x,-4.1*\y){};
      \node (aux1) at (1*\x,-4.1*\y){};
      \node (aux2) at (0*\x,-9.9*\y){};
      \begin{scope}[every node/.style={scale=.65}]
      \draw [myarrowA] (aux0) -- node [mylabel,left]{$h(x)$} ++ (zero0);
      \draw [myarrowA] (aux1) -- node [mylabel,left]{$h(x)$} ++ (zero1);
      \draw [myarrowA] (zero0)  -- node [mylabel,left=.10*\x]{$1\cdot h_{0}^a$} ++ (six0);
      \draw [myarrowA] (zero1)  -- node [mylabel,left=.20*\x]{$1\cdot h_{0}^b$} ++ (six1);
      \draw [myarrowA] (zero1)  -- node [mylabel,right=.05*\x]{$1\cdot h_{0}^b$} ++ (six2);
      \draw [myarrowA] (six0)  -- node [mylabel,right=.40*\x]{$\frac{1}{2}\cdot h_{6}^a$} ++ (seven1);
      \draw [myarrowA] (six1)  -- node [mylabel,left=.18*\x]{$1\cdot h_{6}^b$} ++ (seven0);
      \draw [myarrowA] (six2)  -- node [mylabel,right=.02*\x]{$\frac{1}{2}\cdot h_{6}^c$} ++ (seven1);
      \draw [myarrowA] (seven0) -- node [mylabel,left=.10*\x]{$\frac{1}{2}\cdot h_7^a$} ++ (eight0);
      \draw [myarrowA] (seven1)  -- node [mylabel,right=.08*\x]{$\frac{1}{2}\cdot h_{7}^b$} ++ (eight0);
      \draw [myarrowA] (eight0) -- node[mylabel,left]{$1\cdot h_8$} ++ (nine0);
      \draw [myarrowA] (nine0) -- node [mylabel,left]{$h_9=0$} ++ (aux2);
    \end{scope}
      \begin{scope}[on background layer]
      \node [mygroup] (zero) at (0*\x,-5*\y){}; 
      \node [mygroup] (six) at (0*\x,-6*\y){}; 
      \node [mygroup] (seven) at (0*\x,-7*\y){}; 
      \node [mygroup] (eight) at (0*\x,-8*\y){}; 
      \node [mygroup] (nine) at (0*\x,-9*\y){}; 
      \end{scope}
      \node [myboxA0,left=.2 of zero] (zeroN) {$0$};
      \node [myboxA0,left=.2 of six] (sixN) {$6$};
      \node [myboxA0,left=.2 of seven] (sevenN) {$7$};
      \node [myboxA0,left=.2 of eight] (eightN) {$8$};
      \node [myboxA0,left=.2 of nine] (nineN) {$9$};
      \draw [myarrowA0,out=-115,in=115] (zeroN) to (sixN);
      \draw [myarrowA0,bend right] (sixN) to (sevenN);
      \draw [myarrowA0,bend right] (sevenN) to (eightN);
      \draw [myarrowA0,bend right] (eightN) to (nineN);
    \end{tikzpicture}
    \vspace{-15pt}
    \caption{Sketch of the radical ideal membership test from \Cref{exmp:prem10vars}.}
    \label{fig:prem10vars}
  \end{figure}
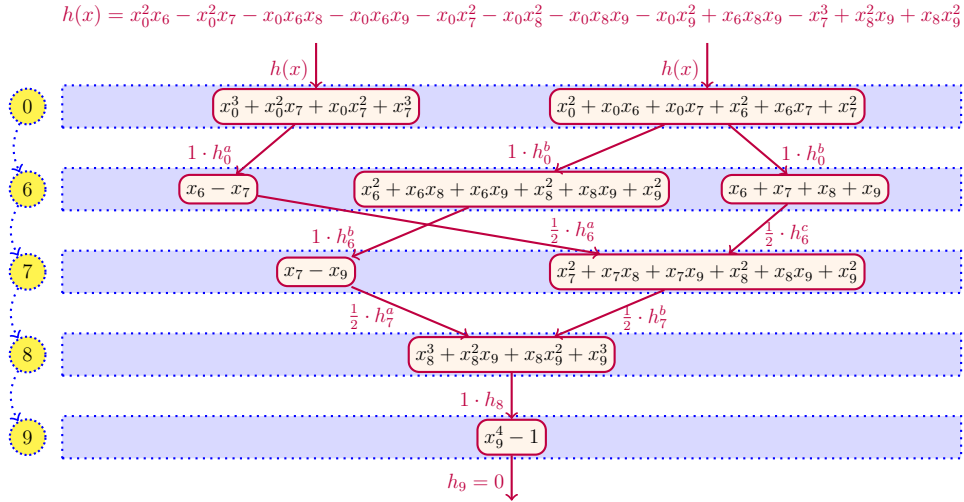

  Let us start by processing the two nodes of rank $0$.
  We have to compute the reduction of $h(x)$ modulo each of these nodes.
  Afterwards, we will substitute $x_0$ in these reduced polynomials with a random value on $\K$; in this case we choose $\hat{x}_0=1$.
  Let $h_0^a,h_0^b$ be the polynomials obtained after the reduction and substitution, as shown in \Cref{fig:prem10vars}.
  These two polynomials will be sent to the adjacent rank~$6$ nodes.

  Consider now a rank~$p$ node~$f_p$ that receives certain polynomials from its adjacent rank~$l$ nodes.
  We now perform a random linear combination of these incoming polynomials, then we reduce this linear combination modulo~$f_p$, and lastly we substitute $x_p$ with a random value~$\hat{x}_p$.
  For this example the linear combination will be an average, and the random points $\hat{x}_p$ will be one.
  \Cref{fig:prem10vars} indicates the polynomials received and output by each node.
  For instance, $h_8$ is obtained by reducing $\frac{1}{2}(h_7^a+h_7^b)$ modulo $f_8=x_8^3+x_8^2x_9+x_8x_9^2+x_9^3$ and then plugging in $\hat{x}_8 = 1$.
  The polynomials obtained with this procedure are shown below. 
  Note that the last polynomial computed is zero, agreeing with the fact that $h(x)$ vanishes on the variety.

  {\footnotesize
  \begin{align*}
    h_0^a &= x_6x_8x_9 - x_6x_8 - x_6x_9 + x_6 - x_7^3 - x_7^2 - x_7 + x_8^2x_9 - x_8^2 + x_8x_9^2 - x_8x_9 - x_9^2 \\
    h_0^b &= -x_6^3 - x_6^2 + x_6x_8x_9 - x_6x_8 - x_6x_9 + x_8^2x_9 - x_8^2 + x_8x_9^2 - x_8x_9 - x_9^2 \\
    h_6^a &= -x_7^3 - x_7^2 + x_7x_8x_9 - x_7x_8 - x_7x_9 + x_8^2x_9 - x_8^2 + x_8x_9^2 - x_8x_9 - x_9^2 \\
    h_6^b &= -x_8^3 - x_8^2x_9 - x_8x_9^2 - x_9^3 \\
    %h_6^c &= x_7^3 + 3x_7^2x_8 + 3x_7^2x_9 - x_7^2 + 3x_7x_8^2 + 5x_7x_8x_9 - x_7x_8 + 3x_7x_9^2 - x_7x_9 + x_8^3 + 3x_8^2x_9  - x_8^2 + 3x_8x_9^2 - x_8x_9 + x_9^3 - x_9^2\\
    h_6^c &= x_7^3 + x_7^2(3x_8 + 3x_9 - 1) + x_7(3x_8^2 + 5x_8x_9 - x_8 + 3x_9^2 - x_9) + x_8^3 + 3x_8^2x_9  - x_8^2 + 3x_8x_9^2 - x_8x_9 + x_9^3 - x_9^2\\
    h_7^a &= h_7^b = -x_8^3 - x_8^2x_9 - x_8x_9^2 - x_9^3 \\
    h_8 &= h_9= 0
  \end{align*}
  }
\end{example}

\Cref{alg:member} generalizes the procedure from the above example.
Observe that each node $f_l$ of the network has an associated polynomial $H(f_l)$, which is first reduced modulo~$f_l$, then we substitute the value $\hat{x}_l$ and finally we pass this polynomial to the adjacent nodes.
Also note that that we choose one random scalar $\hat{x}_l$ for each variable, and one random scalar $r_{lp}$ for each arc of the network.

\begin{algorithm}[htb]
  \caption{Radical ideal membership}
  \label{alg:member}
  \begin{algorithmic}[1]
    \Require{Chordal network $\mathcal{N}$ (triangular, squarefree) and polynomial $h(x)$ such that all its variables are contained in a path of the elimination tree.}
    \Ensure{True, if $h$ vanishes on $\V(\mathcal{N})$. False, otherwise.}
    \Procedure{RIdealMembership}{$\mathcal{N},h$}
    \State $x_m:=\mvar(h)$ 
    \For {$f$ node of $\mathcal{N}$}
    \State\label{line:initialize} $H(f):= (h$ \algorithmicif\ $\rank(f)=m$ \algorithmicelse\ $0)$
    \EndFor
    \For {$l = 0:n-1$}
    \State $\hat{x}_{l}:=$ random scalar
    \For {$f_l$ node of $\mathcal{N}$ of rank $l$}
    \State $H(f_l):=H(f_l)\bmod f_l$
    \State plug in $\hat{x}_l$ in $H(f_l)$
    \For {$(f_l,f_p)$ arc of $\mathcal{N}$}
    \State $r_{lp}:=$ random scalar
    \State $H(f_p):=H(f_p)+r_{lp}\,H(f_l)$
    \EndFor
    \EndFor
    \EndFor
    \State \Return (True \algorithmicif\ {$H(f_{n-1})=0$ for all rank $n-1$ nodes} \algorithmicelse\ False)
    \EndProcedure
  \end{algorithmic}
\end{algorithm}

\subsubsection*{Correctness}
We proceed to show the correctness of \Cref{alg:member}.
We will need a preliminary lemma and some new notation.
For any~$l$, let~$X^l$ denote the subtree of the elimination tree consisting of $x_l$ and all its descendants (e.g., $X^{n-1}$ consists of all variables).
For a rank~$l$ node~$f_l$ of the network, we will say that an \emph{$f_l$-subchain} $C_l$ is the subset of a chain~$C$, with $f_l\in C$, restricted to nodes of rank~$i$ for some $x_i\in X^l$.

\begin{lemma}\label{thm:memberoutput}
  Let $\mathcal{N}$ be a chordal network whose elimination tree is a path, and let $h\in \K[X]$.
  Let $f_l$ be a rank~$l$ node of $\mathcal{N}$.
  In \Cref{alg:member}, the final value of~$H(f_l)$ is given by plugging in the values $\hat{x}_1\hat{x}_2,\ldots,\hat{x}_l$ in the polynomial
  \begin{align*}
    \sum_{C_l} r_{C_l} h \bmod C_l,
  \end{align*}
  where the sum is over all $f_l$-subchains $C_l$, and where $r_{C_l}$ denotes the product of the random scalars $r_{ij}$ along the subchain $C_l$.
\end{lemma}
\begin{proof}
  See \Cref{s:membershipproofs}.
\end{proof}

\begin{theorem}
  Let $\mathcal{N }$ be a chordal network, triangular and squarefree, and let $q$ be a bound on the main degrees of its nodes.
  Let $h\in \K[X]$ be such that all its variables are contained in a path of the elimination tree.
  \Cref{alg:member} behaves as follows:
  \begin{itemize}[leftmargin=.5in]
    \item if $h$ vanishes on $\V(\mathcal{N})$, it always returns ``True''.
    \item if not, it returns ``False'' with probability at least $1/2$, assuming that the random scalars $r_{lp}, x_l$ are chosen (i.i.d.\ uniform) from some set $S\subset \K$ with $|S|\geq 2 n q$.
  \end{itemize}
\end{theorem}
\begin{proof}
  Denoting $h_C:=h\bmod C$, \Cref{thm:memberoutput} tells us that \Cref{alg:member} checks whether $\sum_C r_C h_C(\hat{X}) = 0$, where $r_C$ is the product of all scalars $r_{lp}$ along the chain $C$.
  If $h$ vanishes on $\V(\mathcal{N})$, then each $h_C$ is zero and thus the algorithm returns ``True''.
  Assume now that $h$ does not vanish on $\V(\mathcal{N})$, and thus at least one $h_C$ is nonzero.
  Let $R$ be the set of all random scalars $r_{lp}$ used in the algorithm, which we now see as variables.
  Consider the polynomial
  \begin{align*}
    H(X,R):= \sum_C r_C(R)\, h_C(X),
  \end{align*}
  and note that it is nonzero.
  Observe that the degree of $H(X,R)$ is at most $nq$, since $\deg(r_C)\leq n$ and $\deg(h_C)\leq n (q-1)$.
  Using the Schwartz-Zippel lemma (see e.g.,~\cite[\S6.9]{VonZurGathen2013}), the probability that $H$ evaluates to zero for random values $r_{lp}, \hat{x}_l\in S$ is at most $n q/|S|\leq 1/2$.
\end{proof}

\begin{remark}
  The above theorem requires that $\K$ contains sufficiently many elements.
  If necessary, we may consider a field extension $\mathbb{L}\supset \K$ and perform all computations over $\mathbb{L}[X]$. 
\end{remark}

\subsubsection*{Combining multiple paths}
We now extend \Cref{alg:member} to work for other polynomials~$h$.
Specifically, we assume that the polynomial can be written as $h = \sum_i h_i$ where the variables of each $h_i$ belong to a path of the elimination tree.
We let $x_{m_i}:=\mvar(h_i)$ denote the main variables, and we can assume that they are all distinct.
We only need to make two simple modifications to \Cref{alg:member}.
\begin{enumerate}[leftmargin=.5in,label=(\roman*)]
  \item\label{line:initialize2} Previously, we initialized the algorithm with nonzero values in a single rank (see line~\ref{line:initialize}).
    We now initialize the algorithm in multiple ranks: $H(f_{m_i})=h_i$ if $\rank(f_{m_i}) = m_i$.
  \item\label{line:normalize} When combining the incoming polynomials to a node~$f_p$, we now take a random \emph{affine} combination (i.e., $\sum_{l} r_{lp}H(f_l)$ for some scalars~$r_{lp}$ such that $\sum_{l} r_{lp} = 1$).
    Note that in the example from \Cref{fig:prem10vars} we took the average of the incoming nodes, so this condition is satisfied.
\end{enumerate}
The first modification is quite natural given the decomposition of the polynomial~$h$.
The second item is less intuitive, but it is simply a normalization to ensure that all polynomials~$h_i$ are scaled in the same manner.
The correctness of this modified algorithm follows from the fact that \Cref{thm:memberoutput} remains valid, as shown next.

\begin{lemma}\label{thm:memberoutput2}
  Let $\mathcal{N }$ be a chordal network and let $h=\sum_i h_i\in \K[X]$ be such that the variables of each $h_i$ are contained in a path of the elimination tree.
  With the above modifications to \Cref{alg:member}, the final value of~$H(f_l)$ is as stated in \Cref{thm:memberoutput}.
\end{lemma}
\begin{proof}
  See \Cref{s:membershipproofs}.
\end{proof}

\begin{remark}
  Note that any $h$ can be written as $h = \sum_{i=0}^{n-1}h_i$, where $h_i$ corresponds to the terms with main variable $x_i$.
  Even when the elimination tree is a path, it is usually more efficient to decompose it in this manner and use \Cref{alg:member} with the above modifications.
\end{remark}

\subsubsection*{Complexity}

We finally proceed to prove the complexity bound from \Cref{thm:memberlinear}.
We restrict ourselves to polynomials~$h$ that preserve the sparsity structure given by the chordal graph~$G$.
More precisely, we assume that the variables of each of the terms of~$h$ correspond to a clique of~$G$, or equivalently, that $h = \sum_l h_l$ for some $h_l\in \K[X_l]$.
Naturally, we will use \Cref{alg:member} with the two modifications from above.
The key idea to notice is that \Cref{alg:member} preserves chordality, as stated next.

\begin{lemma}\label{thm:Gsparse}
  Assume that in \Cref{alg:member} the initial values of $H(f_l)$ are such that $H(f_l)\subset \K[X_l]$ (where $\rank(f_l)=l$).
  Then the same condition is satisfied throughout the algorithm.
\end{lemma}
\begin{proof}
  The update rule used in \Cref{alg:member} is of the form $H(f_p):= H(f_p) + r_{lp} \phi_l(\tilde{h}_{l})$ for some $\tilde{h}_l\in \K[X_l]$, where $\phi_l$ denotes the functional that plugs in $\hat{x}_l$.
  Using \Cref{thm:cliquecontainment}, we have $\phi_l(\tilde{h}_l)\subset \K[X_l\setminus\{x_l\}]\subset \K[X_p]$.
  The result follows.
\end{proof}

\begin{proof}[Proof of \Cref{thm:memberlinear}]
  We consider \Cref{alg:member} with the modifications~\ref{line:initialize2} and~\ref{line:normalize} from above.
  Note that the $q$-dominated condition allows us to bound the degrees of all polynomials computed in \Cref{alg:member}.
  Furthermore, since chordality is preserved (\Cref{thm:Gsparse}), then all polynomials will have at most $q^\kappa$ terms.
  The complexity of the algorithm is determined by the cost of polynomial divisions and polynomial additions.
  Polynomial addition takes linear time in the number of terms, and it is performed once for each arc of the network.
  Thus, their total cost is $O(nW^2q^\kappa)$.
  As for polynomial division, $h\bmod f$ can be obtained in $O(|h|\, |f|\log{|f|})$, where $|\cdot|$ denotes the number of terms~\cite{Monagan2011}.
  Their total cost is $\widetilde{O}(nW q^{2\kappa})$, since there is one operation per node of the network.
\end{proof}

\section{Monomial ideals}\label{s:monomial}

We already showed how to compute chordal network representations of some zero-dimensional ideals.
Before proceeding to the general case, we will consider the special class of monomial ideals.
Recall that an ideal is monomial if it is generated by monomials.
Monomial ideals might have positive-dimension, but their special structure makes their analysis particularly simple.
As in \Cref{exmp:triangnested}, we will see that any monomial ideal admits a compact chordal network representation.
We will also show how such chordal network can be effectively used to compute its dimension, its equidimensional components, and its irreducible components.
These methods will be later generalized to arbitrary polynomial ideals.

\subsection{Chordal triangularization}\label{s:chordmon}

\Cref{alg:chordtriang} will be exactly the same for monomial ideals as in the zero-dimensional case.
The only difference is that for the triangulation operations we need to specify the type of decomposition used, as explained now.

We will say that a set of monomials $T$ is \emph{triangular} if it consists of variables, i.e., $T=\{x_{i_1},\ldots,x_{i_m}\}$.
It is well known that a monomial ideal is prime if and only if it is generated by variables.
It is also known that the minimal primes of a monomial ideal are also monomial.
It follows that any monomial ideal $I$ \emph{decomposes} as
  $\V(I) = \bigcup_T \V(T),$
where the union is over some triangular monomial sets $T$.

By using the above decomposition in each triangulation operation, chordal triangularization can now be applied to monomial ideals, as established in the proposition below.
We point out that even though this decomposition seems quite different from the one of \Cref{s:triangularsets}, both are special instances of more general theory that will be discussed in \Cref{s:regularchains}.

\begin{proposition}
  Let~$F$ be a set of monomials supported on a chordal graph~$G$.
  \Cref{alg:chordtriang} computes a $G$-chordal network, whose chains give a triangular decomposition of~$F$.
\end{proposition}
\begin{proof}
  Proving that the variety is preserved in the algorithm is essentially the same as for the chordally zero-dimensional case (\Cref{thm:preservevariety}).
  It is straightforward to see that the chains of the output are triangular (i.e., they consist of variables).
\end{proof}

\begin{example}\label{exmp:triangnested2}
  Consider the ideal $I = \ideal{x_0x_1, x_0x_2, x_0x_3, x_1x_2, x_1x_4, x_2x_5, x_3x_4, x_3x_5, x_4x_5}$.
  The result of chordal triangularization is shown to the left of \Cref{fig:triangnested2}.
  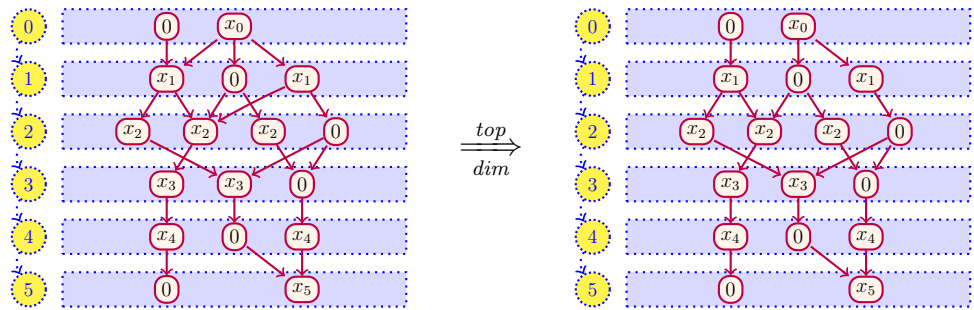
\begin{figure}[htb]
    \centering
      \def\y{.7}
      \tikzstyle{myboxA} = [scale=0.65,draw=purple,fill=orange!8,thick,rectangle,rounded corners]
      \tikzstyle{myarrowA} = [purple,thick,->]
      \tikzstyle{myboxA0} = [scale=0.62,fill=yellow!80,text=blue,draw=blue,thick,densely dotted,circle]
      \tikzstyle{myarrowA0} = [blue,thick,dotted,->,bend right]
      \tikzstyle{mygroup0} = [draw=blue,fill=blue!15, thick, dotted, minimum height=13]
      \def\x{.9}
      \tikzstyle{mygroup} = [mygroup0, minimum width = 130]

    \begin{tikzpicture}[scale=1]
      \node [myboxA] (zero0) at (-1*\x,0*\y) {$0$};
      \node [myboxA] (zero1) at (0*\x,0*\y) {$x_0$};
      \node [myboxA] (one0) at (-1*\x,-1*\y) {$x_1$};
      \node [myboxA] (one1) at (0*\x,-1*\y) {$0$};
      \node [myboxA] (one2) at (1*\x,-1*\y) {$x_1$};
      \node [myboxA] (two0) at (-1.5*\x,-2*\y) {$x_2$};
      \node [myboxA] (two1) at (-.5*\x,-2*\y) {$x_2$};
      \node [myboxA] (two2) at (.5*\x,-2*\y) {$x_2$};
      \node [myboxA] (two3) at (1.5*\x,-2*\y) {$0$};
      \node [myboxA] (three0) at (0*\x,-3*\y) {$x_3$};
      \node [myboxA] (three1) at (-1*\x,-3*\y) {$x_3$};
      \node [myboxA] (three2) at (1*\x,-3*\y) {$0$};
      \node [myboxA] (four0) at (0*\x,-4*\y) {$0$};
      \node [myboxA] (four1) at (-1*\x,-4*\y) {$x_4$};
      \node [myboxA] (four2) at (1*\x,-4*\y) {$x_4$};
      \node [myboxA] (five0) at (1*\x,-5*\y) {$x_5$};
      \node [myboxA] (five1) at (-1*\x,-5*\y) {$0$};
      \draw [myarrowA] (zero0) to (one0);
      \draw [myarrowA] (zero1) to (one0);
      \draw [myarrowA] (zero1) to (one1);
      \draw [myarrowA] (zero1) to (one2);
      \draw [myarrowA] (one0) to (two0);
      \draw [myarrowA] (one0) to (two1);
      \draw [myarrowA] (one1) to (two1);
      \draw [myarrowA] (one1) to (two2);
      \draw [myarrowA,out=-155,in=30] (one2) to (two1);
      \draw [myarrowA] (one2) to (two3);
      \draw [myarrowA] (two0) to (three0);
      \draw [myarrowA] (two1) to (three1);
      \draw [myarrowA] (two2) to (three2);
      \draw [myarrowA] (two3) to (three2);
      \draw [myarrowA] (two3) to (three0);
      \draw [myarrowA] (three0) to (four0);
      \draw [myarrowA] (three1) to (four1);
      \draw [myarrowA] (three2) to (four2);
      \draw [myarrowA] (four0) to (five0);
      \draw [myarrowA] (four1) to (five1);
      \draw [myarrowA] (four2) to (five0);
      \begin{scope}[on background layer]
      \node [mygroup] (zero) at (0*\x,0*\y){}; 
      \node [mygroup] (one) at (0*\x,-1*\y){}; 
      \node [mygroup] (two) at (0*\x,-2*\y){}; 
      \node [mygroup] (three) at (0*\x,-3*\y){}; 
      \node [mygroup] (four) at (0*\x,-4*\y){}; 
      \node [mygroup] (five) at (0*\x,-5*\y){}; 
      \end{scope}
      \node [myboxA0,left=.2 of zero] (zero00) {0};
      \node [myboxA0,left=.2 of one] (one00) {1};
      \node [myboxA0,left=.2 of two] (two00) {2};
      \node [myboxA0,left=.2 of three] (three00) {3};
      \node [myboxA0,left=.2 of four] (four00) {4};
      \node [myboxA0,left=.2 of five] (five00) {5};
      \draw [myarrowA0] (zero00) to (one00);
      \draw [myarrowA0] (one00) to (two00);
      \draw [myarrowA0] (two00) to (three00);
      \draw [myarrowA0] (three00) to (four00);
      \draw [myarrowA0] (four00) to (five00);
    \end{tikzpicture}
    $\quad$
    \raisebox{4.3\height}{$\xRightarrow[\mathit{\;dim\;}]{\mathit{\;top\;}}$}
    $\quad$
    \begin{tikzpicture}[scale=1]
      \node [myboxA] (zero0) at (-1*\x,0*\y) {$0$};
      \node [myboxA] (zero1) at (0*\x,0*\y) {$x_0$};
      \node [myboxA] (one0) at (-1*\x,-1*\y) {$x_1$};
      \node [myboxA] (one1) at (0*\x,-1*\y) {$0$};
      \node [myboxA] (one2) at (1*\x,-1*\y) {$x_1$};
      \node [myboxA] (two0) at (-1.5*\x,-2*\y) {$x_2$};
      \node [myboxA] (two1) at (-.5*\x,-2*\y) {$x_2$};
      \node [myboxA] (two2) at (.5*\x,-2*\y) {$x_2$};
      \node [myboxA] (two3) at (1.5*\x,-2*\y) {$0$};
      \node [myboxA] (three0) at (0*\x,-3*\y) {$x_3$};
      \node [myboxA] (three1) at (-1*\x,-3*\y) {$x_3$};
      \node [myboxA] (three2) at (1*\x,-3*\y) {$0$};
      \node [myboxA] (four0) at (0*\x,-4*\y) {$0$};
      \node [myboxA] (four1) at (-1*\x,-4*\y) {$x_4$};
      \node [myboxA] (four2) at (1*\x,-4*\y) {$x_4$};
      \node [myboxA] (five0) at (1*\x,-5*\y) {$x_5$};
      \node [myboxA] (five1) at (-1*\x,-5*\y) {$0$};
      \draw [myarrowA] (zero0) to (one0);
      \draw [myarrowA] (zero1) to (one1);
      \draw [myarrowA] (zero1) to (one2);
      \draw [myarrowA] (one0) to (two0);
      \draw [myarrowA] (one0) to (two1);
      \draw [myarrowA] (one1) to (two1);
      \draw [myarrowA] (one1) to (two2);
      %\draw [myarrowA,out=-155,in=30] (one2) to (two1);
      \draw [myarrowA] (one2) to (two3);
      \draw [myarrowA] (two0) to (three0);
      \draw [myarrowA] (two1) to (three1);
      \draw [myarrowA] (two2) to (three2);
      \draw [myarrowA] (two3) to (three2);
      \draw [myarrowA] (two3) to (three0);
      \draw [myarrowA] (three0) to (four0);
      \draw [myarrowA] (three1) to (four1);
      \draw [myarrowA] (three2) to (four2);
      \draw [myarrowA] (four0) to (five0);
      \draw [myarrowA] (four1) to (five1);
      \draw [myarrowA] (four2) to (five0);
      \begin{scope}[on background layer]
      \node [mygroup] (zero) at (0*\x,0*\y){}; 
      \node [mygroup] (one) at (0*\x,-1*\y){}; 
      \node [mygroup] (two) at (0*\x,-2*\y){}; 
      \node [mygroup] (three) at (0*\x,-3*\y){}; 
      \node [mygroup] (four) at (0*\x,-4*\y){}; 
      \node [mygroup] (five) at (0*\x,-5*\y){}; 
      \end{scope}
      \node [myboxA0,left=.2 of zero] (zero00) {0};
      \node [myboxA0,left=.2 of one] (one00) {1};
      \node [myboxA0,left=.2 of two] (two00) {2};
      \node [myboxA0,left=.2 of three] (three00) {3};
      \node [myboxA0,left=.2 of four] (four00) {4};
      \node [myboxA0,left=.2 of five] (five00) {5};
      \draw [myarrowA0] (zero00) to (one00);
      \draw [myarrowA0] (one00) to (two00);
      \draw [myarrowA0] (two00) to (three00);
      \draw [myarrowA0] (three00) to (four00);
      \draw [myarrowA0] (four00) to (five00);
    \end{tikzpicture}
    \caption{Chordal network from \Cref{exmp:triangnested2}, and its top-dimensional part.}
    \label{fig:triangnested2}
  \end{figure}
\end{example}

As in the chordally zero-dimensional case, we can also prove that the complexity is linear in~$n$ when the treewidth is bounded.

\begin{theorem}\label{thm:monlinear}
  Let~$F$ be a set of monomials supported on a chordal graph~$G$ of clique number~$\kappa$.
  Then $F$ can be represented by a triangular chordal network with at most~$n\,2^\kappa$ nodes, which can be computed in time $O(n\,2^{O(\kappa)})$.
\end{theorem}
\begin{proof}
  Note that after the $l$-th triangulation round we will have at most $2^\kappa$ rank~$l$ nodes, since the triangular monomial sets in $\K[X_l]$ are in bijection with the subsets of $X_l$.
  A similar argument proves that the width of the network is bounded by $2^\kappa$ after an elimination round, and thus throughout the algorithm.
  The cost of computing a triangular decomposition in $\K[X_l]$ is polynomial in $2^{|X_l|}$, since we can simply enumerate over all possible triangular monomial sets.
  Thus, the cost of all triangulation operations is $O(nW\,2^{O(\kappa)})=O(n\,2^{O(\kappa)})$.
  The cost of the elimination and merging operations is negligible.
\end{proof}

\subsection{Computing with chordal networks}\label{s:usingtriangularmon}

Let $\mathcal{N}$ be a chordal network representation of a monomial ideal~$I$.
We will show how to effectively use $\mathcal{N}$ to solve the following problems:
\begin{description}
  \item[Dimension] Determine the dimension of $I$.
  \item[Top-dimensional part] Describe the top-dimensional part of $\V(I)$.
  \item[Irreducible components] Determine the minimal primes of $I$.
\end{description}
The above problems can be shown to be hard in general by using the correspondence between minimal vertex covers of a graph and the irreducible components of its edge ideal (see \Cref{exmp:triangnested}).
We will see that, given the chordal network, the first two problems can be solved in linear time with a dynamic program.
The third one is much more complicated, since we need to enumerate over all chains of the network to verify if they are minimal.
In order to do this efficiently, we will need to address the following problems.
\begin{description}
  \item[Dimension count] Classify the number of chains $C$ of $\mathcal{N}$ according to its dimension.
  \item[Isolate dimension~$\boldsymbol{d}$] Enumerate all chains $C$ of $\mathcal{N}$ such that $\dim(\V(C))=d$.
\end{description}

We proceed to solve each of the problems from above.
To simplify the exposition, we will assume for this section that the elimination tree is a path, but it is not difficult to see that all these methods will work for arbitrary chordal networks.

\subsubsection*{Dimension}

Let us see that it is quite easy to compute the dimension of $\V(\mathcal{N})$.
Since the variety $\V(T)$ of a triangular monomial set is a linear space, its dimension is $\dim(\V(T))=n -|T|$.
Therefore, $\dim(\V(\mathcal{N}))= n-\min_{C}|C|$, where the minimum is taken over all chains of the network.
Note that we ignore the zero entries of $C$.
In particular, for the network in \Cref{fig:triangnested2} we have $\dim(\V(\mathcal{N}))=6-4=2$.

We reduced the problem to computing the smallest cardinality of a chain of $\mathcal{N}$.
This can be done using a simple dynamic program, which is quite similar to the one in \Cref{alg:degree}.
For each node $f_l$ we save the value $\ell(f_l)$ corresponding to the length of the shortest chain up to level $l$.
For an arc $(f_l,f_p)$ with $f_p\neq 0$, the update rule is simply 
$\ell(f_p):= \min(\ell(f_p),1+\ell(f_l))$.
It follows that we can compute in linear time the dimension of~$\V(\mathcal{N})$.

\subsubsection*{Top-dimensional part}
We can get a chordal network $\mathcal{N}_{\mathit{top}}$ describing its top-dimensional part by modifying the procedure that computes the dimension.
Indeed, assume that for some arc $(f_l,f_p)$ we have $\ell(f_p)<1+\ell(f_l)$ and thus the update $\ell(f_p):= \min(\ell(f_p),1+\ell(f_l))$ is not needed.
This means that the arc $(f_l,f_p)$ is unnecessary for the top-dimensional component.
By pruning the arcs of $\mathcal{N}$ in such manner we obtain the wanted network $\mathcal{N}_{\mathit{top}}$.

\begin{example}
  Let $\mathcal{N}$ be the network on the left of \Cref{fig:triangnested2}.
  Note that $\mathcal{N}$ has $9$ chains, two of them are $C_1=(x_1,x_2,x_3,x_5)$, $C_2=(x_0,x_1,x_2,x_3,x_5)$, of dimensions~$2$ and~$1$.
  By pruning some arcs, we obtain its highest dimensional part $\mathcal{N}_{\mathit{top}}$, shown to the right of \Cref{fig:triangnested2}.
  This network $\mathcal{N}_{\mathit{top}}$ only has $6$ chains; note that $C_2$ is not one of them.
  In this case neither of the chains removed was minimal (e.g., $C_2\supsetneq C_1$), so that $\V(\mathcal{N})=\V(\mathcal{N}_{\mathit{top}})$.
  Thus, both $\mathcal{N}$ and $\mathcal{N}_{\mathit{top}}$ are valid chordal network representations of the ideal from \Cref{exmp:triangnested2}, although the latter is preferred since all its chains are minimal.
  Similarly, the network from \Cref{fig:triangnested} was obtained by using chordal triangularization and then computing its highest dimensional part.
\end{example}

\subsubsection*{Irreducible components}
Chordal triangularization can also aid in computing the minimal primes of an ideal (geometrically, the irreducible components).
In the monomial case, any chain of~$\mathcal{N}$ defines a prime ideal, and thus we only need to determine which chains are minimal with respect to containment.
In some cases it is enough to prune certain arcs of the network (e.g., \Cref{fig:triangnested2}), but this is not always possible.

Unfortunately, we do not know a better procedure than simply iterating over all chains of the network checking for minimality.
Nonetheless, we can make this method much more effective by proceeding in order of decreasing dimension.
This simple procedure is particularly efficient when we are only interested in the minimal primes of high dimension, as will be seen in \Cref{s:birthdeath}.
In the remaining of the section we will explain how to enumerate the chains by decreasing dimension (this is precisely the dimension isolation problem).

\subsubsection*{Dimension count}
Classifying the number of chains according to its dimension can be done with a very similar dynamic program as for computing the dimension.
As discussed above, the dimension of a chain is simply given by its cardinality.
For a rank~$l$ node $f_l$ of the network and for any $0\leq k\leq l+1$, let $c_k(f_l)$ denote the number of chains of the network (up to level $l$) with cardinality exactly $k$.
Then for an arc $(f_l,f_p)$ with $f_p\neq 0$ the update rule is simply $c_k(f_p):= c_k(f_p)+c_{k-1}(f_l)$.

\subsubsection*{Dimension isolation}
For simplicity of exposition we only describe how to produce one chain~$C$ of dimension~$d$, but it is straightforward to then generate all of them.
As in \Cref{exmp:sampling}, we follow a bottom up strategy, successively adding nodes to the chain.
We first need to choose a rank~$n-1$ node $f_{n-1}$ that belongs to at least one chain of dimension~$d$.
Using the values $c_k(f_l)$ from above, we can choose any $f_{n-1}$ for which $c_{n-d}(f_{n-1})\geq 1$.
Assuming that we chose some $f_{n-1}\neq 0$, we now need to find an adjacent rank~$n-2$ node $f_{n-2}$ such that $c_{n-d-1}(f_{n-2})\geq 1$.
It is clear how to continue.

\section{The general case}\label{s:positivedim}

We finally proceed to compute chordal network representations of arbitrary polynomial ideals.
We will also see how the different chordal network algorithms developed earlier (e.g., radical ideal membership, isolating the top-dimensional component) have a natural extension to this general setting.

\subsection{Regular chains}\label{s:regularchains}

The theory of triangular sets for positive-dimensional varieties is more involved; we refer to~\cite{Hubert:2003aa, Wang2001} for an introduction.
We now present the concept of regular chains, which is at the center of this theory.

A set of polynomials $T\subset \K[X]\setminus \K$ is a \emph{triangular set} if its elements have distinct main variables.
Let $h$ be the product of the initials (\Cref{defn:triangularzero}) of the polynomials in $T$.
The geometric object associated to $T$ is the \emph{quasi-component}
\begin{align*}
  \W(T):=\V(T)\setminus \V(h)\subset \overline{\K}^n.
\end{align*}
The attached algebraic object is the \emph{saturated ideal}
\begin{align*}
  \sat(T):= \ideal{ T }:  h ^\infty 
  = \{f\in \K[X]: h^Nf\in \ideal{T} \mbox{ for some } N\in \N\}.
\end{align*}
Note that $\V(\sat(T))=\overline{\W(T)}$, where the closure is in the Zariski topology.

Polynomial pseudo-division is a basic operation in triangular sets.
Let $f,g$ be polynomials of degrees $d,e$ in $x:=\mvar(g)$.
The basic idea is to see $f,g$ as univariate polynomials in~$x$ (with coefficients in $\K[X\setminus\{x\}]$), and in order that we can always divide $f$ by $g$, we first multiply by some power of $\init(g)$.
Formally, the \emph{pseudo-remainder} of $f$ by $g$ is $\prem(f,g):=f$ if $d<e$, and otherwise $\prem(f,g):=\init(g)^{d-e+1}f \bmod(g)$.
Pseudo-division can be extended to triangular sets in the natural way. 
The \emph{pseudo-remainder} of $f$ by $T=\{t_1,\ldots,t_k\}$, where $\mvar(t_1)>\cdots>\mvar(t_k)$, is
\begin{align*}
  \prem(f,T) = \prem(\cdots(\prem(f,t_1)\cdots,t_k).
\end{align*}

\begin{definition}
  A \emph{regular chain} is a triangular set $T$ such that for any polynomial $f$ 
  \begin{align*}
    f\in \sat(T) &\iff \prem(f,T)=0.
  \end{align*}
\end{definition}
\begin{remark}
  Note that a zero-dimensional triangular set (\Cref{defn:triangularzero}) is a regular chain, since pseudo-reduction coincides with Gr\"obner bases reduction.
\end{remark}

Regular chains have very nice algorithmic properties.
In particular, they are always consistent (i.e., $\W(T)\neq 0$), and furthermore $\dim(\W(T))=n-|T|$.
\Cref{tab:comparison} summarizes some of these properties, comparing them with Gr\"obner bases.

\begin{table}[htbp]
  \centering
  \caption{Gr\"obner bases vs.\ regular chains}
    \begin{tabular}{ccc}
    \toprule
    & Gr\"obner basis ($\mathcal{G}$) & Regular chain ($T$) \\
    \midrule
    Geometric object & $\V(\mathcal{G})$ & $\W(T)$ \\
    Algebraic object & $\ideal{ \mathcal{G} }$ & $\sat(T)$ \\
    Feasible & if $1\notin \mathcal{G}$ & always \\
    Ideal membership & $\mathrm{Remainder} = 0$ & $\mathrm{PseudoRemainder} =0$ \\
    Dimension & from Hilbert series & $n - |T|$ \\
    Elimination ideal & $\mathcal{G}_{\mathrm{lex}}\cap\K[x_l,\ldots,x_{n-1}]$ & $T\cap\K[x_l,\ldots,x_{n-1}]$ \\
    \bottomrule
    \end{tabular}
  \label{tab:comparison}
\end{table}

\begin{definition}
  A \emph{triangular decomposition} of a polynomial set $F$ is a collection $\T$ of regular chains, such that
  $\V(F) = \bigcup_{T\in \T} \W(T).$
\end{definition}
\begin{remark}
  There is a weaker notion of decomposition that is commonly used: $\mathcal{T}$ is a \emph{Kalkbrener} triangular decomposition if
  $\V(F) = \bigcup_{T\in \T} \overline{\W(T)}.$
\end{remark}

\begin{example}\label{exmp:triangadjminors4}
  Let $F=\{x_0x_3 - x_1x_2, x_2x_5 - x_3x_4, x_4x_7 - x_5x_6\}$ consist of the adjacent minors of a $2\times 4$ matrix.
  It can be decomposed into $8$ regular chains:
  {\small
  \begin{gather*}
    ( x_0x_3 - x_1x_2, x_2x_5 - x_3x_4, x_4x_7 - x_5x_6 ),\,
    ( x_0x_3 - x_1x_2, x_4, x_5 ),\,
    ( x_2, x_3, x_4x_7 - x_5x_6 ),\\
    ( x_1, x_3, x_4, x_5 ),\,
    ( x_1, x_3, x_5, x_7 ),\,
    ( x_2, x_3, x_5, x_7 ),\,
    ( x_2, x_3, x_6, x_7 ),\,
    ( x_0x_3 - x_1x_2, x_2x_5 - x_3x_4, x_6, x_7 ).
  \end{gather*}
  }
  Note that the first three triangular sets (first line) have dimension $8-3=5$.
  Observe that the quasi-components $\W(T)$ of these three sets do not cover the points for which $x_3=x_7=0$, which is why we need the remaining five sets.
  However, these three triangular sets alone give a Kalkbrener decomposition of the variety.
\end{example}

\subsection{Regular systems}

In the study of triangular sets, it is useful to consider systems of polynomials containing both equations $\{f_i(x)=0\}_i$ and \emph{inequations} $\{h_j(x)\neq 0\}_j$.
Following the notation of~\cite{Wang2001}, we say that a \emph{polynomial system} $\mathfrak{F}=(F,H)$ is a pair of polynomial sets $F,H\subset \K[X]$, and its associated geometric object is the quasi-variety 
\begin{align*}
  \Z(\mathfrak{F}) := \{x\in \overline{\K}^n: f(x)=0 \mbox{ for } f\in F,\; h(x)\neq 0 \mbox{ for }h\in H\}.
\end{align*}
For instance, the quasi-component $\W(T)$ of a triangular set is the quasi-variety of the polynomial system $(T,\init(T))$, where $\init(T)$ is the set of initials of $T$.

For a polynomial system $\mathfrak{F}=(F,H)$ we denote by
$\elim{p}{\mathfrak{F}}$ the polynomial system $(\elim{p}{F},\elim{p}{H})$.
We also denote by $\mathfrak{F}_1+\mathfrak{F}_2$ the concatenation of two polynomial systems, i.e, $(F_1,H_1)+(F_2,H_2):=(F_1\cup F_2,H_1\cup H_2)$.

\begin{definition}\label{defn:regsys}
  A \emph{regular system} is a pair $\mathfrak{T}=(T,U)$ such that $T$ is triangular and for any $0\leq k<n$:
  \begin{enumerate}[leftmargin=.5in,label=(\roman*)]
    \item\label{line:regsys1} either $T^{\langle k\rangle}= \emptyset$ or $U^{\langle k\rangle}=\emptyset$, where the superscript~${}^{\langle k\rangle}$ denotes the polynomials with main variable~$x_k$.
    \item\label{line:regsys2} $\init(f)(\hat{x}^{k+1})\neq 0$ for any $f\in T^{\langle k\rangle}\cup U^{\langle k\rangle}$ and $\hat{x}^{k+1}\in \Z(\elim{k+1}{\mathfrak{T}})\subset\overline{\K}^{n-k-1}$.
  \end{enumerate}
  A regular system is \emph{squarefree} if the polynomials $f(x_k,\hat{x}^{k+1})\subset \K[x_k]$ are squarefree for any $f\in T^{\langle k\rangle}\cup U^{\langle k\rangle}$ and any $\hat{x}^{k+1}\in \Z(\elim{k+1}{\mathfrak{T}})\subset\overline{\K}^{n-k-1}$.
\end{definition}

For a regular system $(T,U)$ the set $T$ is a regular chain, and conversely, for a regular chain $T$ there is some $U$ such that $(T,U)$ is a regular system~\cite{Wang2000}.
Wang showed how to decompose any polynomial system in characteristic zero into (squarefree) regular systems~\cite{Wang2000,Wang2001}.

\begin{definition}
  A \emph{triangular decomposition} of a polynomial system $\mathfrak{F}$ is a collection $\T$ of regular systems, such that
    $\Z(\mathfrak{F}) = \bigcup_{\mathfrak{T}\in \T} \Z(\mathfrak{T}).$
\end{definition}

\begin{remark}[Binomial ideals]\label{remk:binomials}
  Consider a polynomial system $\mathfrak{F}=(F,U)$ such that $F$ consists of binomials (two terms) and $H=\{x_{i_1},\ldots,x_{i_m}\}$ consists of variables.
  We can decompose $\mathfrak{F}$ into regular systems $\mathfrak{T}=(T,U)$ that preserve the binomial structure.
  Assume first that $H=\{x_0,\ldots,x_{n-1}\}$ contains all variables.
  Equivalently, we are looking for the zero set of $F$ on the torus $(\overline{\K}\setminus\{0\})^n$.
  It is well known that~$F$ can be converted to (binomial) triangular form $T$ by computing the Hermite normal form of the matrix of exponents~\cite[\S3.2]{sturmsolving}, and the inequations~$U$ correspond to the non-pivot variables.
  For an arbitrary~$H$, we can enumerate over the choices of nonzero variables.
  %Therefore, there is always a decomposition into at most $2^n$ regular systems.
  %\begin{example}
  %  Let us return to the equations~$F$ from \Cref{exmp:triangadjminors4}.
  %  The system $(F,H)$, where $H=\{x_0,\ldots,x_7\}$, is equivalent to the regular system
  %  \begin{align*}
  %    \mathfrak{T} = (\{ x_0x_3 - x_1x_2, x_2x_5 - x_3x_4, x_4x_7 - x_5x_6 \}, \{x_1,x_3,x_5,x_6,x_7\}).
  %  \end{align*}
  %\end{example}
\end{remark}

\subsection{Chordal triangularization}

\Cref{alg:chordtriang} extends to the positive-dimensional case in the natural way, although with one important difference:
the nodes of the chordal network will be polynomial systems, i.e., pairs of polynomial sets.
We now describe the modifications of the main steps of the algorithm:
\begin{description}
  \item[Initialization] The nodes of the induced $G$-chordal network are now of the form $\mathfrak{F}_l = (F_l,H_l)$, where $F_l=F\cap\K[X_l]$ and $H_l = \emptyset$.
  \item[Triangulation] 
    For a node $\mathfrak{F}_l$ we decompose it into regular systems $\mathfrak{T}$ and we replace $\mathfrak{F}_l$ with a node for each of them.
  \item[Elimination] Let $\mathfrak{F}_l$ be the rank~$l$ node we will eliminate, and let $\mathfrak{F}_p$ be an adjacent rank~$p$ node.
    Then we create a rank~$p$ node $\mathfrak{F}_p':= \mathfrak{F}_p+ \elim{p}{\mathfrak{F}_l}$.
  \item[Termination] After all triangulation/elimination operations, we may remove the inequations from the nodes of the network, i.e., replace $\mathfrak{F}=(F,H)$ with $F$.
\end{description}

\begin{example}\label{exmp:adjminors4}
  \begin{figure}[htb]
    \centering

    \def\y{0.9}
    \tikzstyle{myboxA} = [scale=0.65,draw=purple,fill=orange!8,thick,rectangle,rounded corners]
    \tikzstyle{myarrowA} = [purple,thick,->]
    \tikzstyle{myboxA0} = [fill=yellow!80,text=blue,draw=blue,thick,densely dotted,circle]
    \tikzstyle{myarrowA0} = [blue,thick,dotted,->,bend right]
    \tikzstyle{mygroup0} = [draw=blue,fill=blue!15, thick, dotted, minimum height=16]
    \begin{tikzpicture}[scale=1]
      \def\x{1.0}
      \tikzstyle{mygroup} = [mygroup0, minimum width = 65]
      \node [myboxA] (zero0) at (0*\x,0*\y) {$x_0x_3-x_1x_2$};
      \node [myboxA](two0) at (0*\x,-1*\y) {$x_2x_5-x_3x_4$};
      \node [myboxA](four0) at (0*\x,-2*\y) {$x_4x_7-x_5x_6$};
      \draw [myarrowA] (zero0) to (two0);
      \draw [myarrowA] (two0) to (four0);
      \begin{scope}[on background layer]
      \node [mygroup] (zero) at (0*\x,0*\y){}; 
      \node [mygroup] (two) at (0*\x,-1*\y){}; 
      \node [mygroup] (four) at (0*\x,-2*\y){}; 
      \end{scope}
    \end{tikzpicture}
    \raisebox{4.6\height}{$\xRightarrow[\mathit{elim}]{\mathit{tria}}$}
    \begin{tikzpicture}[scale=1]
      \def\x{1.6}
      \tikzstyle{mygroup} = [mygroup0, minimum width = 130]
      \node [myboxA] (zero0) at (.04*\x,0*\y) {$x_0x_3-x_1x_2$};
      \node [myboxA] (zero1) at (-1*\x,0*\y) {$x_1$};
      \node [myboxA] (zero2) at (1*\x,0*\y) {$0$};
      \node [myboxA](two0) at (.04*\x,-1*\y) {$x_2x_5-x_3x_4/x_3$};
      \node [myboxA](two1) at (-1*\x,-1*\y) {$x_2x_5,x_3$};
      \node [myboxA](two2) at (1*\x,-1*\y) {$x_2,x_3$};
      \node [myboxA](four0) at (.04*\x,-2*\y) {$x_4x_7-x_5x_6$};
      \draw [myarrowA] (zero0) to (two0);
      \draw [myarrowA] (zero1) to (two1);
      \draw [myarrowA] (zero2) to (two2);
      \draw [myarrowA] (two0) to (four0);
      \draw [myarrowA] (two1) to (four0);
      \draw [myarrowA] (two2) to (four0);
      \begin{scope}[on background layer]
      \node [mygroup] (zero) at (0*\x,0*\y){}; 
      \node [mygroup] (two) at (0*\x,-1*\y){}; 
      \node [mygroup] (four) at (0*\x,-2*\y){}; 
      \end{scope}
    \end{tikzpicture}
    \raisebox{2.5\height}{$\xRightarrow[\mathit{elim}]{\mathit{tria}}$}
    \begin{tikzpicture}[scale=1]
      \def\x{2.4}
      \tikzstyle{mygroup} = [mygroup0, minimum width = 172]
      \node [myboxA] (zero0) at (.13*\x,0*\y) {$x_0x_3-x_1x_2$};
      \node [myboxA] (zero1) at (-.55*\x,0*\y) {$x_1$};
      \node [myboxA] (zero2) at (.9*\x,0*\y) {$0$};
      \node [myboxA](two0) at (.13*\x,-1*\y) {$x_2x_5-x_3x_4/x_3$};
      \node [myboxA](two1) at (-1*\x,-1*\y) {$0/x_3$};
      \node [myboxA](two2) at (.9*\x,-1*\y) {$x_2,x_3$};
      \node [myboxA](two3) at (-.55*\x,-1*\y) {$x_3$};
      \node [myboxA](four0) at (.13*\x,-2*\y) {$x_4x_7-x_5x_6/x_5$};
      \node [myboxA](four1) at (-1*\x,-2*\y) {$x_4,x_5$};
      \node [myboxA](four2) at (.9*\x,-2*\y) {$x_4x_7-x_5x_6$};
      \node [myboxA](four3) at (-.55*\x,-2*\y) {$x_4x_7,x_5$};
      \draw [myarrowA] (zero0) to (two0);
      \draw [myarrowA] (zero0) to (two1);
      \draw [myarrowA] (zero1.south east) to (two2.north west);
      \draw [myarrowA] (zero1) to (two3);
      \draw [myarrowA] (zero2) to (two2);
      \draw [myarrowA] (two0) to (four0);
      \draw [myarrowA] (two1) to (four1);
      \draw [myarrowA] (two2) to (four2);
      \draw [myarrowA] (two3) to (four3);
      \begin{scope}[on background layer]
      \node [mygroup] (zero) at (0*\x,0*\y){}; 
      \node [mygroup] (two) at (0*\x,-1*\y){}; 
      \node [mygroup] (four) at (0*\x,-2*\y){}; 
      \end{scope}
    \end{tikzpicture}
    \\\vspace{15pt}
    \raisebox{.3\height}{$\xRightarrow{\mathit{tria}}$}
    \begin{tikzpicture}[scale=1]
      \def\x{3.65}
      \tikzstyle{mygroup} = [mygroup0, minimum width = 235]
      \tikzstyle{myboxA} = [scale=0.63,draw=purple,fill=orange!8,thick,rectangle,rounded corners]
      \node [myboxA] (zero0) at (.03*\x,0*\y) {$x_0x_3-x_1x_2$};
      \node [myboxA] (zero1) at (-.56*\x,0*\y) {$x_1$};
      \node [myboxA] (zero2) at (.7*\x,0*\y) {$0$};
      \node [myboxA](two0) at (.03*\x,-1*\y) {$x_2x_5-x_3x_4/x_3$};
      \node [myboxA](two1) at (-1*\x,-1*\y) {$0/x_3$};
      \node [myboxA](two2) at (.7*\x,-1*\y) {$x_2,x_3$};
      \node [myboxA](two3) at (-.56*\x,-1*\y) {$x_3$};
      \node [myboxA](four0) at (-.56*\x,-2*\y) {$x_4x_7-x_5x_6/x_5x_7$};
      \node [myboxA](four1) at (.71*\x,-2*\y) {$x_6,x_7/x_5$};
      \node [myboxA](four2) at (-1*\x,-2*\y) {$x_4,x_5$};
      \node [myboxA](four3) at (.03*\x,-2*\y) {$x_4x_7-x_5x_6/x_7$};
      \node [myboxA](four4) at (.42*\x,-2*\y) {$x_5,x_7$};
      \node [myboxA](four5) at (1*\x,-2*\y) {$x_6,x_7$};
      \draw [myarrowA] (zero0) to (two0);
      \draw [myarrowA] (zero0) to (two1);
      \draw [myarrowA] (zero1.south east) to (two2.north west);
      \draw [myarrowA] (zero1) to (two3);
      \draw [myarrowA] (zero2) to (two2);
      \draw [myarrowA] (two0.west) to (four0);
      \draw [myarrowA] (two0.east) to (four1);
      \draw [myarrowA] (two1) to (four2);
      \draw [myarrowA] (two2) to (four3.north);
      \draw [myarrowA] (two2) to (four4);
      \draw [myarrowA] (two2) to (four5);
      \draw [myarrowA] (two3) to (four2);
      \draw [myarrowA] (two3.-25) to (four4.north);
      \begin{scope}[on background layer]
      \node [mygroup] (zero) at (0*\x,0*\y){}; 
      \node [mygroup] (two) at (0*\x,-1*\y){}; 
      \node [mygroup] (four) at (0*\x,-2*\y){}; 
      \end{scope}
    \end{tikzpicture}
    %\\\vspace{15pt}
    \raisebox{2.4\height}{$\xRightarrow{\mathit{term}}$}
    %\raisebox{1.9\height}{$\xRightarrow{\mathit{\substack{remove\\ineqs}}}$}
    \begin{tikzpicture}[scale=1]
      \def\x{1.8}
      \tikzstyle{mygroup} = [mygroup0, minimum width = 130]
      \node [myboxA] (zero0) at (0.26*\x,0*\y) {$x_0x_3-x_1x_2$};
      \node [myboxA] (zero1) at (-.47*\x,0*\y) {$x_1$};
      \node [myboxA] (zero2) at (1.0*\x,0*\y) {$0$};
      \node [myboxA](two0) at (0.26*\x,-1*\y) {$x_2x_5-x_3x_4$};
      \node [myboxA](two1) at (-1*\x,-1*\y) {$0$};
      \node [myboxA](two2) at (1*\x,-1*\y) {$x_2,x_3$};
      \node [myboxA](two3) at (-.47*\x,-1*\y) {$x_3$};
      \node [myboxA](four0) at (0.26*\x,-2*\y) {$x_4x_7-x_5x_6$};
      \node [myboxA](four1) at (1*\x,-2*\y) {$x_6,x_7$};
      \node [myboxA](four2) at (-1.0*\x,-2*\y) {$x_4,x_5$};
      \node [myboxA](four4) at (-.47*\x,-2*\y) {$x_5,x_7$};
      \draw [myarrowA] (zero0) to (two0);
      \draw [myarrowA] (zero0) to (two1);
      %\draw [myarrowA] (zero1.south east) to (two2.north west);%sobra
      \draw [myarrowA] (zero1) to (two3);
      \draw [myarrowA] (zero2) to (two2);
      \draw [myarrowA] (two0) to (four0);
      \draw [myarrowA] (two0) to (four1);%sobra
      \draw [myarrowA] (two1) to (four2);
      \draw [myarrowA] (two2.-120) to (four0);
      \draw [myarrowA] (two2) to (four1);
      \draw [myarrowA] (two2.-140) to (four4.35);%sobra
      \draw [myarrowA] (two3) to (four2);
      \draw [myarrowA] (two3) to (four4);
      \begin{scope}[on background layer]
      \node [mygroup] (zero) at (0*\x,0*\y){}; 
      \node [mygroup] (two) at (0*\x,-1*\y){}; 
      \node [mygroup] (four) at (0*\x,-2*\y){}; 
      \end{scope}
      \node [myboxA0,scale=0.60,left=.1 of zero] (zero00) {01};
      \node [myboxA0,scale=0.60,left=.1 of two] (two00) {23};
      \node [myboxA0,scale=0.55,left=.1 of four] (four00) {4-7};
      \draw [myarrowA0] (zero00) to (two00);
      \draw [myarrowA0] (two00) to (four00);
    \end{tikzpicture}
    %\raisebox{1.3\height}{$\xRightarrow{\mathit{\substack{prune\\unneeded\\nodes}}}$}
    %\begin{tikzpicture}[scale=1]
    %  \def\x{1.4}
    %  \tikzstyle{mygroup} = [mygroup0, minimum width = 120]
    %  \node [myboxA] (zero0) at (0*\x,0*\y) {$x_0x_3-x_1x_2$};
    %  \node [myboxA] (zero2) at (1.0*\x,0*\y) {$0$};
    %  \node [myboxA](two0) at (0*\x,-1*\y) {$x_2x_5-x_3x_4$};
    %  \node [myboxA](two1) at (-1*\x,-1*\y) {$0$};
    %  \node [myboxA](two2) at (1*\x,-1*\y) {$x_2,x_3$};
    %  \node [myboxA](four0) at (0*\x,-2*\y) {$x_4x_7-x_5x_6$};
    %  \node [myboxA](four2) at (-1.0*\x,-2*\y) {$x_4,x_5$};
    %  \draw [myarrowA] (zero0) to (two0);
    %  \draw [myarrowA] (zero0) to (two1);
    %  \draw [myarrowA] (zero2) to (two2);
    %  \draw [myarrowA] (two0) to (four0);
    %  \draw [myarrowA] (two1) to (four2);
    %  \draw [myarrowA] (two2) to (four0);
    %  \begin{scope}[on background layer]
    %  \node [mygroup] (zero) at (0*\x,0*\y){}; 
    %  \node [mygroup] (two) at (0*\x,-1*\y){}; 
    %  \node [mygroup] (four) at (0*\x,-2*\y){}; 
    %  \end{scope}
    %  \node [myboxA0,left=.2 of zero] (zero0) {$0/1$};
    %  \node [myboxA0,left=.2 of two] (two0) {$2/3$};
    %  \node [myboxA0,left=.2 of four] (four0) {$4/5$};
    %  \draw [myarrowA0] (zero0) to (two0);
    %  \draw [myarrowA0] (two0) to (four0);
    %\end{tikzpicture}
    \caption{Chordal triangularization from \Cref{exmp:adjminors4}.}
    \label{fig:adjminors4}
  \end{figure}
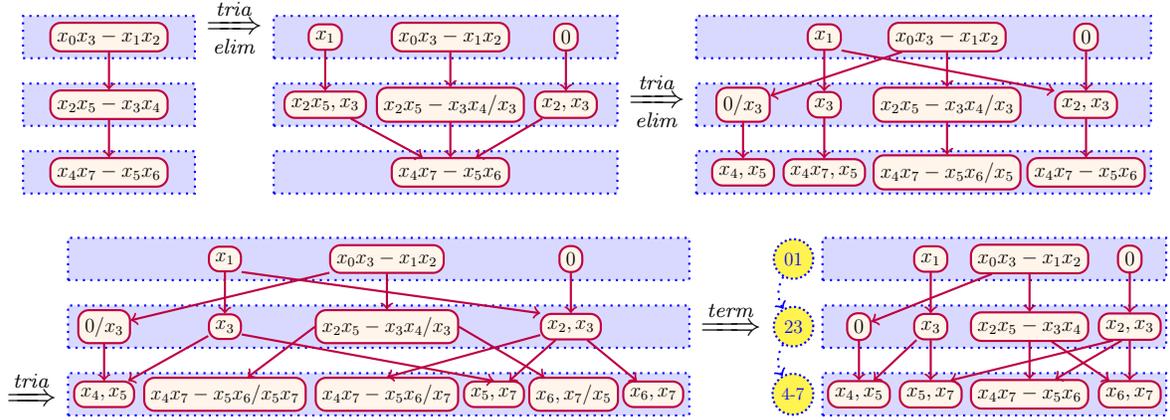

  \Cref{fig:adjminors4} illustrates the chordal triangularization algorithm for the polynomial set~$F$ from \Cref{exmp:triangadjminors4}.
  The nodes of the chordal network are polynomial systems $\mathfrak{F}=(F,H)$, which we represent in the figure as $(F/H)$.
  Note that in the termination step, after all triangulation/elimination operations, we remove the inequations to simplify the network.
  The final network has $8$ chains, which coincide with the triangular decomposition from \Cref{exmp:triangadjminors4}.
\end{example}

We can now compute chordal network representations of arbitrary systems.

\begin{theorem}\label{thm:chordtriangpos}
  Let $F\subset\K[X]$ be supported on a chordal graph $G$.
  With the above modifications, \Cref{alg:chordtriang} computes a $G$-chordal network $\mathcal{N}$, whose chains give a triangular decomposition of~$F$.
  Furthermore, this decomposition is squarefree if all triangulation operations are squarefree.
\end{theorem}
\begin{proof}
  See \Cref{s:positivedimproofs}.
\end{proof}

\begin{remark}
  We have noticed that chordal triangularization is quite efficient for binomial ideals.
  \Cref{remk:binomials} partly explains this observation.
  However, we do not yet know whether it will always run in polynomial time when the treewidth is bounded.
\end{remark}

\subsection{Computing with chordal networks}\label{s:usingtriangular}
We just showed how to compute chordal network representations of arbitrary polynomial systems.
We now explain how to extend the chordal network algorithms from \Cref{s:membership} and \Cref{s:usingtriangularmon} to the general case.

\subsubsection*{Elimination}
Since regular chains possess the same elimination property as lexicographic Gr\"obner bases, the approach from \Cref{s:elimination} works in the same way.

\subsubsection*{Radical ideal membership}
\Cref{alg:member} extends to the positive-dimensional case simply by replacing polynomial division with pseudo-division.
Note that we require a squarefree chordal network, which can be computed as explained in \Cref{thm:chordtriangpos}.

\subsubsection*{Dimension and equidimensional components}

The dimension of a regular chain $T$ is $n-\nobreak|T|$, which is the same as for the monomial case.
Thus, we can compute the dimension as in \Cref{s:usingtriangularmon}.
Similarly, we can compute a chordal network describing the highest dimensional component, and also isolate any given dimension of the network.

\begin{example}
  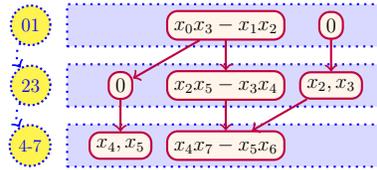
\begin{figure}[htb]
    \centering

    \begin{tikzpicture}[scale=1]
      \def\y{.8}
      \tikzstyle{myboxA} = [scale=0.65,draw=purple,fill=orange!8,thick,rectangle,rounded corners]
      \tikzstyle{myarrowA} = [purple,thick,->]
      \tikzstyle{myboxA0} = [fill=yellow!80,text=blue,draw=blue,thick,densely dotted,circle]
      \tikzstyle{myarrowA0} = [blue,thick,dotted,->,bend right]
      \tikzstyle{mygroup0} = [draw=blue,fill=blue!15, thick, dotted, minimum height=16]
      \def\x{1.4}
      \tikzstyle{mygroup} = [mygroup0, minimum width = 120]
      \node [myboxA] (zero0) at (0*\x,0*\y) {$x_0x_3-x_1x_2$};
      \node [myboxA] (zero2) at (1.0*\x,0*\y) {$0$};
      \node [myboxA](two0) at (0*\x,-1*\y) {$x_2x_5-x_3x_4$};
      \node [myboxA](two1) at (-1*\x,-1*\y) {$0$};
      \node [myboxA](two2) at (1*\x,-1*\y) {$x_2,x_3$};
      \node [myboxA](four0) at (0*\x,-2*\y) {$x_4x_7-x_5x_6$};
      \node [myboxA](four2) at (-1.0*\x,-2*\y) {$x_4,x_5$};
      \draw [myarrowA] (zero0) to (two0);
      \draw [myarrowA] (zero0) to (two1);
      \draw [myarrowA] (zero2) to (two2);
      \draw [myarrowA] (two0) to (four0);
      \draw [myarrowA] (two1) to (four2);
      \draw [myarrowA] (two2) to (four0);
      \begin{scope}[on background layer]
      \node [mygroup] (zero) at (0*\x,0*\y){}; 
      \node [mygroup] (two) at (0*\x,-1*\y){}; 
      \node [mygroup] (four) at (0*\x,-2*\y){}; 
      \end{scope}
      \node [myboxA0,scale=0.60,left=.2 of zero] (zero00) {01};
      \node [myboxA0,scale=0.60,left=.2 of two] (two00) {23};
      \node [myboxA0,scale=0.55,left=.2 of four] (four00) {4-7};
      \draw [myarrowA0] (zero00) to (two00);
      \draw [myarrowA0] (two00) to (four00);
    \end{tikzpicture}
    \caption{Top-dimensional part of the chordal network from \Cref{fig:adjminors4}.}
    \label{fig:adjminors4b}
  \end{figure}
  \Cref{fig:adjminors4b} shows the highest dimensional part of the chordal network from \Cref{fig:adjminors4}.
  This network only has $3$ chains, which give a Kalkbrener decomposition of the variety (see \Cref{exmp:triangadjminors4}).
  Likewise, the chordal network from \Cref{fig:triangadjminors} gives a Kalkbrener triangular decomposition of the ideal of adjacent minors of a $2\times 7$ matrix.
\end{example}

\subsubsection*{Irreducible components}
Unlike the monomial case, the chains of an arbitrary chordal network $\mathcal{N}$ will not necessarily define prime ideals (see \Cref{s:radicalirreducible}).
However, in some interesting cases it will be the case, thanks to the following well known property.

\begin{theorem}\label{thm:primeideal}
  Let $T = \{t_1,\ldots,t_k\}$ be a regular chain. 
  Assume that $\mdeg(t_i)=1$ for $1\leq i<k$ and that $t_k$ is an irreducible polynomial.
  Then $\sat(T)$ is a prime ideal.
\end{theorem}
\begin{proof}
  This follows from~\cite[Thm 6.2.14]{Wang2001}.
\end{proof}

%Given a chordal network $\mathcal{N}$, we can test in linear time whether all its chains $C$ have the form of the above theorem.
In particular, note that all chains of the chordal network from \Cref{fig:triangadjminors} are of this form.
We will see in \Cref{s:determinantal} that the same holds for other families of ideals.
Assume now that all chains of the network define a prime ideal.
A plausible strategy to compute all minimal primes (or only the high dimensional ones) is as follows:
\begin{enumerate}[leftmargin=.5in,label=(\roman*)]
  \item Iterate over all chains $T$ of the network in order of decreasing dimension.
  \item For a chain $C$, and a minimal prime~$I'$ previously found, determine whether $I'\subset \sat(C)$ by checking whether $\prem(f,C)=0$ for each generator~$f$ of~$I'$.
  \item If $I:=\sat(C)$ does not contain any previously found prime, compute generators for~$I$ by using Gr\"obner bases.
    We have a new minimal prime.
\end{enumerate}

\section{Examples}\label{s:applications}

We conclude this paper by exhibiting some examples of our methods.
We implemented our algorithms in \texttt{Sage}~\cite{sage} using Maple's library \texttt{Epsilon}~\cite{epsilon} for triangular decompositions, and \texttt{Singular}~\cite{singular} for Gr\"obner bases.
The experiments are performed on an i7 PC with 3.40GHz, 15.6 GB RAM, running Ubuntu 14.04. 

\subsection{Commuting birth and death ideal}\label{s:birthdeath}
We consider the binomial ideal $I^{n_1,\ldots,n_k}$ from~\cite{Evans2010aa}.
This ideal models a generalization of the one-dimensional birth-death Markov process to higher dimensional grids.
In~\cite{Evans2010aa} it is given a parametrization of its top-dimensional component, as well as the primary decomposition of some small cases.
In~\cite{Kahle2010} Kahle uses his Macaulay2 package \texttt{Binomials}, specialized in binomial ideals, to compute primary decompositions of larger examples.
We now show how our methods can greatly surpass Kahle's methods in computing the irreducible decomposition when the treewidth is small.

We focus on the case of a two dimensional grid:
\begin{align*}
  I^{n_1,n_2} =
  \langle &
  U_{i,j} R_{i,j+1} - R_{i,j} U_{i + 1 ,j},\;
  D_{i,j + 1} R_{i,j} - R_{i,j + 1} D_{i + 1 ,j + 1},\\
  &D_{i + 1 ,j + 1} L_{i + 1 ,j} - L_{i + 1 ,j + 1} D_{i,j + 1},\;
  U_{i + 1 ,j} L_{i + 1 ,j + 1} - L_{i + 1 ,j} U_{i,j}  
  \rangle_{\,0\leq i < n_1,\; 0\leq j<n_2}.
\end{align*}
We let the parameter $n_1$ take values between $1$ to $100$, while $n_2$ is either~$1$ or~$2$.
\Cref{tab:birthdeathchordtriang} shows the time used by \Cref{alg:chordtriang} for different values of $n_1,n_2$.
Observe that, for small values of $n_2$, our methods can handle very high values of $n_1$ thanks to our use of chordality.
For comparison, we note that even for the case $n_1=10, n_2=1$ \texttt{Singular}'s Gr\"obner basis algorithm (grevlex order) did not terminate within 20~hours of computation.
Similarly, neither \texttt{Epsilon}~\cite{epsilon} nor \texttt{RegularChains}~\cite{regularchains} were able to compute a triangular decomposition of $I^{10,1}$ within 20~hours. %not even 24 hours

\begin{table}[htbp]
  \centering
  \caption{Time required by chordal triangularization on ideals $I^{n_1,n_2}$. No other software we tried~\cite{singular,epsilon,regularchains,Kahle2010}
  %(\texttt{Singular}, \texttt{Epsilon}, \texttt{RegularChains}, \texttt{Binomials})
can solve these problems.}
    \begin{tabular}{c|ccccc}
    \toprule
    $n_1$ & 20    & 40    & 60    & 80    & 100 \\
    \midrule
    \multicolumn{1}{c|}{$n_2=1$} & 0:00:45 & 0:02:16 & 0:04:03 & 0:06:28 & 0:09:13 \\
    \multicolumn{1}{c|}{$n_2=2$} & 0:36:07 & 1:59:24 & 3:30:33 & 6:15:25 & 9:00:52 \\
    \bottomrule
    \end{tabular}
  \label{tab:birthdeathchordtriang}
\end{table}

We now consider the computation of the irreducible components of the ideal $I^{n_1,1}$.
We follow the strategy described after \Cref{thm:primeideal}, using \texttt{Sage}'s default algorithm to compute saturations.
\Cref{tab:birthdeathprimes} compares this strategy (including the time of \Cref{alg:chordtriang}) with the algorithm from \texttt{Binomials}~\cite{Kahle2010}.
It can be seen that our method is more efficient.
In particular, for the ideal $I^{7,1}$ Kahle's algorithm did not finish within 60~hours of computation.

\begin{table}[htbp]
  \centering
  \caption{Irreducible components of the ideals $I^{n_1,1}$.}
    \begin{tabular}{cc|ccccccc}
    \toprule
    \multicolumn{2}{c|}{$n_1$} & 1     & 2     & 3     & 4     & 5     & 6     & 7 \\
    \midrule
    \multicolumn{2}{c|}{$\#$components} & 3     & 11    & 40    & 139   & 466   & 1528  & 4953 \\
    \multirow{2}[0]{*}{time} & ChordalNet & 0:00:00 & 0:00:01 & 0:00:04 & 0:00:13 & 0:02:01 & 0:37:35 & 12:22:19 \\
          & Binomials & 0:00:00 & 0:00:00 & 0:00:01 & 0:00:12 & 0:03:00 & 4:15:36 & - \\
    \bottomrule
    \end{tabular}
  \label{tab:birthdeathprimes}
\end{table}

Comparing \Cref{tab:birthdeathchordtriang} with \Cref{tab:birthdeathprimes} it is apparent that computing a triangular chordal network representation is considerably simpler than computing the irreducible components.
Nonetheless, if we are only interested in the high dimensional components the complexity can be significantly improved.
Indeed, in \Cref{tab:birthdeathtopprimes} we can see how we can very efficiently compute all components of the seven highest dimensions.

\begin{table}[htbp]
  \centering
  \caption{High dimensional irreducible components of the ideals $I^{n_1,1}$.}
    \begin{tabular}{c|ccccc|cccc}
    \multicolumn{1}{c}{} & \multicolumn{5}{c}{Highest 5 dimensions} & \multicolumn{4}{c}{Highest 7 dimensions} \\
    \toprule
    $n_1$ & 20    & 40    & 60    & 80    & 100   & 10    & 20    & 30    & 40 \\
    \midrule
    \multicolumn{1}{c|}{$\#$comps} & 404   & 684   & 964   & 1244  & 1524  & 2442  & 5372  & 8702  & 12432 \\
    \multicolumn{1}{c|}{time} & 0:01:07 & 0:04:54 & 0:15:12 & 0:41:52 & 1:34:05 & 0:05:02 & 0:41:41 & 3:03:29 & 9:53:09 \\
    \bottomrule
    \end{tabular}
  \label{tab:birthdeathtopprimes}
\end{table}

\subsection{Lattice walks}
We now show a simple application of our radical membership test.
We consider the lattice reachability problem from~\cite{Diaconis1998} (see also~\cite{Kosaraju1982}).
Given a set of vectors $\mathcal{B}\subset \mathbb{Z}^n$, construct a graph with vertex set $\N^n$ in which $u,v\in \N^n$ are adjacent if $u-v \in \pm \mathcal{B}$.
The problem is to decide whether two vectors $s,t\in \N^n$ are in the same connected component of the graph.
This problem is equivalent to an ideal membership problem for certain binomial ideal $I_{\mathcal{B}}$~\cite{Diaconis1998}.
Therefore, our radical membership test can be used to prove that $s,t$ are not in the same connected component (but it may fail to prove the converse).
We consider the following sample problem.

  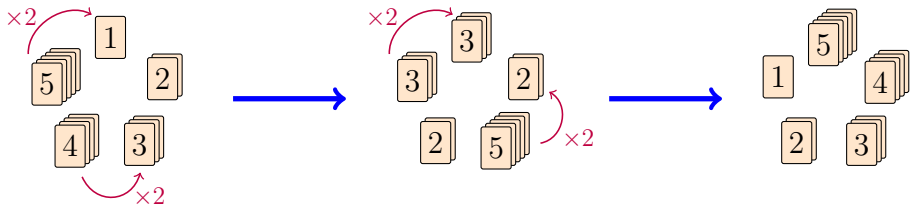
\begin{figure}[htb]
    \centering
\begin{tikzpicture}[scale=1,->]
    \def\x{0.4}
    \def\y{\x*1.4}
    \def\r{\x*1.9}
    \tikzstyle{mybox} = [fill=orange!20,rounded corners=2*\x]
    \tikzstyle{mylabel} = [scale=0.8]
    \tikzstyle{myarrow} = [purple,line width=.6]
    \tikzstyle{myarrow2} = [blue,line width=2]
    \def\mycard[#1,#2]{\draw[mybox] (#1,#2) rectangle +(\x,\y)}
    \def\mycardlab[#1,#2,#3]{\draw[mybox] (#1,#2) rectangle +(\x,\y) node[pos=.5] {#3}}
    \def\mystack[#1,#2,#3]{
      \pgfmathtruncatemacro\I{#1-1}
      \ifnum\I>0
      \foreach \i in {1,...,\I}{ 
        \pgfmathsetmacro{\j}{\i-.5*#1}
        \mycard[#2-\j*\x*.15,#3-\j*\x*.12]; }
      \fi
      \pgfmathsetmacro{\j}{#1-.5*#1}
      \mycardlab[#2-\j*\x*.15,#3-\j*\x*.12,#1]; 
    }

    \foreach \a in {1,2,...,5}{
      \pgfmathsetmacro{\CosValue}{\r*cos((\a-1)*360/5)}
      \pgfmathsetmacro{\SinValue}{\r*sin((\a-1)*360/5)}
      \mystack[\a,-4.8+ \SinValue ,\CosValue ];
      \node (p\a) at (-4.8+ \SinValue ,\CosValue) {};
    }

    \foreach \a/\b in {1/3,2/2,3/5,4/2,5/3}{
      \pgfmathsetmacro{\CosValue}{\r*cos((\a-1)*360/5)}
      \pgfmathsetmacro{\SinValue}{\r*sin((\a-1)*360/5)}
      \mystack[\b,\SinValue ,\CosValue ];
      \node (q\a) at (\SinValue ,\CosValue) {};
    }

    \foreach \a in {1,2,...,5}{
      \pgfmathsetmacro{\CosValue}{\r*cos((-\a)*360/5)}
      \pgfmathsetmacro{\SinValue}{\r*sin((-\a)*360/5)}
      \mystack[\a,4.8+ \SinValue ,\CosValue ];
      \node (t\a) at (4.8+ \SinValue ,\CosValue) {};
    }

    \draw [myarrow,bend left=60] ($(p5)+(-.5*\x,\y)$) to node[mylabel,left] {$\times 2\,$} ($(p1)+(-.2*\x,\y)$);
    \draw [myarrow,bend right=70,looseness=1.4] ($(p4)+(.6*\x,-.4*\y)$) to node[mylabel,right] {$\,\;\times 2$} ($(p3)+(.3*\x,-.3*\y)$);

    \draw [myarrow,bend left=60] ($(q5)+(-.5*\x,\y)$) to node[mylabel,left] {$\times 2\,$} ($(q1)+(-.2*\x,\y)$);
    \draw [myarrow,bend right=70,looseness=1.2] ($(q3)+(1.6*\x,.4*\y)$) to node[mylabel,below right] {$\!\!\times 2$} ($(q2)+(1.2*\x,.1*\y)$);

    \draw [->,myarrow2] (-3.0,.5*\x) -- (-1.5,.5*\x); 
    \draw [->,myarrow2] (2.0,.5*\x) -- (3.5,.5*\x); 

\end{tikzpicture}
    \vspace{-5pt}
    \caption{Illustration of the card problem using 5 decks.}
    \label{fig:stackcards}
  \end{figure}
\begin{problem}
  There are~$n$ card decks organized on a circle.
  Given any four consecutive decks we are allowed to move the cards as follows:
  we may take one card from each of the inner decks and place them in the outer decks (one in each),
  or we may take one card from the outer decks and place them on the inner decks.
  Initially the number of cards in the decks are $1,2,\ldots,n$.
  Is it possible to reach a state where the number of cards in the decks is reversed~\footnote{A combinatorial argument proves that this is only possible if all prime divisors of~$n$ are at least~$5$. However, this argument does not generalize to other choices of the final state (e.g., we cannot reach a state where the number of cards is $2,1,3,4,5,\ldots,n$ for any~$n$).}
  (i.e., the $i$-th deck has $n-i+1$ cards)?
\end{problem}

The above problem is equivalent to determining whether $f_n\in I_n$, where
\begin{align*}
  &f_n := x_0x_1^{2}x_2^3\cdots x_{n-1}^n- x_0^nx_1^{n-1}\cdots x_{n-1},
  &I_n := \{ x_{i}x_{i+3} - x_{i+1}x_{i+2}: 0\leq i< n\},
\end{align*}
and where the indices are taken modulo~$n$.
\Cref{tab:chameleon} compares our method against \texttt{Singular}'s Gr\"obner basis (grevlex order) and \texttt{Epsilon}'s triangular decomposition.
Even though the ideal $I_n$ is not radical, in all experiments performed we obtained the right answer.
%Even though we only test membership to the radical, in all experiments performed we obtained the right answer.
%(even when we tried other choices of~$f_n$).
Note that the complexity of our method is almost linear.
This contrasts with the exponential growth of both \texttt{Singular} and \texttt{Epsilon}, which did not terminate within $20$~hours for the cases $n=30$ and $n=45$.
We do not include timings for \texttt{Binomials} and \texttt{RegularChains} since they are both slower than \texttt{Singular} and \texttt{Epsilon}.

\begin{table}[htbp]
  \centering
  \caption{Time (seconds) to test (radical) ideal membership on the ideals $I_n$. }
    \begin{tabular}{c|ccccccccccc}
    \toprule
    $n$   & 5     & 10    & 15    & 20    & 25    & 30    & 35    & 40    & 45    & 50 & 55\\
    \midrule
    ChordalNet  & 0.7  & 3.0  & 8.5  & 14.3 & 21.8 & 29.8 & 37.7 & 48.2 & 62.3 & 70.6 & 84.8\\
    Singular  & 0.0  & 0.0  & 0.2  & 17.9 & 1036.2 & -     & -     & -     & -     & - & -\\
    Epsilon  & 0.1  & 0.2  & 0.4  & 2.0  & 54.4 & 160.1 & 5141.9 &17510.1 & -     & - & -\\
    Test result & true  & false & false & false & true  & false & true  & false & false & false & true\\
    \bottomrule
    \end{tabular}
  \label{tab:chameleon}
\end{table}

\subsection{Finite state diagram representation}\label{s:determinantal}
One of the first motivations in this paper was the very nice chordal network representation of the irreducible components of the ideal of adjacent minors of a $2\times n$ matrix.
We will see now that similar chordal network representations exist for other determinantal ideals.

First, notice that the chordal network in \Cref{fig:triangadjminors} has a simple pattern.
Indeed, there are three types of nodes $A_i = \{x_{2i}x_{2i+3}-x_{2i+1}x_{2i+2}\}$, $B_i = \{0\}$, $C_i = \{x_{2i},x_{2i+1}\}$, and we have some valid transitions: $A_i \to \{A_{i+1},B_{i+1}\}$, $B_i\to \{C_{i+1}\}$, $C_i\to \{A_{i+1},B_{i+1}\}$.
This transition pattern is represented in the state diagram shown in \Cref{fig:diagadjminors2}.
Following the convention from automata theory, we mark the initial states with an incoming arrow and the terminal states with a double line.

  \begin{figure}[htb]
    \centering
    \null\hfill
    \subfloat[$2\times n$ matrix]{

    \begin{tikzpicture}[framed,->, >=stealth', auto, semithick]
      %\tikzset{every edge/.style={ draw,semithick,purple }}
      \def\x{1.0}
      \def\y{1.2}
      \tikzstyle{state}=[scale=0.85,fill=orange!8,rounded corners,draw=purple,thick,text=black]
      \def\initstate[#1]{\draw[<-,blue](#1.west)--++(-25pt,0)}
      \node[state]    (A)  at (0*\x,1*\y) 
      {$\left|\begin{smallmatrix}
x_{2i}   & x_{2i+2}\\ 
x_{2i+1} & x_{2i+3}\\
\end{smallmatrix}\right|$};
      \node[state,double]    (B) at (-1*\x,0*\y)  {$0$};
      \node[state]    (C) at (1*\x,0*\y)  {$x_{2i},x_{2i+1}$};
      \initstate[A];
      \initstate[B];
      \path
      (A) edge[loop above,looseness=10]     node{}         (A)
          edge[bend right]     node{}         (B)
      (B) edge[bend right]     node{}         (C)
      (C) edge[bend right]     node{}         (A)
          edge[bend right]     node{}         (B);
    \end{tikzpicture}
      \label{fig:diagadjminors2}
    }
    \hfill
    \subfloat[$3\times n$ matrix]{

    \begin{tikzpicture}[framed,->, >=stealth', auto, semithick]
      \def\x{1.4}
      \def\y{1.0}
      \tikzstyle{state}=[scale=0.85,fill=orange!8,rounded corners,draw=purple,thick,text=black]
      \def\initstate[#1]{\draw[<-,blue](#1.west)--++(-25pt,0)}
      \node[state]    (A) at (0*\x,0*\y)                    
      {$\left|\begin{smallmatrix}
x_{3i}   & x_{3i+3} & x_{3i+6}\\ 
x_{3i+1} & x_{3i+4} & x_{3i+7}\\
x_{3i+2} & x_{3i+5} & x_{3i+8}
\end{smallmatrix}\right|$};
      \node[state]    (Z1) at (1*\x,1.3*\y)      {$0$};
      \node[state]    (Z2) at (1*\x,-1.3*\y)  {$0$};
      \node[state,double]    (Z3) at (2*\x,-1.5*\y)  {$0$};
      \node[state]    (B) at (3*\x,.7*\y)   
      {$\left|\begin{smallmatrix}
x_{3i}   & x_{3i+3}\\ 
x_{3i+2} & x_{3i+5}
\end{smallmatrix}\right|,
\left|\begin{smallmatrix}
x_{3i+1} & x_{3i+4}\\
x_{3i+2} & x_{3i+5}
\end{smallmatrix}\right|$};
      \node[state]    (C) at (3*\x,-.7*\y)   {$x_{3i},x_{3i+1},x_{3i+2}$};
      \initstate[A];
      \initstate[Z1];
      \initstate[Z2];
      \path
      (A) edge[loop above]     node{}         (A)
      (A.50)   edge[bend left]     node{}         (Z1.south west)
      (A.-50)  edge[bend right]     node{}         (Z2.north west)
      (Z1) edge[bend left]     node{}         (B)
      (Z2.east) edge[bend right]     node{}         (Z3.190)
      (Z3) edge[bend right]     node{}         (C)
      (B.west) edge[bend left]     node{}         (Z1)
      (B.-160) edge[bend right]     node{}         (Z2)
      (B.-170) edge[]     node{}         (A)
      (C.west) edge[bend right]     node{}         (Z2)
      (C.160)  edge[bend left]     node{}         (Z1)
      (C.170)  edge[]     node{}         (A);
    \end{tikzpicture}
      \label{fig:diagadjminors3}
    }
    \hfill\null
    \caption{State diagrams for ideals of adjacent minors of a matrix.}
  \end{figure}
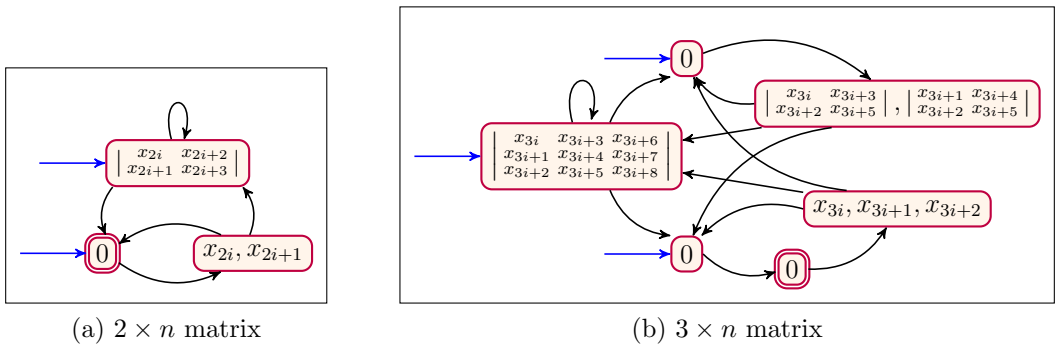

We can also consider the ideal of $3\times 3$ adjacent minors of a $3\times n$ matrix.
As seen in \Cref{fig:diagadjminors3}, a very similar pattern arises.
In order to make sense of such diagram let us think of how to generate a $3\times n$ matrix satisfying all these minor constraints.
Let $v_1,\ldots,v_n\in \overline{\K}^3$ denote the column vectors.
Given $v_{i+1},v_{i+2}$ we can generate $v_i$ as follows:
it can be the zero vector,
or it can be a multiple of $v_{i+1}$,
or it can be a linear combination of $v_{i+1},v_{i+2}$.
These three choices correspond to the three main states shown in the diagram.
Note now that if $v_i$ is the zero vector then we can ignore it when we generate $v_{i-1}$ and $v_{i-2}$.
This is why in order to reach the state $(x_{3i},x_{3i+1},x_{3i+2})$ we have two pass two trivial states.
Similarly, if $v_i$ is parallel to $v_{i+1}$ then we can ignore $v_{i+1}$ when we generate $v_{i-1}$.

It is easy to see that the above reasoning generalizes if we consider the adjacent minors of a $k\times n$ matrix.
Therefore, for any fixed $k$, the ideal of adjacent minors of a $k\times n$ minors has a finite state diagram representation (and thus it has a chordal network representation of linear size).
Since the nodes of the network are given by minors, then all chains of the network are of the form of \Cref{thm:primeideal}.
Thus, the decomposition obtained is into irreducible components.

  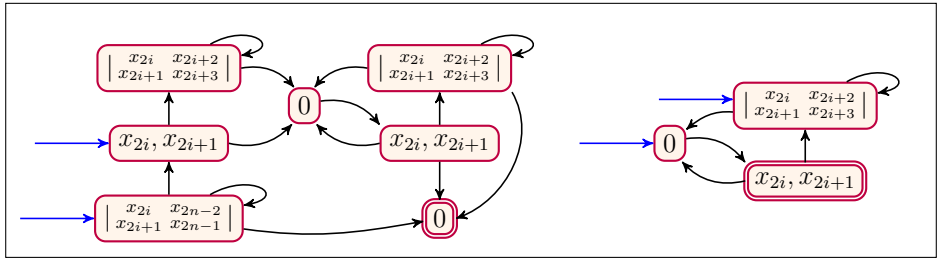
\begin{figure}[htb]
    \centering

    \begin{tikzpicture}[framed,->, >=stealth', semithick,every loop/.style={in=10,out=30,looseness=3.5}]
      \def\x{1.8}
      \def\y{.5}
      \def\xd{4.7*\x}
      \def\yd{-1*\y}
      \def\initlen{28pt}
      \tikzstyle{state}=[scale=0.85,fill=orange!8,rounded corners,draw=purple,thick,text=black]
      \def\initstate[#1]{\draw[<-,blue](#1.west)--++(-1*\initlen,0)}
      \node[state]    (A1) at (1*\x,-3*\y)                    
      {$\left|\begin{smallmatrix}
x_{2i}   & x_{2n-2}\\ 
x_{2i+1} & x_{2n-1}\\
\end{smallmatrix}\right|$};
      \node[state]    (A2) at (1*\x,1*\y)                    
      {$\left|\begin{smallmatrix}
x_{2i}   & x_{2i+2}\\ 
x_{2i+1} & x_{2i+3}\\
\end{smallmatrix}\right|$};
      \node[state]    (A3) at (3*\x,1*\y)                    
      {$\left|\begin{smallmatrix}
x_{2i}   & x_{2i+2}\\ 
x_{2i+1} & x_{2i+3}\\
\end{smallmatrix}\right|$};
      \node[state]    (Z1) at (2*\x,0*\y)      {$0$};
      \node[state]    (C1) at (1*\x,-1*\y)   {$x_{2i},x_{2i+1}$};
      \node[state]    (C2) at (3*\x,-1*\y)   {$x_{2i},x_{2i+1}$};
      \node[state,double]    (Z0) at (3*\x,-3*\y) {$0$};
      \initstate[A1];
      \initstate[C1];

      \node[state]    (A4) at (1*\x+\xd,1*\y+\yd)                    
      {$\left|\begin{smallmatrix}
x_{2i}   & x_{2i+2}\\ 
x_{2i+1} & x_{2i+3}\\
\end{smallmatrix}\right|$};
      \node[state]    (Z2) at (0*\x+\xd,0*\y+\yd)      {$0$};
      \node[state,double]    (C3) at (1*\x+\xd,-1*\y+\yd)   {$x_{2i},x_{2i+1}$};
      %\node[state,double]    (C4) at (0*\x+\xd,-3*\y+\yd)   {$x_{2n-2},x_{2n-1}$};
      %\initstate[A4];
      \draw[<-,blue] (A4.175) --++(-1*\initlen,0*\y);
      \initstate[Z2];

      \path
      %first connected component
      (A1) edge[loop]     node{}         (A1.north west)
           edge[]     node{}         (C1)
           edge[out=-8,in=-170]     node{}         (Z0)
      (C1) edge[]     node{}         (A2)
      (C1.east)     edge[bend right]     node{}         (Z1)
      (A2) edge[loop above]     node{}         (A2)
      (A2.east)     edge[bend left]     node{}         (Z1)
      (Z1) edge[bend left]     node{}         (C2.north west)
      (C2) edge[]     node{}         (A3)
           edge[]     node{}         (Z0)
      (C2.west)     edge[bend left]     node{}         (Z1)
      (A3) edge[loop above]     node{}         (A3)
      (A3.west)     edge[bend right]     node{}         (Z1)
      (A3.south east) edge[in=20,out=-60]     node{}         (Z0.east)
      %second connected component
      (Z2) edge[bend left]     node{}         (C3.north west)
           %edge[bend right]     node{}         (C4)
      (C3) edge[]     node{}         (A4)
      (C3.west)     edge[bend left]     node{}         (Z2)
      (A4) edge[loop above]     node{}         (A4)
      (A4.-175) edge[out=175,in=40]     node{}         (Z2);
    \end{tikzpicture}
    \caption{State diagram for the ideal of cyclically adjacent minors of a $2\times n$ matrix.}
    \label{fig:diagcycminors}
  \end{figure}

Many other families of ideals admit a simple state diagram representation.
For instance, the ideal generated by the $n$ cyclically adjacent minors of a $2\times n$ matrix (see \Cref{fig:diagcycminors}).
%An interesting property of this case is that it has two equidimensional components.
Interestingly, this chordal network has two equidimensional components.
Similarly, the ideal of (cyclically) adjacent permanental minors has a finite state diagram representation.
We can also easily provide families of zero-dimensional problems with such property (e.g., \Cref{fig:triangcycle}), since they often admit a chordal network of linear size (\Cref{thm:linearnetwork}).
A similar reasoning applies for monomial ideals.
It is natural to ask for further examples of this behaviour.

\begin{question}
  Characterize interesting families of ideals (parametrized by~$n$) whose triangular decomposition admits a finite state diagram representation (and thus have a chordal network representation of size~$O(n)$).
\end{question}

\begin{remark}
The class of binomial edge ideals~\cite{Herzog2010} is a natural starting point for this question, given that it generalizes both the ideal of adjacent minors (\Cref{fig:diagadjminors2}) and cyclically adjacent minors (\Cref{fig:diagcycminors}) of a $2\times n$ matrix.
\end{remark}

%\nocite {baz,fuzz,bong}
\bibliography{references}
\bibliographystyle{plain}

\appendix
\section{Additional proofs}\label{s:additionalproofs}

\subsection{Proofs from \Cref{s:zerodim}}\label{s:zerodimproofs}

\begin{proof}[Proof of \Cref{thm:maximalclique}]
  Let $m<n$ be such that $X_m$ is a maximal clique, and consider a rank~$m$ node $F_m\subset \K[X_m]$ to which we will apply a triangulation operation.
  Also let $F_m':=F\cap\K[X_m]$ be the unique initial node of rank~$m$.
  By assumption, $F_m'$ is zero-dimensional.
  Note that when we create a new node of rank~$m$ in an elimination operation, we copy the equations from a previous rank~$m$ node.
  In particular, we must have that $F_m'\subset F_m$, and therefore $F_m$ is also zero-dimensional.
  This proves the lemma for this case.

  Consider now some $p<n$ such that $X_p$ is not maximal, which means that $x_p$ is not a leaf of the elimination tree.
  Since $X_p$ is not maximal, there is a child $x_l$ of $x_p$ such that $X_l=X_p\cup\{x_l\}$.
  By induction, we may assume that the lemma holds for all nodes of rank~$l$.
  Consider a rank~$p$ node $F_p\subset \K[X_p]$ that we want to triangulate, and let $F_l$ of rank~$l$ be adjacent to $F_p$.
  Let $F_l'$ be the same rank~$l$ node, but before the $l$-th elimination round.
  By induction, $F_l'\subset\K[X_l]$ is zero-dimensional.
  Therefore, $\elim{p}{F_l'}\subset\K[X_l\setminus\{x_l\}]=\K[X_p]$ is also zero-dimensional, and as $\elim{p}{F_l'}\subset F_p$, we conclude that $F_p$ is zero-dimensional.
\end{proof}

\begin{lemma}\label{thm:radicalzerodim}
  Let $X_1,X_2\subset X$ and let $I_1\subset \K[X_1]$, $I_2\subset\K[X_2]$ be radical zero-dimensional ideals.
  Then $I_1+I_2\subset\K[X_1\cup X_2]$ is also radical and zero-dimensional.
\end{lemma}
\begin{proof}
  This is a direct consequence of the following known fact (see e.g.,~\cite[Thm~2.2]{sturmsolving}):
  an ideal $I\subset\K[X]$ is radical and zero-dimensional if and only if for any $x_i\in X$ there is a nonzero squarefree polynomial $f\in I\cap\K[x_i]$.
\end{proof}

\begin{proof}[Proof of \Cref{thm:radicalnetwork}]
  For any $l$, let $X^l$ denote the subtree of the elimination tree consisting of $x_l$ and all its descendants.
  For a chordal network $\mathcal{N}$, we will say that an $l$-subchain $C_l$ is the subset of a chain $C$ restricted to nodes with rank $i$ for some $x_i\in X^l$.
  Note that any chain is also a $(n-1)$-chain.
  Thus, it suffices to show that every $l$-subchain is radical after the $l$-th triangulation round in \Cref{alg:chordtriang}, and we proceed to show it by induction.

  If $x_l$ is a leaf in the elimination tree, then any $l$-subchain is just the output of a triangulation operation and thus it is radical.
  Assume that the result holds for all $l< p$.
  Let $T_p$ be a rank $p$ node obtained after the $p$-th triangulation round.
  Let $C$ be a $p$-subchain containing $T_p$; we want to show that $\ideal{C}$ is radical.
  Let $x_{l_1},\ldots,x_{l_k}$ be the children of $x_p$.
  For each $l_j$, let $C_{l_j}$ be the $l_j$-subchain obtained by restricting $C$ to ranks in $X^{l_j}$.
  Also let $C_{l_j}'$ be the same $l_j$-subchain, but before the $l_j$-th elimination round.
  Observe that
  \begin{align*}
    \ideal{C} = \ideal{T_p} + \sum_{j}\ideal{C_{l_j}'}.
  \end{align*}
  Note that $\ideal{T_p}$ is zero-dimensional and radical, and by induction the same holds for each $\ideal{ C_{l_j}'}$.
  It follows from \Cref{thm:radicalzerodim} that $\ideal{C}$ is radical.
\end{proof}

\begin{proof}[Proof of \Cref{thm:qdominatedconstant}]
  It was shown in~\cite{gao2009counting} that the complexity of Buchberger's algorithm is $q^{O(k)}$ if the equations $x_i^q-x_i$ are present, and the same analysis works for any $q$-dominated ideal.
  Given a Gr\"obner basis, the LexTriangular algorithm~\cite{Lazard1992} computes a triangular decomposition in time $D^{O(1)}$, where $D\leq q^k$ is the number of standard monomials.
  For irreducible (or squarefree) decompositions, we can reduce the problem to the univariate case by using a rational univariate representation~\cite{Rouillier1999} (here we need that $\K$ contains sufficiently many elements).
  This representation can also be obtained in $D^{O(1)}$.
  Since the complexity of univariate (squarefree) factorization~\cite{Kaltofen1992} is polynomial in the degree~($D$), the result follows.
\end{proof}

\begin{proof}[Proof of \Cref{thm:boundW2}]
  Let us see that the result holds after each triangulation and elimination round.
  We showed in \Cref{thm:boundW} that after the $l$-th triangulation round all rank~$l$ nodes have disjoint varieties, and thus there are at most $|\V(F\cap\K[X_l])|\leq q^\kappa$ of them.
  Consider now the $l$-th elimination round, and let us see that all the resulting rank~$p$ nodes ($x_p$ parent of $x_l$) also have disjoint varieties, and thus the same bound holds.

  Assume by induction that all rank~$p$ nodes have disjoint varieties before the $l$-th elimination round.
  Let $F_p$ be a rank~$p$ node (before the elimination) and let $T_1,\ldots,T_k$ be its adjacent rank~$l$ nodes.
  We just need to show that the new rank~$p$ nodes $F_p\cup \elim{p}{T_1}, \ldots, F_p\cup \elim{p}{T_k}$ have disjoint varieties (or are the same).
  By assumption, each $T_i\subset \K[X_l]$ defines a maximal (or prime) ideal, and thus $\elim{p}{T_i}\subset\K[X_l\setminus \{x_l\}]$ also defines a maximal ideal.
  Therefore, $\V(\elim{p}{T_i}),\V(\elim{p}{T_j})$ are either equal or disjoint, and it follows that the same holds for $\V(F_p\cup \elim{p}{T_i}), \V(F_p\cup \elim{p}{T_j})$.
\end{proof}

\subsection{Proofs from \Cref{s:membership}}\label{s:membershipproofs}

\begin{lemma}\label{thm:euclidean}
  Let $\mathbb{L}$ be a ring and let $f\in \mathbb{L}[y]$ be a monic univariate polynomial.
  Let $\phi:\mathbb{L}[y]\to \mathbb{L}[y]$ be an endomorphism such that $\phi(f)=f$ and $\deg(\phi(h))\leq \deg(h)$ for any $h\in \mathbb{L}[y]$.
  Then $\phi(h \bmod f) = \phi(h) \bmod f$, for any $h\in \mathbb{L}[y]$.
\end{lemma}
\begin{proof}
  Consider the Euclidean division $h = qf+ r$, where $q,r\in \mathbb{L}[y]$ and $\deg(r)<\deg(f)$.
  Then $\phi(h) = \phi(q)f+ \phi(r)$ and $\deg(\phi(r))\leq\deg(r)<\deg(f)$, so this is the Euclidean division of $\phi(h)$.
  It follows that $\phi(h\bmod f) =\phi(r)= \phi(h)\bmod f$.
\end{proof}

\begin{proof}[Proof of \Cref{thm:memberoutput}]
  We proceed by induction on $l$.
  The base case, $l=0$, is clear.
  Assume now that the lemma holds for some $l$, and let us prove it for $p:=l+1$.
  Let $f_p$ be a rank $p$ node and let $f_{l,1},f_{l,2},\ldots,f_{l,k}$ be its adjacent rank~$l$ nodes.
  Let us denote as $\phi_l$ the functional that plugs in the values $\hat{x}_0,\ldots,\hat{x}_l$.
  By induction, we know that 
  \begin{align*}
    H(f_{l,i}) = \phi_l(\sum_{C_{l,i}} r_{C_{l,i}} h \bmod C_{l,i})
  \end{align*}
  where the sum is over all $f_{l,i}$-subchains $C_{l,i}$.
  Note that the algorithm sets
  \begin{align*}
    H(f_p) = \phi_p(\sum_i r_{l,i} H(f_{l,i}) \bmod f_p),
  \end{align*}
  where $\phi_p$ is the functional that plugs in the value $\hat{x}_p$.
  Therefore,
  \begin{align*}
    H(f_p) 
    &= \phi_p(\sum_i r_{l,i} \phi_l( \sum_{C_{l,i}} r_{C_{l,i}} h \bmod C_{l,i})\bmod f_p)
    = \phi_p( \phi_l(\sum_i \sum_{C_{l,i}} r_{l,i}r_{C_{l,i}} h \bmod C_{l,i})\bmod f_p).
  \end{align*}
  Since any $f_p$-subchain is of the form $C_p = C_{l,i}\cup \{f_p\}$ for some $i$, we can rewrite 
  \begin{align*}
    H(f_p) 
    = \phi_p(\phi_l(\sum_{C_p} r_{C_p} h \bmod C_{p}')\bmod f_p),
  \end{align*}
  where the sum is over all $f_p$-subchains $C_p$, and where $C_p':= C_p \setminus \{f_p\}$.
  To complete the proof we just need to see that $\phi_l$ commutes with $\bmod f_p$.
  This follows from \Cref{thm:euclidean} by setting $y=x_p$ and $\mathbb{L}=\K[X\setminus\{x_p\}]$.
\end{proof}

\begin{proof}[Proof of \Cref{thm:memberoutput2}]
  Let $x_{m_i}$ denote the main variable of $h_i$, which is one of the ranks where the algorithm is initialized.
  It is enough to prove the lemma for ranks $l$ where the paths (in the elimination tree) starting from different $x_{m_i}$ first meet.
  Thus, we restrict ourselves to some $m_1,\ldots,m_k$ such that their respective paths all meet at rank~$l$.
  More precisely, we assume that $X^l_{m_i}\cap X^l_{m_j}=\{x_l\}$ for $i\neq j$, where $X^{l}_{m_i}$ denotes the path in the elimination tree connecting $x_{m_i}$ to $x_l$.

  By applying \Cref{thm:memberoutput} to each $h_i$, it follows that the final value of $H(f_l)$ is given by plugging in the values of $\hat{x}_1,\hat{x}_2,\ldots,\hat{x}_l$ in the polynomial
  \begin{align*}
    \sum_i \sum_{C_{i}} r_{C_{i}} h_i \bmod C_{i},
  \end{align*}
  where $C_{i}$ is an $f_l$-subchain restricted to the path $X^l_{m_i}$.
  Let $C = \bigcup_{i} C_{i}$ be the $f_l$-subchain obtained by combining them.
  We want to show that the above expression is equal to
  \begin{align*}
    \sum_{C} r_{C} (h_1+\cdots+h_k) \bmod C.
  \end{align*}
  Note now that $h_i$ does not involve any variable in $X^l_{m_j}$ for $j\neq i$.
  Thus, $ h_i \bmod C =  h_i \bmod C_i$.
  Observe that $X_{m_i}^l,X_{m_j}^l$ have no common arcs since they only meet at level $l$, and thus $r_{C} = \prod_i r_{C_{i}}$.
  Denoting $h_{C_{i}}:=h_i \bmod C_{i}$, the problem reduces to proving the following equality:
  \begin{align}\label{eq:identityC}
    \sum_i \sum_{C_{i}} r_{C_{i}} h_{C_{i}}
    = \sum_{C} r_{C_{1}} r_{C_{2}}\cdots r_{C_{k}} (h_{C_{1}}+h_{C_2}+\cdots+h_{C_k}).
  \end{align}

  In order to prove~\eqref{eq:identityC}, let us look at the right hand side as a polynomial in variables $h_{C_1},\ldots,h_{C_k}$.
  Note that the coefficient of $h_{C_1}$ in such polynomial is
  \begin{align*}
    \sum_{C\supset C_1} r_{C_1}r_{C_2}\cdots r_{C_k}
    = r_{C_1} \prod_{i=2}^k (\sum_{C_i} r_{C_i})
  \end{align*}
  and we want to show that this expression reduces to $r_{C_1}$.
  Recall that the scalar coefficients $r_{C_i}$ are normalized (this was the second modification made to \Cref{alg:member}).
  It follows that $\sum_{C_{i}} r_{C_{i}}= 1$ for all~$i$, and thus~\eqref{eq:identityC} holds.
\end{proof}
\subsection{Proofs from \Cref{s:positivedim}}\label{s:positivedimproofs}
\begin{proof}[Proof of \Cref{thm:chordtriangpos}]
  We have to show that:
  chordality is preserved,
  the variety is preserved,
  and the chains in the output are regular systems.
  The proofs of first two statements are essentially the same as for the chordally zero-dimensional case (\Cref{thm:preservechordal} and \Cref{thm:preservevariety}).
  It only remains to show that the chains of the output are regular systems.
  Proving that the chains are squarefree is very similar, so we skip it.

  Let $X^l$ denote the subtree of the elimination tree consisting of $x_l$ and its descendants.
  We say that an $l$-subchain is the subset of a chain given by nodes of rank $i$ for some $x_i\in X^l$.
  We will show by induction on~$l$ that after the $l$-th triangulation round every $l$-subchain is a regular system.
  The base case is clear.
  Assume that the result holds for all $l<p$. 
  Let $\mathfrak{T}_p$ be a rank $p$ node obtained after the $p$-th triangulation round.
  Let $\mathfrak{C}$ be a $p$-subchain containing $\mathfrak{T}_p$; we want to show that it is a regular system.
  It is easy to see that $\mathfrak{C}$ is triangular and that condition~\ref{line:regsys1} from \Cref{defn:regsys} is satisfied.
  We just need to check condition~\ref{line:regsys2}.

  Let $f\in \mathfrak{C}$ be a rank~$k$ polynomial; we want to show that $\init(f)(\hat{x}^{k+1})\neq 0$ for any $\hat{x}^{k+1}\in \Z(\elim{k+1}{\mathfrak{C}})$.
  First consider the case that $k\geq p$, which means that $f\in \mathfrak{T}_p$.
  The result follows from the fact that $\mathfrak{T}_p$ is a regular system. 
  Assume now that $k<p$, in which case there must be a child $x_l$ of $x_p$ such that $x_k\in X^l$.
  This means that $f$ belongs to an $l$-subchain $\mathfrak{C}_l$, which is a subset of $\mathfrak{C}$.
  Let $\mathfrak{C}_l'$ be the same $l$-subchain, but before the $l$-th elimination round.
  By induction, we know that $\mathfrak{C}_l'$ is a regular system, and thus $\init(f)(\hat{x}^{k+1})\neq 0$ for any $\hat{x}^{k+1}\in \Z(\elim{k+1}{\mathfrak{C}_l'}).$
  The result follows by noticing that $\Z(\elim{k+1}{\mathfrak{C}})\subset \Z(\elim{k+1}{\mathfrak{C}_l'})$.
\end{proof}

%\section*{Acknowledgments}
%We would like to acknowledge the assistance of volunteers in putting
%together this example manuscript and supplement.

\end{document}